\documentclass[11pt,english]{article}

\usepackage{verbatim}
\usepackage{float}
\usepackage{amsmath}
\usepackage{amsthm}
\usepackage{amssymb}
\usepackage{thmtools}
\usepackage{graphicx}
\usepackage{mathtools}
\usepackage{enumitem}
\usepackage{array}
\usepackage{xspace}

\usepackage[procnumbered,ruled,vlined,linesnumbered, algo2e]{algorithm2e}
\usepackage[normalem]{ulem}

\usepackage{tikz, relsize}
\usepackage{lmodern}
\usetikzlibrary{positioning, shapes, shadows, arrows}
\usetikzlibrary{decorations.pathreplacing}
\tikzstyle{abstract}=[circle, draw=black, fill=white]
\tikzstyle{labelnode}=[circle, draw=white, fill=white]
\tikzstyle{line} = [draw, -latex']
\tikzset{fontscale/.style = {font=\relsize{#1}}}

\usepackage[left=1in,top=1in,right=1in,bottom=1in]{geometry}

\makeatletter

\floatstyle{ruled}
\newfloat{algorithm}{tbp}{loa}
\providecommand{\algorithmname}{Algorithm}
\floatname{algorithm}{\protect\algorithmname}

\usepackage{xcolor}
\usepackage{nameref}
\definecolor{ForestGreen}{rgb}{0.1333,0.5451,0.1333}
\definecolor{DarkRed}{rgb}{0.8,0,0}
\definecolor{Red}{rgb}{1,0,0}
\usepackage[linktocpage=true,
pagebackref=true,colorlinks,
linkcolor=DarkRed,citecolor=ForestGreen,
bookmarks,bookmarksopen,bookmarksnumbered]
{hyperref}

\DontPrintSemicolon
\SetKw{KwAnd}{and}
\SetProcNameSty{textsc}
\SetFuncSty{textsc}

\usepackage{mathtools} %

\usepackage{cleveref}

\declaretheorem[numberwithin=section,refname={Theorem,Theorems},Refname={Theorem,Theorems},name={Theorem}]{thm}
\declaretheorem[numberlike=thm,refname={Theorem,Theorems},Refname={Theorem,Theorems},name={Theorem}]{theorem}
\declaretheorem[numberlike=thm,refname={Lemma,Lemmas},Refname={Lemma,Lemmas},name={Lemma}]{lem}
\declaretheorem[numberlike=thm,refname={Lemma,Lemmas},Refname={Lemma,Lemmas},name={Lemma}]{lemma}

\declaretheorem[numberlike=thm,refname={Corollary,Corollaries},Refname={Corollary,Corollaries},name={Corollary}]{corollary}

\declaretheorem[numberlike=thm,refname={Fact,Facts},Refname={Fact,Facts},name={Fact}]{fact}

\declaretheorem[numberlike=thm,refname={Proposition,Propositions},Refname={Proposition,Propositions},name={Proposition}]{prop}
\declaretheorem[numberlike=thm,refname={Observation,Observations},Refname={Observation,Observations}]{observation}

\declaretheorem[style=definition,numberlike=thm,refname={Definition,Definitions},Refname={Definition,Definitions},name={Definition}]{defn}
\declaretheorem[style=remark,numberwithin=section,refname={Remark,Remarks},Refname={Remark,Remarks},name={Remark}]{rem}
\declaretheorem[numberlike=thm,refname={Claim,Claims},Refname={Claim,Claims}]{claim}
\declaretheorem[numberlike=thm,refname={Invariant,Invariants},Refname={Invariant,Invariants}]{invariant}

\usepackage{footnote}
\makesavenoteenv{tabular}
\makesavenoteenv{table}

\makeatletter
\newcommand\footnoteref[1]{\protected@xdef\@thefnmark{\labelcref{#1}}\@footnotemark}
\makeatother

\makeatother

\usepackage[english]{babel}

\global\long\def\Otil{\tilde{O}}
\global\long\def\Ohat{\widehat{O}}
\global\long\def\poly{\mathrm{poly}}

\global\long\def\rev#1{#1^{(\mathrm{rev})}}

\global\long\def\ex{\mathrm{ex}}
\global\long\def\ab{\mathrm{ab}}
\global\long\def\val{\mathrm{val}}
\global\long\def\vol{\mathrm{vol}}
\global\long\def\dist{\mathrm{dist}}

\global\long\def\congest{\mathsf{cong}}
\global\long\def\len{\mathsf{len}}
\global\long\def\davg{d_{avg}}
\global\long\def\pset{\mathcal{P}}
\global\long\def\cP{\mathcal{P}}
\global\long\def\cset{\mathcal{C}}
\global\long\def\Shat{\hat{S}}

\global\long\def\cutorcert{\textsc{CutOrCertify}}
\global\long\def\Mto{\overrightarrow{M}}
\global\long\def\Mback{\overleftarrow{M}}
\global\long\def\Pto{\overrightarrow{P}}
\global\long\def\Pback{\overleftarrow{P}}

\global\long\def\skp{\mathrm{skip}}
\global\long\def\finite{\not\infty}
\global\long\def\infinite{\infty}
\global\long\def\Vinf{V^{\infty}}
\global\long\def\Rbar{\overline{R}}
\global\long\def\Lbar{\overline{L}}

\newcommand{\Omegahat}{\widehat{\Omega}}
\newcommand{\Omegatil}{\tilde{\Omega}}

\newcommand{\mustar}{\mu^*}
\newcommand{\estar}{E^*}
\newcommand{\rmatching}{\textrm{{\sc Robust-Matching}}}
\newcommand{\matchingtoosmall}{\textrm{{\sc Matching-Too-Small}}}
\newcommand{\gkappa}{G^{\kappa}}
\newcommand{\counter}{\textrm{{\sc Counter}}}
\newcommand{\matchingorcut}{\textrm{{\sc Matching-Or-Cut}}}

\newcommand{\rwitness}{\textrm{{\sc Robust-Witness}}}
\newcommand{\EmbedMatching}{\textrm{{\sc Embed-Matching}}\xspace}
\newcommand{\embedmatching}{\EmbedMatching}
\newcommand{\EmbedWitness}{\textrm{{\sc Embed-Witness}}}
\newcommand{\embedwitness}{\EmbedWitness}
\newcommand{\vertexmatching}{\textrm{{\sc Vertex-Congested-Matching}}\xspace}

\newcommand{\certifywitness}{\textrm{{\sc Certify-Witness}}}

\newcommand{\pathtowitness}{\textrm{{\sc Forest-From-Witness}}}
\newcommand{\shortoracle}{\textrm{{\sc Path-Inside-Expander}}}
\newcommand{\scchelper}{\textrm{{\sc SCC-Helper}}}
\global\long\def\termMatching{\textsc{TerminalMatching}}
\global\long\def\termWitness{\textsc{TerminalWitness}}

\newcommand{\cutplayer}{\mathcal{C}}
\newcommand{\psetlong}{\pset_{\textrm{long}}}

\newcommand{\alphawit}{\alpha_{\textrm{wit}}}
\newcommand{\epswit}{\epsilon_{\textrm{wit}}}
\newcommand{\alphacmg}{\alpha_{\textrm{cmg}}}
\newcommand{\alphaex}{\alpha_{\textrm{ex}}}

\newcommand{\edel}{{E^{\textrm{del}}}}

\newcommand{\Phidel}{{\Phi^{\textrm{del}}}}

\newcommand{\mstar}{M^*}

\newcommand{\Wstar}{W^*}
\newcommand{\pstar}{\pset^*}
\newcommand{\wu}{W_u}
\newcommand{\wset}{\mathbb{P}}
\newcommand{\mset}{\mathcal{M}}
\newcommand{\Pstar}{P^*}

\newcommand{\eps}{\epsilon}
\newcommand{\algcomment}[1]{\, $\backslash*$ \emph{#1} $\backslash*$}
\newcommand{\cmin}{C_{min}}
\newcommand{\ignore}[1]{}

\newcommand{\efull}{E^{\textrm{{\sc full}}}}
\newcommand{\mfull}{M^{\textrm{{\sc full}}}}
\newcommand{\mother}{M^{\textrm{{\sc other}}}}
\newcommand{\psetstar}{\pset^*}
\newcommand{\psetfull}{\pset^{\textrm{{\sc full}}}}
\newcommand{\psets}{\pset^{S}}
\newcommand{\psetcrit}{\pset^{\textrm{{\sc crit}}}}
\newcommand{\gstar}{G^*}
\newcommand{\vstar}{V^*}
\newcommand{\phistar}{\phi^*}

\newcommand{\fout}{\mathcal{F}_{out}}
\newcommand{\fin}{\mathcal{F}_{in}}

\newcommand{\ShowComment}{true}
\ifdefined\ShowComment

\newcommand{\todo}[1]{{\bf \color{red} TODO: #1}}
\newcommand{\sketch}[1]{{\color{violet} #1}}
\newcommand{\thatchaphol}[1]{\textcolor{blue}{TS: #1}}
\newcommand{\mpg}[1]{\textcolor{olive}{MPG: #1}}
\else
\newcommand{\todo}[1]{}
\newcommand{\sketch}[1]{}
\newcommand{\thatchaphol}[1]{}
\newcommand{\mpg}[1]{}
\fi

\date{}
\title{Deterministic Decremental Reachability, SCC, and Shortest Paths\\via Directed Expanders and Congestion Balancing}

\author{
	Aaron Bernstein \thanks{This work was done while funded by NSF Award 1942010 and the Simon's Group for Algorithms \& Geometry}\\
	Rutgers University \\
	bernstei@gmail.com
	\and
	Maximilian Probst Gutenberg \thanks{The author is supported by Basic Algorithms Research Copenhagen (BARC), supported by Thorup's Investigator Grant from the Villum Foundation under Grant No. 16582 and is supported by a start-up grant of Rasmus Kyng at ETH Zurich.}\\
	University of Copenhagen \\
	maximilian.probst@outlook.com
	\and
	Thatchaphol Saranurak\\
	Toyota Technological Institute at Chicago  \\
	saranurak@ttic.edu
}

\begin{document}

\maketitle

\pagenumbering{gobble}   
\begin{abstract}
Let $G = (V,E,w)$ be a weighted, directed graph subject to a sequence of adversarial edge deletions. In the decremental single-source reachability problem (SSR), we are given a fixed source $s$ and the goal is to maintain a data structure that can answer path-queries $s \rightarrowtail v$ for any $v \in V$. In the more general single-source shortest paths (SSSP) problem the goal is to return an approximate shortest path to $v$, and in the SCC problem the goal is to maintain strongly connected components of $G$ and to answer path queries within each component. All of these problems have been very actively studied over the past two decades, but all the fast algorithms are randomized and, more significantly, they can only answer path queries if they assume a weaker model: they assume an \emph{oblivious} adversary which is not adaptive and must fix the update sequence in advance. This assumption significantly limits the use of these data structures, most notably preventing them from being used as subroutines in static algorithms.

All the above problems are notoriously difficult in the adaptive setting. In fact, the state-of-the-art is still the Even and Shiloach tree, which dates back all the way to 1981 \cite{EvenS} and achieves total update time $O(mn)$. We present the first algorithms to break through this barrier:
\begin{itemize}
	\item \emph{deterministic} decremental SSR/SCC with total update time $mn^{2/3 + o(1)}$ 
	\item \emph{deterministic} decremental SSSP with total update time $n^{2+2/3+o(1)}$
\end{itemize}

To achieve these results, we develop two general techniques for working with dynamic graphs. The first generalizes expander-based tools to dynamic directed graphs. While these tools have already proven very successful in undirected graphs, the underlying expander decomposition they rely on does not exist in directed graphs. We thus need to develop an efficient framework for using expanders in directed graphs, as well as overcome several technical challenges in processing directed expanders. We establish several powerful primitives that we hope will pave the way for other expander-based algorithms in directed graphs.

The second technique, which we call \emph{congestion balancing}, provides a new method for maintaining flow under adversarial deletions. The results above use this technique to maintain an embedding of an expander. The technique is quite general, and to highlight its power, we use it to achieve the following additional result:
\begin{itemize}
	\item The first near-optimal algorithm for decremental bipartite matching
\end{itemize}

\end{abstract}

\newpage

\tableofcontents{}

\newpage

\pagenumbering{arabic}

\section{Introduction}

Let $G = (V,E,w)$ be a weighted, directed graph that is subject to dynamic updates that change the edges of $G$. We consider three closely related problems. In single-source reachability (SSR), we are given a fixed source $s$, and the goal is to maintain a data structure that can answer path queries $s \rightarrowtail v$ for any $v \in V$. The single-source shortest path problem (SSSP) is a generalization of SSR where the goal is return an approximate shortest path from $s$ to $v$. Finally, in dynamic strongly-connected components (SCC), the goal is to maintain a data structure such that given any two vertices $u,v \in V$, it can determine whether they are in the same SCC, i.e. whether $u$ and $v$ are on a common cycle in $G$, and if yes, can report a path between them in either direction.

All three of the above problems have received an enormous amount of attention in the dynamic setting. The most general model is the \emph{fully dynamic} one, where each adversarial update can either insert or delete an edge from $G$. But in this model there are very strong conditional lower bounds for all the above problems \cite{AbboudW14,HenzingerKNS15}.

For this reason, much of the work on these problems focuses on the weaker \emph{decremental} model, where the algorithm is given some input graph $G = (V,E,w)$, and the adversary deletes one edge at a time until the graph is empty. Here, results are typically expressed in terms of the \emph{total} update time over the entire sequence of deletions. Let $n$ be the number of vertices in the original input graph, $m$ the number of edges. The first algorithm for these problems is the Even and Shiloach tree \cite{EvenS} from 1981, which achieves total update time $O(mn)$ (amortized $O(n)$); See \cite{HenzingerK95} for a simple extension to directed graphs. A long line of work has since led to near-optimal algorithms for these problems in \emph{undirected} graphs, including some in the fully dynamic model \cite{Frederickson85,HenzingerK99,HolmLT01,Thorup00,PatrascuD04,NanongkaiS17,Wulff-Nilsen17,NanongkaiSW17,ChuzhoyGLNPS_det_cut}. The directed version is more difficult, but a series of results culminated in a near-optimal total update time $\Otil(m)$ for decremental SSR/SCC \cite{HenzingerKN14,HenzingerKN15,ChechikHILP16,ItalianoKLS17,BernsteinPW19} and moderate improvements for decremental SSSP: for example, total update time $\Otil(mn^{3/4})$ in \cite{GutenbergW20a} and an extremely recent $\Otil(n^2)$ result \cite{nearOptDenseSSSP}.

But all of the above $o(mn)$ algorithms for directed graphs suffer from a crucial drawback: they are randomized, and more significantly, they are only able to return paths if they assume an \emph{oblivious adversary}. Such an adversary cannot change its updates based on the algorithm's answers to path-queries: put otherwise, the adversary must fix its entire update sequence in advance. Much of the recent work in the field of dynamic graphs as a whole has focused on developing so-called adaptive algorithms that do not assume an oblivious adversary. This is important for two reasons. Firstly, adaptive algorithms work in a less restrictive model. Secondly, several recent papers have used dynamic graph algorithms as subroutines within the multiplicative-weight update method to speed up \emph{static} algorithms; for example, decremental shortest paths to speed up various (static) flow algorithms \cite{Madry10,ChuzhoyK19,ChuzhoyS20_apsp}, or incremental min-cut to speed up a TSP algorithm \cite{ChekuriQ17}. These applications to static algorithms all require \emph{adaptive} dynamic algorithms.

Despite all the progress for non-adaptive algorithms, the fastest adaptive algorithm for all the directed problems mentioned above remains the Even and Shiloach tree from 1981, which has total update time $O(mn)$. In this paper, we present the first algorithms to break through this barrier.

\begin{theorem}
\label{thm:main_scc} \label{thm:main_SCC}
Let $G$ be a directed graph. There exists an algorithm for decremental single-source reachability and decremental strongly connected components (SCC) with total update time $mn^{2/3+o(1)}$. The SCC algorithm not only explicitly maintains SCCs, but can answer path queries within an SCC. The algorithms can, respectively, determine whether a vertex $v$ is reachable from $s$, or whether two vertices are in the same SCC, in $O(1)$ time. The time to answer a path query is $P \cdot n^{o(1)}$, where $P$ is the length of the (simple) output path. 
\end{theorem}

\begin{theorem}
	\label{thm:main_sssp}
	Let $G$ be a directed graph with positive weights and let $W$ be the ratio of the largest to smallest weight. There exists an algorithm for decremental $(1+\eps)$-approximate single-source shortest paths with total update time $n^{2 + 2/3+o(1)}\textrm{{\normalfont log}}(W)/\epsilon$. (An update can delete an edge or increase an edge weight.) The query time is $O(1)$ for returning an approximate distance and $|P| \cdot n^{o(1)}$ for an approximate path, where $|P|$ is the length of the (simple) output path. 
\end{theorem}

\paragraph{Related Work}
Probst Gutenberg and Wulff-Nilsen considered a relaxed version of decremental SSSP that can only return \emph{distance estimates}, not an actual path. They showed an adaptive (randomized) algorithm for this problem with total update time  $\Otil(m^{2/3}n^{4/3}) = \Otil (n^{2+2/3})$ \cite{GutenbergW20a}. The adaptivity of this result crucially depends on the assumption that the adversary cannot see the paths used by the algorithm, so these results cannot be extended to the problems we are solving in this paper. Secondly, there are several results (both adaptive and oblivious) on dynamic SSC/SSSP in the \emph{incremental} setting, where the algorithm starts with an empty graph and edges are \emph{inserted} one at a time (see e.g. \cite{Haeupler12,Bender15,BernsteinC18,GutenbergWW20}). These incremental-only results use a very different set of techniques that do not transfer to the decremental setting. 

Directed expanders, key objects in this paper, are closely related to the notion of \emph{directed tree-width} introduced in \cite{Reed99,JohnsonRST01}, which is a key concept in deep structural statements, including the directed grid-minor theorem \cite{KawarabayashiK15,HatzelKK19} and the directed Erdos-Posa theorem \cite{ReedRST96,AmiriKKW16,MasarikMPRS19}.\footnote{In particular, directed expanders are graphs that contain a large \emph{well-linked set} \cite{ChekuriE15,ChekuriEP18} and directed tree-width of a graph is approximated, up to a constant, by the maximum size over all well-linked sets \cite{Reed99}.} The approximation algorithm for a variant of the disjoint paths problem by \cite{ChekuriE15} exploits the \emph{directed well-linked decomposition} which is related to directed expander decomposition stated in this paper. However, their technique is static and not concerned with time-efficiency beyond polynomial time.

\subsection{Techniques}

Our techniques are mostly very different from those of the earlier randomized algorithms, because those crucially relied on ``hiding" their choices from an oblivious adversary.
Our algorithms instead rely on expander-based tools. While these have previously been used to break long-standing barriers for adaptive algorithms in dynamic \emph{undirected} graphs \cite{NanongkaiS17,Wulff-Nilsen17,NanongkaiSW17,ChuzhoyK19,ChuzhoyS20_apsp}, our paper is to first to successfully apply them to dynamic algorithms for directed graphs. Our results require a large number of new techniques; we highlight the most significant ones below.

\paragraph{An efficient framework for directed expanders (Section \ref{sec:ingredient})}
Expander-based algorithms in undirected graphs rely on the following basic decomposition: given any graph $G = (V,E)$, it is possible to partition $E$ into sets $X$ and $R$, such that $X$ is the union of disconnected expanders, and $|R| \ll |E|$. The idea is then to use expander-tools on $X$ and deal with the small set $R$ separately. Unfortunately, such a guarantee is not possible for directed graphs: if $G$ is a dense DAG, then $R$ must contain all the edges of $G$. 

This paper explicitly shows the following decomposition for directed graphs: $E$ can be partitioned into three sets $X,D,R$ such that $X$ is the union of disconnected (directed) expanders, $D$ is acyclic, and $|R| \ll |X|$. (We actually use an analogous decomposition for vertex expanders.) We then use this decomposition as the crux of our new framework, which weaves together new fast algorithms for directed expanders with existing fast algorithms for DAGs. We hope that this framework will pave the way for future work that applies expander-tools to directed graphs. 

\paragraph{Congestion-balancing flow (Section \ref{sec:witness})}  
One of our main technical contributions is a new approach to maintaining a large flow in the presence of adversarial edge deletions (it is new to undirected graphs as well). Intuitively, a flow solution is more robust if it spreads out the congestion among all the edges of the graph. There are, however, two main challenges to formalizing this intuition. The first is that some edges may be more ``crucial" than others, so will necessarily have a higher congestion. The second is that these crucial edges might change over time, whereupon the flow must be rebalanced. We introduce a general approach for efficiently computing the ``right" congestion of each edge. We then show that a potential function based on minimum-cost flow allows us to cleanly analyze the total amount of rebalancing necessary.

In our decremental SSR/SCC/SSSP results, we use congestion-balancing flow to maintain an embedding of an expander. But the technique is quite general, and to highlight its power, we use it achieve significantly improved bounds for the seemingly unrelated problem of decremental bipartite matching (see below). 

\paragraph{New Primitives for Directed Expanders (Sections \ref{sec:pruning} and \ref{sec:CMG})}
Our new framework requires generalizing the essential expander primitives to directed graphs. While some of the primitives transfer almost automatically (e.g. unit flow), others pose significant technical challenge. We highlight two in particular.

In \emph{expander pruning} (Section \ref{sec:pruning}) we are given an expander $G = (V,E)$ subject to adversarial edge deletions. The goal of pruning is to dynamically maintain a set of pruned vertices $P \subseteq V$ such that the induced graph $G[V \setminus P]$ remains an expander. There are two known approaches to pruning in undirected graphs \cite{NanongkaiSW17,SaranurakW19}, but both break down in directed graphs because a sparse cut in one direction may not be sparse in the other. Our approach takes inspiration from \cite{NanongkaiSW17}, but requires a different key subroutine to work in directed graphs. In addition to generalizing the result of \cite{NanongkaiSW17}, our approach also ends up being simpler and cleaner.

The \emph{cut-matching game} (Section \ref{sec:CMG}) is the well-known tool for certifying expansion of graphs and was first introduced in \cite{KhandekarRV09}. There are two state-of-the-art variants: one is randomized but works in directed graphs \cite{Louis10}, while the second recent variant is deterministic but limited to undirected graphs \cite{ChuzhoyGLNPS_det_cut}. In this paper, we develop a cut-matching game that achieves the best of both worlds: it is deterministic and works in directed graphs. 
To do this, we generalize several of the lemmas in \cite{KhandekarKOV2007cut} to bound a more complex entropy-based potential function, and generalize the key subroutine for the cut player in \cite{ChuzhoyGLNPS_det_cut} to work directed graphs.

Both our pruning result and our new cut-matching game are stated as black-box results that can easily be plugged into other algorithms. Given how essential these tools have proven in undirected graphs, we think it is likely our contributions will prove useful for future work on directed expanders.

\subsection{An Additional Result: Decremental Bipartite Mathing}
 
As mentioned above, along the way to our main results we develop improved algorithms for dynamic matching. Consider the problem of maintaining a $(1-\eps)$-approximate maximum matching in an unweighted dynamic graph. In the fully dynamic setting, although there is a wide literature on faster update times for larger approximations, the best known update time for a $(1-\eps)$ approximation is $O(\sqrt{m})$ \cite{GuptaP13}, and there is evidence that $O(\sqrt{m})$ is a hard barrier to break through \cite{HenzingerKNS15,KopelowitzPP16}. For this reason, there has been a series of upper and lower bounds in the more relaxed incremental model, where the algorithm starts with an empty graph and edges are only inserted \cite{Dahlgaard16,BosekLSZ14,Gupta14,GrandoniLSSS19}. Most relevantly to our result, there is an incremental $(1-\eps)$-approximation with amortized $O(\log^2 n)$ update time in bipartite graphs \cite{Gupta14}, later improved to $O(1)$ update time in general graphs \cite{GrandoniLSSS19}. But the techniques of both papers are restricted to the incremental setting, and nothing analogous is known for decremental graphs; in fact, here $O(\sqrt{m})$ remained the best-known. 

We show that a simple application of our congestion-balancing flow technique yields a near-optimal algorithm for $(1-\eps)$-approximate matching in decremental \emph{bipartite} graphs; achieving a similar result for non-bipartite graphs remains an open problem. See Section \ref{sec:matching} for details.

\begin{thm}
	\label{thm:decremental-matching}
	Let $G$ be an unweighted bipartite graph. There exists a decremental algorithm with total update time $O(m\log^3(n)/\eps^4)$ (amortized $O(\log^3(n)/\eps^4)$) that maintains an integral matching $M$ of value at least $\mu(G)(1-\eps)$, where $G$ always refers to the current version of the graph. The algorithm is randomized, but works against an adaptive adversary; if we allow the algorithm to return a fractional matching instead of an integral one, then it is deterministic.
\end{thm}

\section{Preliminaries}
\label{sec:prelim} 

We usually refer to $n$ as the number of vertices
in a graph. We use $\Otil(\cdot)$ and $\Omegatil(\cdot)$ to hide
$\poly\log n$ factors in the big-oh notations. Similarly, we use
$\Ohat(\cdot)$ and $\Omegahat(\cdot)$ to hide $n^{o(1)}$ factors.

Graphs in this paper are directed. Given a graph $G$, the \emph{reverse
	graph} $\rev G$ of $G$ is obtained from $G$ by reversing the direction
of every edge in $G$. For any subset $S,T\subseteq V$, $E(S,T)$
is a set of directed edges $(u,v)$ where $u\in S$ and $v\in T$.
Let $G[S]$ denote the induced subgraph on $S$. Let $w:E\rightarrow\mathbb{R}$
be an edge weight function of $G$. Given $F\subseteq E$, let $w(F)=\sum_{e\in F}w(e)$ be the total weight of $F$; more generally, for any function $g$ on the edges $g(F) = \sum_{e \in E}g(e)$. The weighted in-degree and out-degree
of a vertex $u$ are $\deg^{in}(u)=w(E(V,u))$ and $\deg^{out}(u)=w(E(u,V))$,
respectively. The weighted degree of $u$ is $\deg(u)=\deg^{in}(u)+\deg^{out}(u)$.
The volume of a set $S$ is $\vol(S)=\sum_{u\in S}\deg(u)$. Several of our subroutines on expanders will use small fractional weights. 

For any $S$ with $\vol(S)\le\vol(V\setminus S)$ we refer to $(S,V\setminus S)$ as a \emph{cut} in $G$. Let $\delta^{out}(S)=w(E(S,V\setminus S))$
and $\delta^{in}(S)=w(E(V\setminus S,S))$ denote the total weight
of edges going out and coming in to $S$, respectively. We say that cut $(S,V\setminus S)$ is
\emph{$\epsilon$-balanced} if $\vol(S)\ge\epsilon\vol(V)$, and it
is $\phi$-sparse if $\min\{\delta^{in}(S),\delta^{out}(S)\}<\phi\vol(S)$.
We say that $(L,S,R)$ is a \emph{vertex-cut} of $G$ if $L$,$S$, and $R$ partition the vertex set $V$,
and either $E(L,R)=\emptyset$ or $E(R,L)=\emptyset$. 
Assuming that $|L|\le|R|$,
$(L,S,R)$ is \emph{$\epsilon$-vertex-balanced} if $|L|\ge\epsilon|V|$,
and it is \emph{$\phi$-vertex-sparse} if $|S|<\phi|L|$. %
We add the subscript $G$ to the notations whenever it is not clear
which graph we are referring to.

We say that a data structure supports \emph{SCC path-queries} in $G$,
if given vertices $u$ and $v$, it either correctly reports that
$u$ and $v$ are not strongly connected in $G$ in $O(1)$ time,
or returns a directed simple path $P_{uv}$ from $u$ to $v$
and a directed simple $P_{vu}$ from $v$ to $u$. We say that the
data structure has \emph{almost path-length query time} if, whenever
a path $P$ is returned, the data structure takes only $\Ohat(|P|)$
time to output the path. We emphasize that the returned path must
be simple.\footnote{Otherwise one can arbitrarily increase the length of the returned path through cycles and hence it can be trivial to achieve almost path-length query time.}

A \emph{decremental} graph $G$ is a graph undergoing a sequence of
deletions of edges and of isolated vertices. There is an easy reduction from
decremental SSR from source $s$ to decremental SCC: just add an edge from every $v \in V$ to $s$.

\section{High-level Overview}
\label{sec:overview}

We start with the definition of directed \emph{expanders} which are the central
object of this paper.
\begin{defn}
[Expanders]A directed graph $G$ is a \emph{$\phi$-vertex expander}
if it has no $\phi$-vertex-sparse vertex-cut. Similarly, $G$ is
\emph{$\phi$-(edge) expander} if it has no $\phi$-sparse cut.\footnote{Note that an isolated vertex is an expander (in both edge and vertex
versions).} 
\end{defn}

Intuitively, expanders are graphs that are ``robustly connected''
and, in particular, they are strongly connected. It is well-known
that many problems become much easier on expanders. So, given a problem
on general graph, we would like to reduce the problem to expanders. 

It turns out that every \emph{undirected} graph admits the following
\emph{expander decomposition}: for any $\phi>0$, a $\tilde{O}(\phi$)-fraction
of vertices/edges can be removed so that the remaining is a set of
vertex-disjoint $\phi$-vertex/edge expander. Unfortunately, this
is impossible in directed graphs. Consider, for example, a DAG. However,
a DAG is the only obstacle; for any $\phi>0$, we can remove a $\Otil(\phi)$-fraction
of vertices/edges, so that the remaining part can be partitioned into a DAG and a set of vertex-disjoint $\phi$-vertex/edge
expanders. This observation can be made precise as follows.\footnote{Although this decomposition is easy to prove by simply recursively
cutting a $\phi$-sparse cut, it was never explicitly stated before
to our best knowledge.}
\begin{fact}
[Directed Expander Decomposition]\label{prop:exp decomp}Let $G=(V,E)$
be any directed $n$-vertex graph and $\phi>0$ be a parameter. There
is a partition $\{R,X_{1},\dots,X_{k}\}$ of $V$ such that 
\begin{enumerate}
\item $|R|\le O(\phi n\log n)$;
\item $G[X_{i}]$ is a $\phi$-vertex expander for each $i$;
\item Let $D$ be obtained from $G$ by deleting $R$ and contracting each
$X_{i}$. Then, $D$ is a DAG.
\end{enumerate}
\end{fact}

The edge version of \Cref{prop:exp decomp} can be stated as follows:
for any unweighted $m$-edge graph $G=(V,E)$, there is a partition
$\{X_{1},\dots,X_{k}\}$ of $V$ and $R\subset E$ where $|R|\le O(\phi m\log m)$,
each $G[X_{i}]$ is a $\phi$-expander, and $D$ is a DAG (where $D$
is defined as above). It can be generalized to weighted graphs as
well.

This decomposition motivates the framework of our algorithm, although for the sake of efficiency we only maintain an approximate version (see Invariant \ref{inv:main} below.)
The decomposition suggests that we need four main ingredients: 
\begin{enumerate}
\item a dynamic expander decomposition in directed graphs, 
\item a fast algorithm on vertex-expanders,
\item a fast algorithm on DAGs, and
\item a way to
deal with the small remaining part $R$.
\end{enumerate}
Our algorithm
will run in time $\Ohat(m|R|)=\Ohat(mn^{2/3})$, as we choose
$\phi=n^{-1/3}$. Note that we do not work with edge-expanders because then
$R$ would have size $|R|=\Otil(\phi m)$, which is too big for us.
See \Cref{sec:ingredient} for how all components fit together. 

Here, let us focus on fast algorithms on expanders. One of our main
tasks is to certify that a given (sub)-graph $G$ is a vertex-expander.
This leads us to the notion of embedding: 
\begin{defn}
[Embedding and Embedded Graph]Let $G=(V,E)$ be a directed graph.
An \emph{embedding} $\pset$ in $G$ is a collection of simple directed
paths in $G$ where each path $P\in\pset$ has associated \emph{value}
$\val(P)>0$. We say that $\pset$ has \emph{length} $\len$ if every
path $P\in\pset$ contains at most $\len$ edges. We say that $\pset$
has \emph{vertex-congestion} $\congest$ if, for every vertex $v\in V$,
$\sum_{P\in\pset_{v}}\val(P)\le\congest$ where $\pset_{v}$ is the
set of paths in $\pset$ containing $v$. We say that $\pset$ has
\emph{edge-congestion }$\congest$ if, for every edge $e\in E$, $\sum_{P\in\pset_{e}}\val(P)\le\congest$
where $\pset_{e}$ is the set of paths in $\pset$ containing $e$.

Given an embedding $\pset$, there is a corresponding weighted directed
graph $W$ where, for each path $P\in\pset$ from $u$ to $v$, there
is a directed edge $(u,v)$ with weight $\val(P)$. We call $W$ an
\emph{embedded graph }corresponding to $\pset$ and say that $\pset$
embeds $W$ into $G$. 
\end{defn}

The following fact shows that, to certify that $G$ is a vertex expander,
it is enough to \emph{embed} an (edge)-expander $W$ into $G$ with small
congestion. 
\begin{fact}
\label{fact:certify vertex expansion}Let $G=(V,E)$ be a graph. Let
$W=(V,E',w)$ be a $\phi$-expander with minimum weighted degree $1$.
If $W$ can be embedded into $G$ with vertex congestion $\congest$,
then $G$ is a $(\phi/\congest)$-vertex expander. 
\end{fact}

\begin{proof}
Consider a vertex cut $(L,S,R)$ in $G$ where $|L| \le |R|$. 
Suppose that $E(L,R) = \emptyset$, otherwise $E(R,L) = \emptyset$ and the proof is symmetric.
Observe that each edge $e\in E_{W}(L,V\setminus L)$ in $W$ corresponds to a path in $G$ that goes
out of $L$ and, hence, must contain some vertex from $S$.
So the total weight of these edges in $W$ can be at most $\delta_{W}^{out}(L)\le|S|\cdot\congest$.
At the same time, $\delta_{W}^{out}(L)\ge\phi\vol_{W}(L)\ge\phi|L|$
as $W$ is a $\phi$-expander with minimum weighted degree 1. So $|S|\ge\frac{\phi}{\congest}|L|$
as desired.
\end{proof}
In our actual algorithm, instead of certifying that $G$ is a vertex
expander (i.e.~$G$ has no sparse vertex-cut), we relax to the task
to only certifying that $G$ has no balanced sparse vertex-cut. This
motivates the definition of $\phi$-witness which is used throughout
the paper:
\begin{defn}
[Witness]\label{def:witness}We say that $W$ is a\emph{ $\phi$-witness}
of $G$ if $V(W)\subseteq V(G)$, $W$ is a $\Omegahat(1)$-(edge)-expander
where $9/10$-fraction of vertices have weighted degree at least $1/2$,
and there is an embedding of $W$ into $G$ with vertex-congestion
$1/\phi$. (Note that $E(W)$ does not have to be a subset of $E(G)$.)
We say that $W$ is a $\phi$-short-witness if it is a $\phi$-witness
and the embedding has length $\Ohat(1/\phi)$. We say that $W$ is
a \emph{large} witness if $|V(W)|\ge9|V(G)|/10$.\footnote{The constant $9/10$ is somewhat arbitrary.}
\end{defn}

We sometimes informally refer to a graph that contains a large witness
as an \emph{almost vertex-expander}. This is because of the below
fact whose proof is similar to \Cref{fact:certify vertex expansion}.
\begin{fact}
Let $G=(V,E)$ be a graph that contains a large $\phi$-witness $W$.
Then $G$ has no $1/3$-vertex-balanced $(\phi/n^{o(1)})$-vertex-sparse
vertex cut. 
\end{fact}

Now, we have reduced the problem of certifying an almost vertex-expander
to maintaining a large witness. Although finding a low congestion
embedding in vertex expanders can be done very efficiently in the
static setting (using the well known cut-matching game), there is
one crucial obstacle in the dynamic setting. 

Consider the following simple scenario. We start with a complete graph
$G$ and parameter $\phi=\Omegahat(1)$. A standard (static) construction
of a large $\phi$-witness runs in $\Ohat(m)$ time and gives an \emph{unweighted} $\Omegahat(1)$-expander $W$ where all vertex degrees are $\Theta(\log(n))$.
Let $\pset$ be the embedding of $W$. Observe that each path from $\pset$ has
value $1$ and $|\pset|=O(n\log n)$. 

Unfortunately, once the adversary knows $\pset$, he can destroy each embedding path $P\in\pset$ by deleting any edge in $P$. In total,
he can delete only $O(n\log n)$ edges in $G$ to destroy the whole
embedding of $W$. The algorithm would then have to construct a new
witness, which the adversary could again destroy with $O(n\log n)$
deletions. This process continues until $G$ has a balanced, sparse
vertex-cut, which might not happen until $\Omega(n^{2})$ deletions.
That is, this standard approach requires the algorithm to re-embed
a new witness $\tilde{\Omega}(n)$ times, which is not only slow,
but requires too many changes to the witness. 

To overcome this obstable, we use the idea called \emph{congestion balancing} 
to maintain a witness $W$ that only needs to be re-embedded $\tilde{O}(1/\phi)$ times throughout
the entire sequence of deletions (formally stated in \Cref{thm:robust witness}). As a warm-up to
the proof of \Cref{thm:robust witness}, we show in \Cref{sec:matching} how to apply
this idea to the simpler bipartite matching problem.

\section{The Main Components}

\label{sec:ingredient}
In this section, we state all the algorithmic components formally and show how to combine them to prove \Cref{thm:main_scc}.
As we mentioned in \Cref{sec:overview}, our framework needs {\bf 1)} A dynamic expander decomposition {\bf 2)} a fast algorithm on vertex expanders, {\bf 3)} a fast algorithm on DAGs, and {\bf 4)} a way to deal with the small remaining part $\Shat$. 

It turns out that the existing algorithm of Lacki (unrelated to expanders) for separating out any small set of vertices \cite{Lacki11} is a handy tool for taking care of the DAG part and the small remaining part, and allows us to focus on almost vertex expanders. This algorithm has previously used in a similar way in \cite{ChechikHILP16}. We state the algorithm as a reduction below and defer the proof to \Cref{sec:lacki}.

\begin{restatable}[see \cite{Lacki11, ChechikHILP16}]{prop}{LackiProp}
	\label{thm:lacki}

Let $G=(V,E)$ be a decremental graph. Let $\mathcal{A}$ be a data structure that {\bf 1)} maintains a monotonically growing set $S \subseteq V$ and after every adversarial update reports any additions made to $S$ and {\bf 2)} maintains the SCCs in $G \setminus S$ explicitly in total update time $T(m,n)$ and supports SCC path queries in $G \setminus S$ in almost-path-length query time.

Then, there exists a data structure $\mathcal{B}$ that maintains the SCCs of $G$ explicitly and supports SCC path-queries in $G$ (in almost-path-length query time). The total update time is $O(T(m,n) + m|S|\log n)$, where $|S|$ refers to the final size of the set $S$. 
\end{restatable}

As we usually use $G$ to denote an input graph to each subroutine. We denote the input to the top-level algorithm by $\gstar = (\vstar,\estar)$. Motivated by the directed expander decomposition from \Cref{prop:exp decomp} and Lacki's reduction above, we  maintain the following invariant:

\begin{invariant}
	\label{inv:main}
Our decremental SCC algorithm will maintain an incremental set $\Shat$ such that $|\Shat| = \Ohat(n^{2/3})$ and at the end of processing any update, if the (non-singleton) SCCs of $G \setminus \Shat$ are $C_1, ..., C_k$, then each $C_i$ contains a large $\Omegahat(1/n^{1/3})$-short-witness. To ensure that $\Shat$ remains small, the algorithm will only add set $S$ to $\Shat$ if $S$ corresponds to some \emph{sparse} vertex cut $(L,S,R)$.
\end{invariant}

\paragraph{Robust Witness via Congestion-Balancing}
Let $G$ be some SCC in $\gstar \setminus \Shat$ at some point during the update sequence. To preserve Invariant \ref{inv:main}, we need a subroutine that maintains a large $\phi$-witness of $G$ where $\phi = \Omegahat(1/n^{1/3})$. If the subroutine fails to find such a witness, it returns a $\Omega(1/n^{o(1)})$-balanced, $\phi$-sparse vertex-cut $(L,S,R)$; that is, it certifies that $G$ is far from being a vertex expander, and must be further decomposed. (In particular, the top-level algorithm will add $S$ to the boundary set $\Shat$ and recurse on both $L$ and $R$.) Our new technique congestion-balancing flow will allow us to construct a \emph{robust} witness that is suitable to the dynamic setting; see Section \ref{sec:witness} for more details.

\begin{restatable}[Robust Witness Maintenance]{thm}{RobustWitnessTheorem}
	\label{thm:robust witness}There is a
	deterministic algorithm $\rwitness(G,\phi)$ that takes as input  a directed decremental $n$-vertex
	graph $G$ and a parameter $\phi \in (0,1/\log^2(n)]$. The algorithm maintains a large (weighted) $\phi$-short-witness $W$ of $G$ using $\Ohat(m/\phi^{2})$
	total update time such that every edge weight in $W$ is a positive multiple of $1/d$, for some number $d \leq 2\davg$, where $\davg$ is the initial average degree of $G$. The total edge weight in $W$ is $O(n\log n)$. After every edge deletion, 
	the algorithm either updates $W$ or outputs a $(\phi n^{o(1)})$-vertex-sparse
	$(1/n^{o(1)})$-vertex-balanced vertex-cut and terminates. 
	
	Let $W^{(i)}$ be $W$ after the $i$-th update. There exists a set
	$R$ of \emph{reset indices} where $|R|=\Ohat(\phi^{-1})$, such that
	for each $i\notin R$, $W^{(i)}\supseteq W^{(i+1)}$. That is, the
	algorithm has $\Ohat(\phi^{-1})$ phases such that, within each phase,
	$W$ is a decremental graph. The algorithm reports when each phase begins. It explicitly maintains the embedding $\pset$ of $W$ into $G$ and reports all changes made to $W$ and $\pset$.
\end{restatable}
The reason that $W$ only shrinks between each phase is as follows. 
Whenever the adversary deletes some edge $e$ in an embedded path $P$ that corresponds to an edge $e'$ in $W$, we will delete $e'$ from $W$. 
To guarantee that $W$ remains an expander after edge deletions, we run our new expander pruning algorithm in directed graphs (\Cref{thm:pruning}) on $W$ that further removes a small part from $W$ and guarantees that the remaining is still an expander. 
Nevertheless, after too many deletions, $W$ will be too small and we need to re-embed $W$.%

To highlight the strength of this result, the above theorem shows we only needs to re-embed a witness $\Ohat(\phi^{-1})$ times throughout the entire sequence of deletions, whereas the standard technique might require $\Omegatil(n)$ re-embeddings in the worst case as mentioned in \Cref{sec:overview}.

\paragraph{Maintaining Short Distances from a Witness}
Consider some SCC $G$ of $\gstar[\vstar\setminus \Shat]$ with a large $\phi$-witness $W$. We build two separate data structures on $G$. The first, given any vertex $u \in V(G) \setminus V(W)$, returns a path between $u$ and some $w \in V(W)$. The second can answer path queries for any $w_1,w_2 \in V(W)$.  It is easy to see that the two combined can answer SCC path-queries in $G$. The statement of the first data structure is a bit subtle; we give a formal theorem, followed by some intuition for what the theorem statement means. (See Section \ref{sec:forest} for the proof.)

\begin{restatable}{thm}{PathToWitnessTheorem}
\label{thm:ES from witness} \label{thm:path-to-witness} There is a data structure $\pathtowitness(G,W,\phi)$ that takes as input an $n$-vertex $m$-edge graph $G = (V,E)$, a set $W \subseteq V$ with $|W| \geq |V|/2$ and a parameter $\phi > 0$. The algorithm must process two kinds of updates. The first deletes any edge $e$ from $E$; the second removes a vertex from $W$ (but the vertex remains in $V$), while always obeying the promise that $|W| \geq |V|/2$. The data structure must maintain a forest of trees $\fout$ such that every tree $T \in \fout$ has the following properties: all edges of $T$ are in $E(G)$; $T$ is rooted at a vertex of $W$; every edge in $T$ is directed away from the root; and $T$ has depth $\Ohat(1/\phi)$. The data structure also maintains a forest $\fin$ with the same properties, except each edge in $T$ is directed towards the root. 
	
At any time, the data structure may perform the following operation:  it finds a $\Ohat(\phi)$-sparse vertex cut $(L,S,R)$ with $W \cap (L \cup S) = \emptyset$ and replace $G$ with $G[R]$. (This operation is NOT an adversarial update, but is rather the responsibility of the data structure.) The data structure maintains the invariant that every $v \in V$ is present in exactly one tree from $\fout$ and exactly one from $\fin$; given any $v$, the data structure can report the roots of these trees in $O(\log(n))$ time. (Note that as $V$ may shrink over time, this property only needs to hold for vertex $v$ in the \emph{current} set $V$.) The total time spent processing updates and performing sparse-cut operations is $\Ohat(m/\phi)$. 
\end{restatable}

Although the data structure works for \emph{any} set $W$, $W$ will always correspond to a $\phi$-witness in the higher-level algorithm. The adversarial update that removes a vertex from $W$ corresponds to the event that the witness shrinks in the higher-level algorithm. 
The forests $\fin$ and $\fout$ allow the algorithm to return paths of length $\Ohat(1/\phi)$ from any $v \in V(G)$ to/from $W$: find the tree that contains $v$ and follow the path to the root, which is always in $W$. The requirement that each tree has low-depth will be necessary to reduce the update time. But once we add this requirement, we encounter the issue that some vertices may be very far from $W$, so we need to give the data structure a way to remove them from $V(G)$. This is the role of the sparse-cut operation: we will show in the proof that if $v$ is far from $W$, it is always possible to find a sparse vertex cut $(L,S,R)$ such that $v$ is in $L$ and hence removed from $G$. (The higher-level algorithm will process this operation by adding $S$ to $\Shat$, so that $L$ becomes part of a different SCC in $\gstar[\vstar \setminus \Shat]$.)

\paragraph{Maintaining Paths Inside the Witness}

The second data structure shows how to maintain short paths between all pairs of vertices in an (edge) expander. The input $W$ will always correspond to a large $\phi$-witness, and will thus have expansion $1/n^{o(1)}$. This data structure is not new to our paper, as it is essentially identical to an analogous structure for undirected graphs in \cite{ChuzhoyS20_apsp}. The only major difference is that we need to plug in our new expander pruning algorithm for directed graphs (Theorem \ref{thm:pruning}). Note that the theorem below will only allow us to find paths in $E(W)$, not $E(G)$; we show later how to use the embedding of $W$ to convert them to paths in $E(G)$.

\begin{restatable}{thm}{OracleTheorem}
	\label{thm:short-path-oracle} 
	There is a deterministic data structure $\shortoracle(W)$ that takes as input an $n$-vertex $m$-edge $1/n^{o(1)}$-expander $W$ subject to decremental updates. Each update can delete an arbitrary batch of vertices and edges from $W$, but must obey the promise that the resulting graph remains a $\phi$-expander. Given any query $u,v \in V(W)$, the algorithm returns in $n^{o(1)}$ time a directed simple path $P_{uv}$ from $u$ to $v$ and a directed simple path $P_{vu}$
	of $v$ to $u$, both of length at most $n^{o(1)}$. The total update time of the data structure is $\Ohat(m)$.
\end{restatable}

\subsection*{The Algorithm}
\label{sec:scc}

The proof of Theorem \ref{thm:main_scc} combines all the above ingredients. See Algorithm \ref{alg:main} for pseudocode.

\begin{algorithm2e}[h]
	\label{alg:main}
	\caption{Maintaining an SCC-oracle for the main graph $\gstar$ (Theorem \ref{thm:main_scc})} 
	
	\SetAlgoSkip{}
	\SetKwProg{procedure}{Procedure}{}{}
	\SetKwFunction{sccHelper}{Setup for SCC-Helper}
	\SetKwFunction{sample}{Sample}
	\SetKwFunction{visit}{Visit}	
	
	 Initialize  $\Shat \gets \emptyset$, $\phistar \gets n^{-1/3}$, $\mathcal{C} \gets \{\vstar\}$ \tcp*[f]{$\mathcal{C}$ is the collection of SCCs in $\gstar \setminus \Shat$} \;  
	Initialize the framework of Proposition \ref{thm:lacki} \;
	Run $\scchelper(\gstar)$ \tcp*[f]{Will always run $\scchelper(C)$ for every SCC $C \in \mathcal{C}$} \;
	\procedure{Setup for $\scchelper(G)$}{
		Initialize $\rwitness(G,\phistar)$. Let $W$ be the large $\phistar$-witness maintained \label{line:scchelper-rwitness} \;
		
		Initialize $\shortoracle(W)$ \label{line:scchelper-shortoracle}  \;
		
		Initialize $\pathtowitness(G,W,\phistar)$ \label{line:scchelper-pathtowitness} \;
	} %
	
	\procedure{Updating the data structures in $\scchelper(G)$}{
		
		All adversarial edge deletions are fed to $\rwitness$ and $\pathtowitness$ \;
		
		\If(\label{line:scchelper-terminate}){$\rwitness$ in Line \ref{line:scchelper-rwitness} terminates with a cut $(L,S,R)$}{
			$\Shat \gets \Shat \cup S$; remove $V(G)$ from $\mathcal{C}$; add $L,R$ to $\mathcal{C}$ \;
			
			Initialize $\scchelper(G[L])$ and $\scchelper(G[R])$ \;
			
			{\bf Terminate} call $\scchelper(G)$ \tcp*[f]{$V(G)$ is decomposed into $L$ and $R$}\;	
		} %
	
		\If(\label{line:scchelper-new-phase}){$\rwitness$ in Line \ref{line:scchelper-rwitness} starts a new phase and hence creates a new $W$}{
		Initialize new data structures $\shortoracle(W)$ and $\pathtowitness(G,W,\phistar)$ and terminate existing ones in Lines \ref{line:scchelper-shortoracle} and \ref{line:scchelper-pathtowitness} \;
		} %
	
		\If{$\rwitness$ deletes vertices/edges from $W$ within a phase} {
			Feed these deletions as a batch deletion to $\shortoracle(W)$ \;
			
			{\bf if} vertex $v$ is deleted from $W$ {\bf then} feed to $\pathtowitness(G,W,\phistar)$ an update that removes $v$ from $W$ 
		} %
	
		\If(\label{line:scchelper-truncate}){$\pathtowitness$ returns a $\Ohat(\phistar)$-sparse vertex cut $(L,S,R)$ and replaces $G$ with $G[R]$} {
			$\Shat \gets \Shat \cup S$; add $L$ to $\mathcal{C}$; replace $G \in \mathcal{C}$ with $G[R]$ \tcp*[f]{$L$ is removed from SCC $G$}\;
			
			Initialize $\scchelper(G[L])$ \tcp*[f]{L is a new SCC in $\gstar[\vstar \setminus \Shat]$}}} %
\end{algorithm2e}

\paragraph{Analysis Sketch}
The full details of the analysis are left for Section \ref{sec:scc_analysis_long}. The argument has three main parts. The first is that each call $\scchelper(G)$ re-initializes data structure in Line \ref{line:scchelper-new-phase} only $\Ohat(1/\phistar)$ times, since that is the number of phases in $\rwitness$ (Theorem \ref{thm:robust witness}). The second is that every time a vertex $v$ participates in a new call $\scchelper(G)$, $|V(G)|$ must have decreased by a $(1-1/n^{o(1)})$ factor, so $v$ participates in $\Ohat(1)$ calls. The third is that we always have $|\Shat| = \Ohat(n\phistar) = \Ohat(n^{2/3})$, because vertices added to $\Shat$ always correspond to a $\phistar$-sparse cut.

The basic idea for the query is that given any $u,v$ in some SCC $C \in \mathcal{C}$ with Witness $W$, we use $\pathtowitness$ to find paths from $u$ and $v$ to $W$ and use $\shortoracle$ to complete the path inside $W$. The complication is that the resulting path $P$ might not be simple. We can always extract a simple path $P' \subseteq P$, but the query time would be proportional to $|P|$, not $|P'|$. We thus use a more clever query procedure; see Section \ref{sec:query} for details.

\paragraph{Comparison to Previous Work}
Our framework combines many old and new techniques, so we briefly categorize them. Proposition \ref{thm:lacki} and Theorem \ref{thm:path-to-witness} follow from ideas in two earlier papers \cite{Lacki11,ChechikHILP16} that are unrelated to expanders. Theorem \ref{thm:short-path-oracle} easily generalizes from an existing result for undirected graphs \cite{ChuzhoyS20_apsp}, but only once our new directed primitives are in place.

Our primary \emph{new contributions} are threefold: {\bf 1)} A new framework which integrates dynamic expander decomposition with earlier tools for directed graphs in \cite{Lacki11,ChechikHILP16}, {\bf 2)} Robust witness maintenance and congestion-balancing flow, and {\bf 3)} New primitives for directed expanders -- especially directed expander pruning (Theorem \ref{thm:pruning}) and cut-matching game (Theorem \ref{thm:CMG}) -- which are crucial for Theorems \ref{thm:robust witness} and \ref{thm:short-path-oracle} in this section.

\section{Maintaining a Witness via Congestion-Balancing Flow}
\label{sec:witness}

In this section, we present Algorithm $\rwitness$ from Theorem \ref{thm:robust witness}. The algorithm has several components, but the main innovation is a new approach we call congestion-balancing flow. To highlight this approach, we first show how it can be used to yield new results for the simpler problem of decremental bipartite matching (Theorem \ref{thm:decremental-matching}).

\subsection{Warmup: Decremental Bipartite Matching}
\label{sec:matching}

\paragraph{Informal Overview:} We focus on the following problem: say that we are given a bipartite graph $G_0 = (L_0 \cup R_0,E)$ with $|L_0| = n$ and $|E_0| = m$, and say that the graph has a perfect matching (i.e. $\mu(G_0) = n$). We assume that $n$ is a power of $2$. Let $\eps < 1$ be some fixed constant. Now, consider any adversarial sequence of edge deletions, and let $G$ always refer to the current version of the graph. The algorithm must maintain a \emph{fractional} matching in $G$ of size $\geq (1-5\eps)n$ OR certify that $\mu(G) \leq (1-\eps)n$, at which point it can terminate. In other words, the algorithm must maintain a matching until $\mu(G)$ decreases by a $(1-\eps)$ factor. The total update time should be $\Otil(m)$. This algorithm gets us most of the way to proving Theorem \ref{thm:decremental-matching}. (The conversion from fractional to integral matching is done via the black-box of Wajc \cite{Wajc19}.)

Consider the following lazy approach. Start by computing a matching $M$ of size $(1-\eps) n$ in $O(m)$ time (using e.g. Hopcroft-Karp \cite{hopcroft1973n}). The adversary must now delete $\Omega(\eps n)$ edges before $M$ has size $< (1-5\eps) n$, at which point we compute a new matching. This algorithm is too slow: we spend $O(m)$ time to compute a matching that survives for $O(n)$ deletions, for a total update time of $O(m^2 / n)$.

We would like to construct a robust matching that can survive for more than $\Omega(n)$ deletions. We will construct a fractional matching $M$ that attempts to put low value on each edge; this way, the adversary must delete many edges to remove $\eps n$ value from $M$. It may not be possible to put low value on \emph{all} edges, as some edges may be ``crucial" for any matching, but we present a technique for efficiently balancing the edge-congestion. We will then show that over the entire sequence of deletions there can only be a small number of crucial edges, so the adversary cannot profit too often from deleting them.

Our algorithm will run in phases. Each edge is given capacity $\kappa(e)$, which intuitively captures how crucial $e$ is. The algorithm initially sets $\kappa(e) = 1/n$, but $\kappa(e)$ can increase over time; these capacities transfer between phases. At the beginning of each phase, we first run Hopcroft-Karp to ensure that $\mu(G) \sim (1-\eps)n$; if not, we can terminate. So we can assume that we always have $\mu(G) \geq (1-2\eps) n$. We now try to compute a fractional matching $M$ such that $\val(M) \geq (1-4\eps)n$ and $\val(e) \leq \kappa(e) \ \forall e \in E$. If we find such an $M$, we use the lazy approach from before: we wait until the adversary deletes $\eps n$ value from $M$, and then we initiate a new phase. If the algorithm fails to find such an $M$, it instead returns a cut $C$ where the edge-capacities are too small. The algorithm then doubles $\kappa(e)$ for all $e \in C$, and again tries to compute a matching $M$. This process will eventually terminate because we know that $\mu(G) \geq (1-2\eps)n$; thus, once the edge-capacities are high enough, there will certainly be a matching $M$ with $\val(M) \geq (1-4\eps)n$. (Note that we never increase $\kappa(e)$ beyond 1, because a matching already has vertex capacity 1, so any edges with capacity $\geq 1$ effectively have infinite capacity.)

The crux of our algorithm is showing that the \emph{total} number of doubling steps, across all phases, is only $O(\log(n))$. Assuming this fact, let $K = \sum_{e \in G_0} \kappa(e)$. We will show that each doubling step only doubles $\kappa$ along a low-capacity cut, so $K$ only increases by $O(n)$. Since the number of doubling steps is $O(\log(n))$, we always have $K = O(n\log(n))$. This upper bound on $K$ in turn implies that there are only $O(\log(n))$ phases, because each phase must delete $\Omega(n)$ value from the matching $M$, which clearly involves deleting at least $\Omega(n)$ edge-capacity. 

To show that the number of doubling steps is $O(\log(n))$, we introduce the following potential function $\Pi(G,\kappa)$. Let the \emph{cost} of each $e$ be $c(e) = \log(n\kappa(e))$. Now, let $\mset$ be the set of all integral matchings $M$ (ignoring edge capacities) of size at least $(1-2\eps)n$; recall from above that we can assume $\mset \neq \emptyset$. Define $\Pi(G,\kappa)$ to be the minimum cost among all matchings from $\mset$. It is easy to see that each $\Pi(G,\kappa)$ is initially zero and is non-decreasing. Moreover, since every edge has $\kappa(e) \leq 1$ and $c(e) \leq \log(n)$, we also have $\Pi(G,\kappa) = O(n\log(n))$ at all times. We now argue (at a high level) that each doubling step increases $\Pi(G,\kappa)$ by $\Omega(n)$. Let $C$ be the cut that prevented the algorithm from finding a fractional matching $M$ with $\val(M) \geq (1-4\eps)n$. Any integral matching $M \in \mset$ has $\val(M) \geq (1-2\eps)n$, so it must have $\geq 2\eps n$ edges that cross $C$. Moreover, since the cut-capacity is small, $\Omega(\eps n)$ of these crossing edges must have capacity $<1$. The doubling step then doubles $\kappa(e)$ for each such edge, increasing each $c(e)$ by $1$, and thus increasing $c(M)$ by $\Omega(\eps n) = \Omega(n)$, as desired. 

\paragraph{Formal Description and Analysis:}
We now formally state our main subroutine for decremental matching; to avoid the assumption above that $\mu(G) = n$, the input parameter $\mu$ controls the target matching-size. Our  decremental matching result (Theorem \ref{thm:decremental-matching}) follows quite easily from the lemma below; see Section \ref{sec:matching_proof_long} for details. (The conversion from fractional to integral matching is done via a black box of Wajc \cite{Wajc19}.)

\begin{lemma}
	\label{lem:rmatching}
	Let $G_0 = (L_0 \cup R_0, E_0)$ be an unweighted bipartite graph subject to a sequence of adversarial edge deletions.  Given any parameters $\mu \in [1,n]$, $\eps \in (0,1)$, there exists an algorithm $\rmatching(G,\mu)$ which processes the deletions in total update time $O(m\log^2(n)/\eps^3)$ and has the following guarantees:
	\begin{enumerate}
		\item When the algorithm terminates, we have $\mu(G) \leq \mu(1-\eps)$.
		\item Until the algorithm terminates, it maintains a fractional matching $M$ with $\val(M) \geq \mu(1-5\eps).$
	\end{enumerate}
\end{lemma}

See Algorithm \ref{alg:rmatching} for pseudocode of $\rmatching$. The algorithm relies on the following static subroutine for finding a fractional matching of target size $\mu$ that obeys edge capacities $\kappa(e)$ (see \Cref{sec:flow_for_matching} for the proof).

\begin{restatable}{lemma}{MatchingOrCutLemma}
\label{lem:matching-or-cut}
There exists an algorithm $\matchingorcut(G,\kappa,\mu,\eps)$. The input is a graph $G = (L \cup R,E)$ with $|E| = m$ and $|L| = n$, a positive edge-capacity function $\kappa$, and parameters $\mu \in [1,n]$ and $\eps \in (0,1)$. In $O(m\log(n)/\eps)$ time the algorithm returns one of the following:
\begin{enumerate}
	\item A fractional matching $M$ of size $\mu(1-\eps)$ such that $\forall \ e \in E, \val(e) \leq \kappa(e)$.
	\item Sets $S_L \in L$ and $S_R \in R$ such that $\kappa(S_L,R \setminus S_R) + |S_R| \leq \mu + |S_L| - n$. 
\end{enumerate}
\end{restatable}

\begin{observation}
Case 2 of the above lemma certifies the non-existence of a large matching. In particular, \emph{any} matching with edge-capacities $\kappa$ can achieve value at most $\kappa(S_L,R \setminus S_R) + |S_R|$ from vertices in $S_L$, so the matching has value at most $(\kappa(S_L,R \setminus S_R) + |S_R|) + (n - |S_L|) \leq (\mu + |S_L| - n) + (n - |S_L|) = \mu$.	
\end{observation}

\begin{algorithm2e}[h]
	\label{alg:rmatching}
	\caption{Algorithm $\rmatching(G_0 = (L_0 \cup R_0,E_0),\mu,\eps)$}
	
	\SetAlgoSkip{}
	\SetKwProg{procedure}{Procedure}{}{}
	\SetKwBlock{RepeatUntilMatching}{Repeat until $\matchingorcut(G,\kappa,\mu(1-3\eps),\eps)$ {\normalfont returns a matching}}
	
	Assume that $|L_0| = n$ is a power of $2$ \tcp*[f]{otherwise replace $n$ with $n' = 2^{\lceil \log_2(n) \rceil}$} \label{line:d-matching}  \;
	
	Initialize $G = (L \cup R, E) \gets G_0$ \;
	
	Initialize $\kappa(e) = 1/n$ for every edge $e \in E_0$ \;
	
	\procedure(\tcp*[f]{execute before processing adversarial deletions}){Begin New Phase}{
	 {\bf if} $\matchingtoosmall(G,\mu,\eps)$ {\bf then Terminate} Algorithm 	\label{line:phase-begin-matching}\label{line:matching-too-small} \;
		
		\RepeatUntilMatching(\label{line:matching-or-cut}){
			
		 Let $S_L, S_R$ be the cut-sets returned by $\matchingorcut$ \label{line:return-cut-matching} \;
		
		 Let $\estar = \{ e \in E(S_L,R \setminus S_R) \mid \kappa(e) < 1 \}$ \algcomment{If $e \in \estar$ then $\kappa(e) \leq 1/2$} \label{line:estar-matching} \;
		
		 $\kappa(e) \gets 2\kappa(e)$ for all $e \in \estar$ \label{line:kappa-increase-matching}\;
	} %

		 Set $M$ to be the matching returned by $\matchingorcut$  \label{line:return-matching} \;
		 
		$\counter \gets 0$ \tcp*[f]{tracks value deleted from $M$ due to deletions in $G$} \;	
	} %
		
	\procedure{Processing Deletion of edge $(u,v)$}{
	
	Remove edge $(u,v)$ from $G$; if $(u,v) \in M$ then remove it from $M$ \;
	
	$\counter \gets \counter + \val(u,v)$ \;
	
	\If(\label{line:matching-counter-end}){$\counter \geq \eps \mu$}{
			RESET PHASE\label{line:phase-reset}: go back to Line \ref{line:phase-begin-matching} \tcp*[f]{capacities $\kappa$ NOT reset between phases}  \;
		} %
	}	%

	\procedure{$\matchingtoosmall(G,\mu,\eps)$}{
		Compute a $(1-\eps)$-approximate matching $M$ in $G$ in $O(m/\eps)$ time (using e.g. Hopkroft-Karp) \label{line:approx-matching}  \;
		
		{\bf if} $|M| < \mu(1-2\eps)$ {\bf then} return True; {\bf else} return False \;
	} %
\end{algorithm2e}

Now, we analyze Algorithm \ref{alg:rmatching}.

\begin{observation}
	\label{obs:kappa-matching}
	Throughout Algorithm \ref{alg:rmatching}, $\kappa$ is non-decreasing and in particular can only change via doubling in Line \ref{line:kappa-increase-matching}. Moreover, if $\kappa(e) < 1$ then $\kappa(e) \leq 1/2$, and we always have $\kappa(e) \leq 1 \ \forall e \in E(G)$ (here we use the assumption $n$ is a power of $2$; see Line \ref{line:d-matching}).
\end{observation}	

We now introduce our potential $\Pi(G=(V,E),\kappa)$, and state a few simple observations.

\begin{defn}[Min-cost Matching]
Recall that $\kappa(e) \geq 1/n \ \forall e \in E$. Let $\mset$ contain all integral matchings $M$ in $G$ for which $|M| \geq (1-2\eps)\mu$. Define the {\bf cost} of edge $e$ to be $c(e) = \log(n \kappa(e))$, and note that $c(e)$ is always non-negative. For any fractional matching $M$, define $c(M) = \sum_{e \in E} \val(e) c(e)$. Define $\Pi(G,\kappa) = \min_{M \in \mset} c(M)$; we refer to the matching $M$ that achieves this minimum as the min-cost matching. If $\mset = \emptyset$ then $\Pi(G,\kappa) = \infty$.
\end{defn}

\begin{observation} 
If $\kappa(e)$ increases for some edge $e$, then $\Pi(G,\kappa)$ cannot decrease as a result.
Similarly, an edge deletion cannot decrease $\Pi(G,\kappa)$. 
\end{observation}

\begin{observation}
	\label{obs:matching-potential-initial}
At the beginning of Algorithm \ref{alg:rmatching} we have $\Pi(G,\kappa) = 0$ (because for all edges $\kappa(e) = 1/n$, so $c(e) = 0$). Moreover, $\Pi(G,\kappa)$ only increases throughout the algorithm and if at any point $\Pi(G,\kappa) = \infty$, then it will remain infinite forever (this follows from the observations above, as well as the fact that $G$ is decremental).
\end{observation}

\begin{observation}
	\label{obs:matching-infinite-potential}
If $\Pi(G,\kappa) = \infty$ then invoking $\matchingtoosmall$ in Line \ref{line:matching-too-small} returns True and terminates the algorithm. 
\end{observation}

We have established that $\Pi$ start at $0$ and only increases. We now show that as long as the algorithm does not terminate, $\Pi$ is never too large.

\begin{lemma} 
\label{lem:matching-potential-upper} Consider any phase in which the algorithm did not terminate. Let $G$ be the graph and $\kappa$ the capacities at the end of initialization of this phase (Line \ref{line:return-matching}), but before any deletions have been processed. Then $\Pi(G,\kappa) = O(\mu\log(n))$.
\end{lemma}

\begin{proof}

Since Line \ref{line:matching-too-small} did not terminate, there must exist a matching $M$ in $G$ with $|M| \geq (1-2\eps) \mu$. Let $\mstar$ be an arbitrary subset of $M$ with $\lceil (1-2\eps) \mu \rceil$ edges. Note that $\mstar$ is a matching with $|\mstar| \leq \lceil \mu \rceil \leq 2\mu$. Every edge $e$ has $\kappa(e) \leq 1$ (Observation \ref{obs:kappa}), so $c(e) \leq \log(n)$, so $\Pi(G,\kappa) \leq c(\mstar) \leq |\mstar| \log(n) \leq 2\mu\log(n)$.
\end{proof}

\begin{defn}
Let $E_0$ be the edge set of the initial graph $G_0$.
We define $\kappa(E_0) = \sum_{e \in E_0} \kappa(e)$; If $e \in E_0$ is deleted by the adversary, then $\kappa(e)$ is the capacity of $e$ right before the deletion. 
\end{defn}

\begin{lemma}
\label{lem:matching-potential-main}
Consider some invocation of $\matchingorcut(G,\kappa,\mu(1-3\eps),\eps)$ in Line \ref{line:matching-or-cut} that returns cut-sets $S_L, S_R$. Let $\kappa$ be the capacities before the doubling step in Line \ref{line:kappa-increase-matching}, and $\kappa'$ the capacities after doubling. We then have:
\begin{enumerate}
	\item $k'(E_0) \leq k(E_0) + \mu$. AND
	\item $\Pi(G,\kappa') \geq \Pi(G,\kappa) + \eps \mu$.
\end{enumerate}
\end{lemma}

\begin{proof}
The first property is simple. By Lemma \ref{lem:matching-or-cut}, $\kappa(E(S_L,R\setminus S_r)) \leq \mu(1-3\eps)+ |S_L| - n \leq \mu$. Since $\estar \subseteq E(S_L,R\setminus S_r)$ (see Line \ref{line:estar-matching}) and the algorithm doubles all capacities in $\estar$, we have $\kappa'(E_0) - \kappa(E_0) = \kappa(\estar) \leq \mu$.

To prove the second property, note that since the algorithm did not terminate in Line \ref{line:matching-or-cut}, we must have $\mu(G) \geq (1-2\eps) \mu$. Now, let $M$ be \emph{any} matching in $G$ with $|M| \geq (1-2\eps) \mu$. We will show $|M \cap \estar| \geq \eps n$. 

Define $\efull = E(S_L,R\setminus S_R) \setminus \estar = \{e \in E(S_L,R\setminus S_R) \mid \kappa(e) = 1 \}$. Define $\mstar = M \cap \estar$, $\mfull = M \cap \efull$, $M^R = M \cap E(S_L, S_R)$, $\mother = M \cap E(L \setminus S_L, R)$. We know that $|\mstar| + |\mfull| + |M^R| + |\mother| = |M| \geq (1-2\eps)\mu$. On the other hand, we have $|\mother| \leq n - |S_L|$ and $|\mfull| + |M^R| \leq |\efull| + |S_R| \leq \mu(1-3\eps) + |S_L| - n$, where the last inequality follows from the guarantee of Lemma \ref{lem:matching-or-cut}. Combining the two inequalities above yields $|\mstar| \geq \eps \mu$, as desired.

Now, let $c$ and $c'$ be the corresponding cost functions $c(e) = \log(n\kappa(e))$ and $c'(e) = \log(n\kappa'(e))$. let $M'$ be the min-cost matching that minimizes $c'(M)$ for potential function $\Pi(G,\kappa')$. By the above argument $|M' \cap \estar| \geq \eps \mu$. For each edge $e \in M' \cap \estar$ we have $\kappa'(e) = 2\kappa(e)$, so $c'(e) = c(e) + 1$. Thus, $\Pi(G, \kappa') = c(M') = c(M') + |M' \cap \estar| \geq c(M') + \eps \mu \geq \Pi(G,\kappa) + \eps\mu$, as desired.
\end{proof}

\begin{corollary} 
\label{cor:matching-numtimes}
In any Execution of Algorithm \ref{alg:rwitness}, the total number of times that $\matchingorcut$ in Line \ref{line:matching-or-cut} returns a cut is $O(\log(n)/\eps)$. Moreover, we always have $\kappa(E_0) = O(\mu \log(n)/\eps)$.
\end{corollary}

\begin{proof}
First we argue that whenever $\matchingorcut(G,\kappa,...)$ is called, $\Pi(G,\kappa)$ is finite. Note that $\kappa$ only affects the magnitude of $\Pi(G,\kappa)$, not whether it is finite or infinite. Thus, if $\Pi(G,\kappa)$ is finite the first time $\matchingorcut$ is called in a phase, it will be finite every time $\matchingorcut$ is called in that phase. We begin every phase with a call to  $\matchingtoosmall$ (Line \ref{line:phase-begin-witness}), and by Observation \ref{obs:matching-infinite-potential}, if $\Pi(G,\kappa)$ were infinite, then the algorithm would terminate.
	
By Lemma \ref{lem:matching-potential-main}, every time $\matchingorcut$ returns a cut, $\Pi$ increases by at least $\eps \mu$. This completes the proof of the first statement, when combined with the fact that the potential starts at 0 and never decreases (Observation \ref{obs:matching-potential-initial}), and that if finite the potential is always $O(\mu\log(n))$ (Lemma \ref{lem:matching-potential-upper}). The bound on $\kappa(E_0)$ then follows from Property 1 of Lemma \ref{lem:matching-potential-main}. 
\end{proof}

\begin{lemma} 
\label{lem:matching-num-phases}
The total number of phases in any execution of Algorithm \ref {alg:rwitness} is at most $O(\log(n)/\eps^2)$
\end{lemma}

\begin{proof}
Let $\Phidel = \sum_{e \in \edel} \kappa(e)$, where $\edel$ contains all the edges deleted by the adversary so far (among all phases).
Consider any phase that does not terminate the algorithm in Line \ref{line:matching-too-small}. 
By Line \ref{line:matching-counter-end}, the phase can only end when the adversary deletes at least $\eps \mu$ value from the matching for that phase; since every edge obeys $\val(e) \leq \kappa(e)$, this implies that over the course of the phase, $\Phidel$ increases by at least $\eps \mu$. By Corollary \ref{cor:matching-numtimes}, we always have $\Phidel = \kappa(\edel) \leq \kappa(E_0) = O(\mu\log(n)/\eps)$. Thus, the number of phases is $O([\mu\log(n)/\eps]/[\eps \mu]) = O(\log(n)/\eps^2)$.
\end{proof}

\begin{proof}[Proof of Lemma \ref{lem:rmatching}]
We are now ready to prove that algorithm $\rmatching$ (Algorithm \ref{alg:rmatching}) satisfies the requirements of Lemma \ref{lem:rmatching}. The algorithm can only terminate if the $(1-\eps)$-approximate matching $M$ in Line \ref{line:approx-matching} has size $|M| < (1-2\eps)\mu$. But this implies that $\mu(G) \leq |M|/(1-\eps) < \mu (1-2\eps)/(1-\eps) < \mu(1-\eps)$, as needed in Case 1 of Lemma \ref{lem:rmatching}. 

For case 2, consider any phase of the algorithm. At the end of initialization for that phase (Line \ref{line:return-matching}), but before any deletions are processed, Lemma \ref{lem:matching-or-cut} guarantees that the matching $M$ returned by $\matchingorcut(G,\mu(1-3\eps),\eps)$ has $\val(M) \geq (1-3\eps)(1-\eps) \mu \geq (1-4\eps)\mu$. By Line \ref{line:matching-counter-end}, the phase ends after the adversary deletes more than $\eps \mu$ value from the matching. Thus, throughout the phase we have $\val(M) \geq (1-5\eps)\mu$, as desired. 

We now bound the running time. Each phase is dominated by the run-time of $\matchingorcut$ (Line \ref{line:matching-or-cut}), which is $O(m\log(n)/\eps)$. This subroutine might be run multiple times per phase, all but one of which return a cut. The total time is thus $O((m\log(n)/\eps) \cdot ($[\# phases] + [\# invocation of \matchingorcut\ that return a cut]$))$. By Lemma \ref{lem:matching-num-phases} and Corollary \ref{cor:matching-numtimes}, the run-time is $O((m\log(n)/\eps)(\log(n)/\eps^2 + \log(n)/\eps)) = O(m\log^2(n)/\eps^3)$.
\end{proof}

\subsection{Overview of Algorithm \rwitness}

The algorithm for maintaining a witness follows the same congestion-balancing approach as the decremental matching algorithm, but the details are significantly more involved. 

The algorithm will again run in phases. Just as algorithm $\rmatching$ began each phase by checking that the graph contains a large matching, now the algorithm checks that the graph contains a very large $\phi$-witness; if not, the algorithm is able to find a sparse, balanced cut and terminate. From now on we assume such a witness exists.

As described in \Cref{sec:overview}, an arbitrary embedding $\pset$ might not be robust to adversarial deletions, because a small number of edges might have most of the flow. To balance the edge-congestion, we introduce a capacity $\kappa(e)$ on each edge. Initially we set $\kappa(e) = 1/d$, where $d$ is the average degree in the input graph.
At each step, the algorithms uses approximate flows and the cut-matching game to try to find a witness with vertex congestion $\Otil(1/\phi)$ and edge-congestions $\kappa(e)$. If it fails, the subroutine finds a low-capacity cut $C$; it then doubles capacities in $C$ and tries again. Since we assume a witness does exist, the algorithm will eventually find a witness once the edge-capacities are high enough.

Once we have a witness $W$ with embedding $\pset$, we use the lazy approach. Say the adversary deletes an edge $(u,v)$. Because our embedding obeyed capacity constraints, this can remove at most edges from $W$ of total weight at most $\kappa(u,v)$. To maintain expansion, we feed these deletions into our expander pruning algorithm (Theorem \ref{thm:pruning}) to yield a pruned set $P$, and shrink our witness to $W[V(W) - P]$. To guarantee that $W$ remains a large witness, we end the phase once the pruned set $P$ it too large. We will show that we end a phase only after the adversary deletes $\Omegahat(n)$ edge-capacity from the graph.

As with $\rmatching$, the crux of our analysis will be to show that the total of number of doubling steps is $\Ohat(1/\phi)$. To do so, we again use costs $c(e) = \log(d\kappa(e))$ and use a potential function $\Pi(G,\kappa)$ which measures the \emph{min-cost embedding} in $G$ among all \emph{very large} $\phi$-witness. As the vertex congestion is $1/\phi$, this potential $\Pi(G,\kappa)$ is at most $n/\phi$. Also, we are able to show that each doubling step increases the potential by $\Omegahat(n)$ using an argument that is more involved than the one for matching. Therefore, there are at most $\Ohat(1/\phi)$ doubling steps as desired. 

Given this bound, we can bound the total number of phases: each doubling step adds at most $n$ to the total capacity $\kappa$, and the initial capacity is at most $1/d \cdot m = n$. So the final total capacity is at most $\Ohat(n/\phi)$. As each phase must delete $\Omegahat(n)$ capacity, there are at most $\Ohat(1/\phi)$ phases. 

\subsection{Subroutines Used by Algorithm \rwitness}

The rest of this section is devoted to the formal proof of Theorem \ref{thm:robust witness}. For convenience, we restate the theorem below

\RobustWitnessTheorem*

Just as in $\rmatching$ we began each phase by making sure that the matching was still large enough (Line \ref{line:matching-too-small}), so in $\rwitness$ we begin each phase by running $\certifywitness$ (Line \ref{line:phase-begin-witness}) to ensure that the graph is still close enough to a vertex expander. Formally, we certify that there exists a \emph{very large} $\phi$-witness $W$ that can be embedded into $G$. Note that we will never actually use this witness; we only need to ensure that it exists, as this will allow us to bound the running time of the algorithm. If such a witness does not exist, we return a balanced, sparse vertex-cut and terminate the entire algorithm.

We start with a subrotuine $\vertexmatching$ that is given two vertex sets $A,B$ and uses approximate flow to embed a single matching between them with small vertex-congestion, or returns a balanced, sparse vertex-cut. We then show how to use this subroutine as the matching-player in the cut-matching game (Theorem \ref{thm:CMG}) to embed a witness. In the algorithms below, $\phi$ controls the congestion of the embedding, while $\eps$ controls the size of the witness. 
Think of $\eps$ as $1/n^{o(1)}$ and of $\phi$ as $n^{-1/3}$.

\begin{restatable}{lemma}{SparseVertexCutLemma}
\label{lem:matching-vertex} 
\label{lem:vertex-matching}
There is a deterministic algorithm $\vertexmatching(G,A,B,\phi,\eps)$ that, given a directed $n$-vertex graph $G=(V,E)$, two disjoint terminal sets $A,B\subset V$ where $n/4\le|A|\le|B|$, $\phi\in(0,1)$, and $\epsilon\in(0,1)$ , in $\Otil(m/\phi)$ time, either
\begin{itemize}
	\item returns a $O(\phi\log n)$-vertex-sparse $\Omega(\epsilon)$-vertex-balanced
	vertex cut $(L,S,R)$, or
	\item a directed (integral) matching $M$ of size at least $(1-\epsilon)|A|$ from $A$ to $B$ such that there is an embedding $\pset$
	that embeds $M$ into $G$ with vertex congestion $1/\phi$.
\end{itemize}
\end{restatable}
The idea of the above algorithm is to perform $\Otil(1/\phi)$ blocking flow computations. We defer the proof to Section B.5.

The following algorithm finds either a $\Omegahat(\eps)$-vertex-balanced sparse cut, or a $\phi$-witness $W$ that is unweighted and $|V(W)| \ge (1-\eps)n$. As $|V(W)|$ is very close to $n$, we say $W$ is a \emph{very large} witness.
\begin{thm}
\label{lem:certify-witness} 
There is a deterministic algorithm $\certifywitness(G,\phi,\eps)$ that takes as input a directed $n$-vertex graph $G=(V,E)$, $\phi\in(0,1/\log^2(n)]$, and $\eps \in (0,1)$ in $\Ohat(m/\phi)$ time, either
	\begin{itemize}
		\item finds a $\Otil(\phi)$-vertex-sparse $\Omega(\epsilon/n^{o(1)})$-vertex-balanced
		cut $S$, or
		\item certifies that there exists a $\phi$-witness $W$ of $G$ such that $|V(W)|\ge(1-\epsilon)n$ and every edge in $W$ has weight at least $1$. Let $\alphaex = 1/n^{o(1)}$ be the precise expansion factor of $W$ guaranteed by this lemma (we will use this parameter in other lemmas).
	\end{itemize}
\end{thm}

\begin{proof}
	Although there a lot of technical details involved, conceptually speaking the lemma follows quite easily from the cut matching game (Theorem \ref{thm:CMG}) and \vertexmatching (Lemma \ref{lem:matching-vertex}). Define $R = O(\log(n))$ to be the maximum number of rounds in the cut-matching game. Define $\phi' = 4R\phi < 1$ and $\eps' = \eps / \beta$, where $\beta = n^{o(1)}$ will be set later in the proof.
	
	Now, we initiate the cut-matching game. The cut player from theorem \ref{thm:CMG} provides the terminal sets $A_i,B_i$ at every round $i$. The algorithm of this lemma then acts as the matching player: in round $i$, it either return a sparse cut and terminates or embeds matchings  $\overrightarrow{M}_{i}$ and  $\overleftarrow{M}_{i}$. In particular, for each round $i$ of the cut-matching game, the algorithm runs $\vertexmatching(G,A_i,B_i,\phi',\eps')$ as well $\vertexmatching(G,B_i,A_i,\phi',\eps')$ which tries to embed a matching in a reverse direction. We focus on the first of these two invocations, as they are symmetrical. 
	
	If the subroutine $\vertexmatching$ returns a cut $(L,S,R)$, then our algorithm returns the same cut and terminates. Lemma \ref{lem:matching-vertex} guarantees that this cut is $O(\phi'\log(n)) = \Otil(\phi)$-sparse and $\Omega(\eps') = \Omega(\epsilon/n^{o(1)})$-vertex-balanced, as desired. So we assume from now when it returns a path set $\pset_i$ at every round. 
	
	Now let us say that $\vertexmatching$ returns a path set $\pset_i$ that embeds matching $\mstar_i$ from $A$ to $B$. We cannot use this exact matching in the cut matching game because Theorem \ref{thm:CMG} requires a matching of value $|A|$ (a perfect matching), while Property \ref{em-value} only guarantees a matching of value $|A|(1-\eps')$. We thus construct another   matching $F_i$ from $A$ to $B$ ($F$ for fake) such that $\mstar_i \cup F_i$ is a matching of value $|A|$; it is easy to construct such an $F_i$ by starting with $\mstar_i$ and repeatedly adding edges from free vertices in $A$ to free vertices in $B$. (Note that we do not embed these fake edges into $G$.)

	Let $\mstar$ be the union of all the $\mstar_i$, including those ``reverse-direction'' matching from $B_i$ to $A_i$. Let $F$ be the union of all the $F_i$, including those in the reverse direction. Let $\Wstar = (V,\mstar \cup F)$. Theorem $\ref{thm:CMG}$ guarantees that $\Wstar$ is a $\alphacmg = 1/n^{o(1)}$ expander. Note, however, that we cannot return $\Wstar$ as our witness because there is no path set corresponding to edges in $F$ (we never embedded the edges in $F$). We also cannot simply remove $F$ as $\mstar$ on its own might not be an expander.
	
	Instead, we apply directed expander pruning from Theorem \ref{thm:pruning} to $\Wstar$. 
	We feed in all the edges in $F$ as adversarial deletions in the pruning algorithm; since the expansion of $\Wstar$ is at least $\alphacmg = 1/n^{o(1)}$, we can use Corollary \ref{cor:pruning}. Let $P$ be the set returned by pruning, and set $W = \Wstar[V \setminus P]$. 
	
	We now show that $W$ is a $\phi$-witness of the desired size. Let parameter $L$ for pruning be chosen according to Corollary \ref{cor:pruning}, and define $\gamma = \gamma_L(\alphacmg)$ as the parameter from Theorem \ref{thm:pruning}; note that $\gamma = n^{o(1)}$.  By Theorem \ref{thm:pruning}, the expansion factor of $W$ is at least $1/\gamma = 1/n^{o(1)}$. We can thus set parameter $\alphaex$ in the lemma statement to be $\alphaex = 1 / \gamma$. Now, recall that we set $\eps' = \eps / \beta$. We now define $\beta = \gamma\log^2(n)$. By Lemma \ref{lem:vertex-matching}, each set $F_i$ has size at most $\eps'|A| \leq \eps' n = \frac{\eps n}{\gamma\log^2(n)}$, so $F$ has size at most $O(\frac{R \eps n}{\gamma\log^2(n)}) = O(\frac{\eps n}{\gamma \log(n)})$, where the last step follows from $R = O(\log(n))$. Thus, by Theorem \ref{thm:pruning} the pruned set $P$ has volume in $W$ at most $\vol_W(P) \le |F| \cdot \gamma = O(\eps n/\log(n))$. As $W$ has maximum degree $2R$, so $|P| < \eps n$ and $|V(W)| = |V| - |P| \geq (1-\eps)n$. Finally, every edge has weight 1 because $\vertexmatching$ returns integral matchings, so every vertex in $W$ has weighted degree at least $1$ (there are no isolated vertices because $W$ is an expander.)
	
	We must now show that $W$ can be embedded into $G$. We use the embedding $\pset_W \subset \pset$ that is formed by taking all paths in $\pset$ that start AND end in $V \setminus P$, where $P$ is the pruned set from the previous paragraph (note that the middle of the path may still leave $V \setminus P$). It is easy to see that every edge in $W$ has a corresponding path in $\pset_W$, and that the vertex congestion in $\pset_W$ is strictly smaller than in $\pset$. By Lemma \ref{lem:embed-matching}, each $\pset_i$ has vertex-congestion $\phi'$, so since there are at most $2R$ such $\pset_i$ (one in each direction per round of the cut-matching game, which has at most $R$ rounds), $\pset$ has a vertex-congestion of $2R/\phi' < 1/\phi$, as desired.
	
	Finally, we analyze the running time of the algorithm. Each call to $\vertexmatching$ has a running time of $\Ohat(m/\phi') = \Ohat(m/\phi)$; the algorithm makes $O(R) = O(\log(n))$ calls, for a total run-time of $\Ohat(m/\phi)$. The time to construct each $F_i$ is only $O(n)$. Finally, by Corollary \ref{cor:pruning}, the time for pruning is $\Ohat(n)$ as $\Wstar$ has $O(nR)$ unweighted edges.
\end{proof}

\subsection{Embedding a Witness that Obeys Edge Capacities}

We now present an algorithm that tries to find a witness which also obeys the edge capacities $\kappa(e)$. We start by presenting a subroutine that uses an approximate flow algorithm (Lemma \ref{lem:global flow}) to embed a single matching. We then combine this with the cut-matching game to embed a whole witness. If the algorithm fails to find a witness, then one of the approximate-flow computations must have had insufficiently high capacity. We then return the cut $L,S,R$ that certifies this failure. Note that $L,S,R$ might not be sparse in the uncapacitated graph $G$; instead we refer to it as a \emph{bottleneck} cut because the capacities are too low. 

Note that the parameter $d$ establishes a minimum edge-capacity of $1/d$. We will end up setting $d$ to be around the average degree in the input graph. Since the cut-matching game yields a witness with total weight $\Otil(n)$, the witness will have a total of $\Otil(nd) = \Otil(m)$ edges, which will allow us to efficiently run our pruning algorithm on the witness.

\begin{lem} \label{lem:embed-matching} There is an algorithm $\embedmatching(G,\kappa,A,B,\phi,\eps,d)$
with following inputs: an $m$-edge $n$-vertex graph $G=(V,E)$, terminal sets $A,B\subset V$ where
$n/4\le|A|\le|B|$, parameters $\phi \in (0,1/2)$ and $\epsilon \in (0,1)$, a number $d = O(\davg)$ where $\davg$ is the average degree in $G$, and an edge capacity function $\kappa$ where $\kappa(E)=\Ohat(n/\phi)$
and, for each $e\in E$, $\kappa(e)\in[1/d,1/\phi]$ is a positive multiple
of $1/d$. In $\Ohat(m/(\epsilon\phi))$ time the algorithm returns either
\begin{enumerate}
\item a partition $L,S,R$ of $V$ where $\eps n \leq |L| \leq n/2$ and $\kappa(E(L,R)) + |S|/(2\phi) \leq |L| - \eps n$
\item a collection $\pset$ of directed paths from vertices in $A$ to vertices
in $B$ such that
\begin{enumerate}
\item \label{em-multiple}Each path $P\in\pset$ has associated value $\val(P)$ which is a
positive multiple of $1/d$,
\item \label{em-length} Each path $P\in\pset$ has length at most $\Otil(1/(\phi\epsilon^{2}))$.
\item \label{em-value} The total value $\sum_{P\in\pset}\val(P)\in[(1-10\epsilon)|A|,|A|]$,
\item \label{em-vertex} For each $v \in V$, $\sum_{P\in\pset_{v}}\val(P)\le 1/\phi$ where $\pset_{v}$
consists of all paths in $\pset$ that contains $v$.
\item \label{em-edge} For each $e\in E$, $\sum_{P\in\pset_{e}}\val(P)\le \kappa(e)$ where $\pset_{e}$
consists of all paths in $\pset$ that contain $e$.
\end{enumerate}
\end{enumerate}
\end{lem}

\begin{proof}
First, to allow for vertex capacities, we create a graph $G'$ where each $v \in V$ is split into two vertices $v_{in}$ and $v_{out}$. All edges entering $v$ now enter $v_{in}$ and all edges leaving $v$ leave $v_{out}$; there is also a directed edge $(v_{in}, v_{out})$.

We invoke Global Flow (Lemma \ref{lem:global flow}) on $G' = (V',E')$ with the following input: $\Delta(v_{in}) = 1 \ \forall v \in A$ and $T(v_{out}) = 1 \ \forall v \in B$. Set $z = 2\eps n$, $\cmin = 1/d$. The capacity $c(e)$ of edge $e \in E$ is set to $\kappa(e)$ and the capacity of every edge $(v_{in},v_{out})$ is set to $1/\phi$. Note that $c(E') = n/\phi + c(E) = n/\phi + \kappa(E) = \Ohat(n/\phi)$, where the last inequality follows from the bound on $\kappa(E)$ assumed in the lemma. Finally, set parameter $h$ in Lemma \ref{lem:global flow} to be $h = c(E') \frac{10 \log(c(E'))}{\eps n} = \Ohat(\frac{1}{\eps \phi})$. By Lemma \ref{lem:global flow}, the running time is then $\Ohat(mh+\Delta(V)h/c_{\min})= \Ohat(h(m + nd)) = \Ohat(mh) = \Ohat(\frac{m}{\eps\phi})$.

First consider the case that  Lemma \ref{lem:global flow} returns an (edge) cut $S'$ in $G'$. We transform this into a (vertex) cut $(L,S,R)$ in $G$ as follows: $L = \{ v \mid v_{in} \in S' \land v_{out} \in S' \}$, $S = \{ v \mid v_{in} \in S' \land v_{out} \notin S' \}$, $R = \{ v \mid v_{in} \notin S' \}$. 
Lemma \ref{lem:global flow} guarantees that 
$$c(E(S',V\setminus S')) \leq \Delta(S')-z+c(E')\frac{10\log(c(E'))}{h} \leq \Delta(S') - 2\eps n + \eps n = \Delta(S') - \eps n.$$
Now, since $\Delta$ is only non-zero on vertices $v_{in}$, we have that $\Delta(S') = |L| + |S|$. By construction of set $L,S,R$, as well as the fact that every edge $(v_{in},v_{out})$ has capacity $1/\phi$, we also know that 
$$c(E(S',V\setminus S')) \geq c(E(L,R)) + |S|/\phi = \kappa(E(L,R)) + |S|/\phi.$$ 

Combining the above we have that

$$k(E(L,R)) + \frac{|S|}{2\phi} \leq k(E(L,R)) + \frac{|S|}{\phi} - |S| \leq c(E(S',V\setminus S')) - |S|  = c(E(S',V\setminus S')) + |L| - \Delta(S') \leq |L| - \eps n.$$

The above clearly implies that $|L| \geq \eps n$, as desired. We also have that $|L| \leq n/2$ because $|L| \leq |S'|/2 \leq (|V'|/2)/2 = n/2$.

We now turn to the case where Lemma \ref{lem:global flow} returns a flow $f$ in $G'$, which corresponds to a set of paths $\pset$ in $G$. Let us prove that $\pset$ satisfies all the properties of the lemma being proven \emph{except} \ref{em-length}. Properties \ref{em-vertex} and \ref{em-edge} follow immediately from the capacities used in the flow. Property \ref{em-multiple} is true because Lemma \ref{lem:global flow} guarantees that the value of every path is a multiple of $\cmin = 1/d$. For Property \ref{em-value}, note that by Lemma \ref{lem:global flow} we have that
\begin{equation}
\label{eq:pset-val}
\sum_{P\in\pset}\val(P) = \val(f) \geq \Delta(V') - z \geq |A| - 2\eps n \geq |A|(1-8\eps),
\end{equation}
 where the last inequality follows from the assumption of the lemma that $|A| \geq n/4$. 

To ensure Property \ref{em-length}, let $\psetlong = \{ P \in \pset \mid |P| \geq 4h/\eps \}$, and let $\pset' = \pset \setminus \psetlong$. The algorithm returns $\pset'$ instead of $\pset$ as the final path-set. Clearly, since $\pset' \subseteq \pset$, Properties \ref{em-multiple}, \ref{em-vertex} and \ref{em-edge} continue to hold. Property \ref{em-length} also holds by definition of $\pset'$ and the fact that $h = \Otil(\frac{1}{\phi\eps})$. All we have left is to prove \ref{em-value} for $\pset'$. By Equation \ref{eq:pset-val} above, we have that 
$$\sum_{P\in\pset'}\val(P) \geq |A|(1-8\eps) - \sum_{P\in\psetlong}val(P).$$

We now complete the proof by showing that $\sum_{P\in\psetlong}\val(P) \leq \eps n / 4 \leq \eps |A|$. To see this, note that Lemma \ref{lem:global flow} guarantees that $\sum_{P \in \psetlong}|P| \cdot \val(P) \leq hn$. But since $|P| \geq 4h/\eps \ \forall P \in \psetlong$ we have that $\sum_{P \in \psetlong} \val(P) \leq \frac{hn}{4h / \eps} = \eps n / 4$, as desired.  
\end{proof}

\begin{lem} \label{lem:embed-witness} There is an algorithm $\embedwitness(G,\kappa,\phi,d)$ with the following inputs: an $m$-edge $n$-vertex graph $G=(V,E)$, parameter $\phi \in (0,1/2)$, a number $d = O(\davg)$ where $\davg$ is the average degree in $G$, and an edge-capacity function $\kappa$ where $\kappa(E)=\Ohat(n/\phi)$ and, for each $e\in E$, $\kappa(e)\in[1/d,1/\phi]$ is a positive multiple of $1/d$. In $\Ohat(m/\phi)$ time the algorithm returns either
	\begin{enumerate}
		\item a partition $L,S,R$ of $V$ where $\epswit n \leq |L| \leq n/2$ and $\kappa(E(L,R)) + \frac{|S|}{2\phi} \leq |L|$, where $\epswit = 1/n^{o(1)}$ is a parameter we will refer to in other parts of the paper.
		\item A (weighted) $O(\phi\log(n))$-short-witness W of $G$ and a corresponding embedding $\pset$, with the following properties:
		\begin{enumerate}
			\item \label{lem-em-val} For every edge $e\in E$, $\sum_{P\in\pset_{e}}\val(P) = O(\kappa(e)\log(n))$ where $\pset_{e}$ is the set of paths in $\pset$ containing $e$. 
			\item \label{lem-em-size} $|V(W)| = n - o(n)$.
			\item \label{lem-em-weight} The total edge weight in $W$ is $O(n\log(n))$, and every edge weight is a multiple of $1/d$.
			\item \label{lem-em-degree} There are only $o(n)$ vertices in $V(W)$ with weighted degree $\leq 3/4$.
	\end{enumerate}
\end{enumerate}
\end{lem}

\begin{proof}
Although there a lot of technical details involved, conceptually speaking the lemma follows quite easily from the Cut Matching Game (Theorem \ref{thm:CMG}) and \embedmatching (Lemma \ref{lem:embed-matching}). Define $R = O(\log(n))$ to be the maximum number of rounds in the cut-matching game. Recall that $\epswit = 1/n^{o(1)}$ is a parameter we set later.

Now, we initiate the cut-matching game. The cut player from theorem \ref{thm:CMG} provides the terminal sets $A_i,B_i$ at every round $i$. In round i, our algorithm will either return a sparse cut and terminate or embed matchings  $\overrightarrow{M}_{i}$ and  $\overleftarrow{M}_{i}$. In particular, for each round $i$ of the cut-matching game, the algorithm runs $\embedmatching(G,\kappa,A_i,B_i,\phi,\epswit,d)$ as well as \\ \noindent $\embedmatching(G,\kappa,B_i,A_i,\phi,\epswit,d)$. We focus on the first of these two invocations, as they are symmetrical. 

If the subroutine $\embedmatching$ returns a cut $(L,S,R)$, then our algorithm returns the same cut and terminates. Lemma \ref{lem:embed-matching} directly guarantees the properties of $(L,S,R)$ that we need in the lemma being proven.

Now let us say that $\embedmatching$ returns a path set $\pset_i$. We turn this into a fractional matching $\mstar_i$ from $A$ to $B$ in the natural way: for every path $P \in \pset$ from $a \in A$ to $b \in B$, we add an edge from $a$ to $b$ of weight $\val(P)$. By property \ref{em-multiple}, the resulting matching is $1/d$-integral. The only issue is that Theorem \ref{thm:CMG} requires a matching of value $|A|$ (a perfect matching), while Property \ref{em-value} only guarantees a matching of value $|A|(1-10\epswit)$. We thus construct another $1/d$-integral matching $F_i$ ($F$ for fake) such that $\mstar_i \cup F_i$ is a perfect matching. It is easy to construct such an $F_i$ in $O(nd) = O(m)$ time by starting with $\mstar_i$ and repeatedly adding edges of weight $1/d$ from free vertices in $A$ to free vertices in $B$ until the matching is perfect. (Adding multiple copies of the same edge corresponds to increasing the weight of that edge.) Note that we do not embed these fake edges into $G$. 

If in any round $i$ the subroutine $\embedmatching$ returns a cut, then the algorithm terminates. Thus the only case left to consider is when it return a path set $\pset_i$ at every step. Let $\mstar$ be the union of all the $\mstar_i$, including those in the reverse graph. Let $F$ be the union of all the $F_i$, including those in the reverse graph. Let $\Wstar = (V,\mstar \cup F)$. Theorem $\ref{thm:CMG}$ guarantees that $\Wstar$ is a $\alphacmg = 1/n^{o(1)}$ expander. Note, however, that we cannot return $\Wstar$ as our witness because there is no path set corresponding to $F$ (we never embedded the edges in $F$). We also cannot simply remove $F$ as $\mstar$ on its own might not be an expander.

Instead, we apply directed expander pruning from Theorem \ref{thm:pruning}. Let $\Wstar = (V,\mstar \cup F)$. We would like to apply pruning directly to $\Wstar$, but Theorem \ref{thm:pruning} only applies to unweighted graphs. Since the cut-matching game (Theorem \ref{thm:CMG}) guarantees that all edge weights in $W$ are multiples of $1/d$, we can convert $\Wstar$ to an equivalent unweighted multigraph $\Wstar_u$ in the natural way: every edge $e \in \Wstar$ is replaced by $w(e) \cdot d$ copies of an unweighted edge. Note that $\Wstar$ has total weight $O(n\log(n))$, because it contains $O(\log(n))$ matchings; thus $\wu$ contains $O(nd\log(n)) = O(m\log(n))$ edges. We now apply directed pruning to $\Wstar_u$, where we feed in all the edges in $F$ as adversarial deletions; since the expansion of $\Wstar_u$ is at least $\alphacmg = 1/n^{o(1)}$, we can use Corollary \ref{cor:pruning}. Let $P$ be the set returned by pruning, and set $\wu = \Wstar_u[V \setminus P]$ and $W = \Wstar[V \setminus P]$. %

We now show that $W$ is a $O(\phi\log(n))$-witness with the desired properties. Let the pruning parameter $L$ be determined by Corollary \ref{cor:pruning} (with $\alphacmg$ as input variable $\phi$), and define $\gamma = \gamma_L(\alphacmg) = n^{o(1)}$, which is precisely the parameter from Theorem \ref{thm:pruning}.  By Theorem \ref{thm:pruning}, the expansion factor of $\wu$, and hence of $W$, is at least $1/\gamma = 1/n^{o(1)}$, as desired. We now define $\epswit = \gamma/\log^2(n)$. We know that each set $F_i$ has size at most $10\epswit|A| \leq 10\epswit n = \frac{10 n}{\gamma\log^2(n)}$, so $F$ has size at most $O(\frac{Rn}{\gamma\log^2(n)}) = O(\frac{n}{\gamma \log(n)})$, where the last step follows from $R = O(\log(n))$. By Theorem \ref{thm:pruning} the pruned set $P$ satisfies 
$$w(P,V) \leq |F| \cdot \gamma = O(n/\log(n)) = o(n).$$ 
Recall that the cut-matching game (Theorem \ref{thm:CMG}) guarantees that every vertex in $\Wstar$ has weighted degree at least $1$; combined with the above bound on the volume of $w(P,V)$, this proves Properties \ref{lem-em-size} and  \ref{lem-em-degree}. Finally,  Property \ref{lem-em-weight} follows from the fact that $W \subseteq \Wstar$, and $\Wstar$ is the union of $O(\log(n))$ matchings.

We must now show that $W$ can be embedded into $G$. We use the embedding $\pset_W \subset \pset$ that is formed by taking all paths in $\pset$ that start AND end in $V \setminus P$ (note that the middle of the path may still leave $V \setminus P$). It is easy to check that every edge in $W$ has a corresponding path in $\pset_W$, and that the vertex/edge-congestion in $\pset_W$ is strictly smaller than in $\pset$. By Lemma \ref{lem:embed-matching}, each $\pset_i$ has edge-congestion $\phi$, so since there are at most $2R$ such $\pset_i$ (one in each direction per round of the cut-matching game, which has at most $R$ rounds), $\pset$ has a vertex-congestion of $2R\phi = O(\phi\log(n))$. Similarly, the congestion on edge $e$ is at most $2R \kappa(e) = O(\kappa(e)\log(n))$, which proves Property \ref{lem-em-val}.  

Finally, we analyze the running time of the algorithm. Each call to \noindent $\embedmatching$ has a running time of $\Ohat(m/(\epswit\phi)) = \Ohat(m/\phi)$; the algorithm makes $O(R) = O(\log(n))$ calls, for a total run-time of $\Ohat(m/\phi)$. In each round of the cut-matching game, the cut-player from Theorem \ref{thm:CMG} requires $\Ohat(nd) = \Ohat(m)$ time to compute the terminal sets $A_i,B_i$. The time to construct each $F_i$ is $O(nd) = O(n\davg) = O(m)$.  Finally, by Corollary \ref{cor:pruning}, pruning requires $\Ohat(m)$ time. 
\end{proof}

\begin{algorithm2e}
	\label{alg:rwitness}
	\caption{Algorithm $\rwitness(G_0 = (V_0,E_0),\phi)$ (see Theorem \ref{thm:robust witness})}
	
	\SetAlgoSkip{}
	\SetKwProg{procedure}{Procedure}{}{}
	\SetKwBlock{RepeatUntilWitness}{Repeat Until $\EmbedWitness(G,\kappa,\phi',d)$
	 {\normalfont returns a witness}}
 
 	\SetKwBlock{BlockCertifyWitness}{$\certifywitness(G,\phi,\epswit/2)$}
 	
 	Let $n = |V_0|, m = |E_0|$ \;

	Initialize $\phi' = \phi \alphaex / \log^2(n)$ 
	\tcp*[f]{$\alphaex = n^{o(1)}$ is the parameter from Lemma \ref{lem:certify-witness}} \label{line:phi-witness}  \;
		 
	Set $d$ to be the smallest number $\geq \davg$ such that $d / \phi'$ is a power of $2$ \tcp*[f]{Note that $d \in [\davg, 2\davg]$} \label{line:d-witness}  \;
		
	Initialize $G \gets G_0$ \;
		
	Initialize $\kappa(e) = 1/d \ \forall e \in E_0$ \;

	\procedure(\tcp*[f]{execute before processing adversarial deletions}){Begin New Phase}{
		
	\BlockCertifyWitness(\label{line:phase-begin-witness} \tcp*[f]{$\epswit$ is the parameter from Lemma \ref{lem:embed-witness}}) {
			{\bf if} existence of witness certified, {\bf then} continue \;
			{\bf else} return cut given by \certifywitness\ and {\bf Terminate} \;
		} %

		\RepeatUntilWitness(\label{line:embed-witness}){
	
			Let $(L,S,R)$ be the vertex-cut returned by $\EmbedWitness$ \label{line:return-cut-witness}  \;
			
			$\estar \gets \{ e \in E(L,R) \mid \kappa(e) < 1/\phi' \}$ \tcp*[f]{will show: $\kappa(e) \leq 1/(2\phi') \ \forall e \in \estar$} \label{line:estar-witness}  \;
			
			$\kappa(e) \gets 2\kappa(e)$ for all $e \in \estar$ \label{line:kappa-increase-witness}  \;
		} %
	
 Set $W$ to be the witness returned by $\embedwitness(G,\kappa,\phi',d)$ and set $\pset_W$ to be the corresponding embedding	\label{line:return-witness} \;
		
		Create unweighted multi-graph $\wu$ as follows: $V(\wu) = V(W)$ and for every edge $(u,v) \in W$ add $d \cdot w(u,v)$ copies of edge $(u,v)$ to $\wu$. (Here, we use the fact that all weights in $W$ are multiples of $1/d$; See Lemma \ref{lem:embed-witness}.) \tcp*[f]{$\wu$ is basically identical to $W$; we convert to an unweighted graph only so that we can apply pruning from Theorem \ref{thm:pruning}} \;
		
		 Initialize the pruning algorithm from Theorem \ref{thm:pruning} on $\wu$ \label{line:pruning-witness} \;
		
		$\counter \gets 0$ \tcp*[f]{Tracks volume of vertices are pruned from $W$.}	
	} %

	\procedure{Processing Deletion of edge $(u,v)$}{
	 $W_0 \gets W$ \label{line:original-witness} \tcp*[f]{$W_0$ will always refer to the original witness returned in Line \ref{line:return-witness}, before deletions are processed in this phase} \;
		
		Let $\pstar$ contain all paths in $\pset_W$ that go through $(u,v)$ \label{line:pruning-begins}  \;
		
	 	Let $\estar \subseteq E(W)$ contain the edges in $W$ corresponding to $\pstar$ \label{line:witness-delete-estar} \;
		
		$\pset_W \gets \pset_W \setminus \pstar$; $E(W) \gets E(W) \setminus \estar$ \;
		
		Input all copies of edges in $\estar$ as adversarial deletions into the pruning algorithm on $\wu$ from Line \ref{line:pruning-witness}. Let $X$ contain the vertices in $\wu$ that were added to the pruned set as a result of these deletions \;
		
		$\counter \gets \counter + \vol_{W_0}(X)$ \tcp*[f]{tracks total volume pruned from $W_0$}\;
		\If(\label{line:pruning-ends}){$\counter \geq n/50$}{ 
			
		RESET PHASE: go back to Line \ref{line:phase-begin-witness} \tcp*[f]{Note: capacities $\kappa$ are NOT reset between phases} \label{line:phase-reset-witness}  \;
		} %
	
		$W \gets W[V(W) \setminus X]$ \label{line:final-line-witenss} \;
	} %
\end{algorithm2e}

\subsection{Analysis of \rwitness (Algorithm \ref{alg:rwitness})}

Recall that $\phi'$ from Line \ref{line:phi-witness} is the input to Algorithm \embedwitness (Line \ref{line:embed-witness}).

\begin{observation}
\label{obs:kappa}
Throughout Algorithm \ref{alg:rwitness}, $\kappa$ is non-decreasing and in particular only changes by doubling in Line \ref{line:kappa-increase-witness}. Moreover, if $\kappa(e) < 1/\phi'$ then $\kappa(e) \leq 1/(2\phi')$, and we always have $\kappa(e) \leq 1/\phi' \ \forall e \in E(G)$ (here we use that fact that in Line \ref{line:d-witness} we set $d$ so that $d/\phi'$ is a power of $2$).
\end{observation}

\begin{defn} [Min-cost Embedding]
\label{defn:pi}
Define potential function $\Pi(G,\kappa)$ as follows. Let $d$ be the parameter from Line \ref{line:d-witness} of Algorithm \ref{alg:rwitness}, and recall that $\kappa(e) \geq 1/d \ \forall e \in E(G)$. Let $\wset$ be a collection of all path sets $\pset$ such that $\pset$ embeds a $\phi$-witness $W$ into $G$ for which $|V(W)| \geq (1-\epswit/2) n$ and $W$ is a $\alphaex$-expander. Define the {\bf cost} of an edge $e$ to be $c(e) = \log(d \kappa(e))$; note that since $\kappa(e) \geq 1/d$, $c(e)$ is always non-negative. For any path set $\pset$, define $\val(e) = \sum_{P \in \pset_e} \val(P)$, where $\pset_e$ is the set of paths going through $e$. Define $c(\pset) = \sum_{e \in E} c(e) \val(e)$. Then, we define $\Pi(G,\kappa) = \min_{\pset \in \wset} c(\pset)$, and we call the corresponding $\pset$ the {\bf minimum cost embedding} into $G$. If $\wset = \emptyset$ then $\Pi(G,\kappa) = \infty$.
\end{defn}

We now state a few simple observations 

\begin{observation}
If $\kappa(e)$ increases for some edge $e$, then $\Pi(G,\kappa)$ cannot decrease as a result.
\end{observation}

\begin{observation}
Let $G = (V,E)$ and $G' = (V,E')$ an edge-subgraph with $E' \subset E$. Then, for any capacity function $\kappa$, $\Pi(G,\kappa) \leq \Pi(G',\kappa)$ (they could both be infinite).
\end{observation}

\begin{proof}
Let $\pset'$ be the minimum-cost embedding into $G'$. It is not hard to check that $\pset'$ is also a valid embedding into $G$. 
\end{proof}

\begin{observation}
\label{obs:witness-potential-initial}
At the beginning of Algorithm \ref{alg:rwitness}, $\Pi(G,\kappa) = 0$ (because for all $e \in E$, $\kappa(e) = 1/d$ so $c(e) = 0$). Moreover, $\Pi(G,\kappa)$ only increases throughout the course of the algorithm, and if $\Pi(G,\kappa) = \infty$ then it will remain so forever (this follows from the  observations above, as well as the fact that $G$ is decremental, so edges are never inserted).
\end{observation}

\begin{observation}
\label{obs:infinite-potential}
If $\Pi(G,\kappa) = \infty$ then $\certifywitness(G,\phi,\epswit/2)$ from Line \ref{line:phase-begin-witness} of Algorithm \ref{alg:rwitness} returns a sparse cut and terminates. (Because $\Pi(G,\kappa) = \infty$ means that $\wset = \emptyset$, so there is no valid witness.)
\end{observation}

We have established that $\Pi$ starts at $0$ and only increases. We now show that as long as the algorithm does not terminate, $\Pi$ is never too large. 

\begin{lemma}
\label{lem:witness-potential-upper}
Consider any phase in which the algorithm did not terminate. Let $G$ be the graph at the beginning of that phase (before any deletions have been processed in that phase), and let $\kappa$ be the capacities at the end of initialization for that phase (Line \ref{line:return-witness}). Then $\Pi(G,\kappa) = \Otil(n/\phi)$.
\end{lemma}

\begin{proof}
Since the algorithm did not terminate in this phase, $\certifywitness$ in Line \ref{line:phase-begin-witness} must have certified the existence of some $\phi$-witness $W$ with embedding $\pset$. Note that this witness satisfies all the properties in the definition of $\Pi(G,\kappa)$; thus, $\Pi(G,\kappa) \leq c(\pset)$. 
We complete the proof by showing that $c(\pset) = \Otil(n/\phi)$. Firstly, note that because $\pset$ has vertex-congestion $1/\phi$, $\sum_{P \in \pset} |P| \leq n/\phi$. Secondly, by Observation \ref{obs:kappa}, for every edge $e$ we always have $c(e) \leq \log(d\kappa(e)) \leq \log(d/\phi) = O(\log(n))$. We thus have $c(\pset) = \sum_{P \in \pset} \sum_{e \in P} c(e) = O(\log(n)\sum_{P \in \pset} |P|) = O(n\log(n)/\phi)$.
\end{proof}

\begin{defn}
\label{def:kappa-zero-witness}
Let $E_0$ be the edge set of the input graph to Algorithm $\rwitness(G_0,\phi)$, before any adversarial deletions. Note that even if $e \in E_0$ is later deleted by the adversary, $\kappa(e)$ is still well-defined:  $\kappa(e)$ cannot increase after $e$ is deleted, so it is equal to the capacity right before $e$ is deleted. We can thus define $\kappa(E_0) = \sum_{e \in E_0} \kappa(e)$.
\end{defn}

\begin{lemma}
\label{lem:witness-potential-main}
Consider some invocation of $\embedwitness(G,\kappa,\phi',d)$ in Line \ref{line:embed-witness} of Algorithm \ref{alg:rwitness} that returns a cut $(L,S,R)$. Let $\kappa$ be the capacity function before the doubling step in Line \ref{line:kappa-increase-witness}, and $\kappa'$ the capacity function after the doubling step. Then, the following holds:
\begin{enumerate}
	\item $\kappa'(E_0) \leq \kappa(E_0) + n$. 
	\item $\Pi(G,\kappa') \geq \Pi(G,\kappa) + n^{1-o(1)}$. 
\end{enumerate}
\end{lemma}

\begin{proof}
The first property is simple. By Lemma \ref{lem:embed-witness}, we have $\kappa(E(L,R)) \leq |L| \leq n$. Since $\estar \subseteq E(L,R)$ (see Line \ref{line:estar-witness}), and the algorithm doubles all capacities in $\estar$, we have that $\kappa'(E_0) - \kappa(E_0) = \kappa(\estar) \leq \kappa(E(L,R)) \leq n$, as desired.

To prove the second property, note that since the algorithm did not terminate in Line \ref{line:phase-begin-witness}, there must exist some embedding $\pset$ of a $\phi$-witness $W = (V_W,E_W)$ as in Lemma \ref{lem:certify-witness}. In particular, $W$ has expansion $\alphaex$ and 

\begin{equation}
\label{eq:vw}
|V_W| \geq V - \epswit n / 2.
\end{equation}

To complete the proof, we now establish the following claim:

\begin{claim}
\label{claim:potential}
Let $W = (V_W,E_W)$ be \emph{any} witness satisfying the properties of the witness certified by $\certifywitness(G,\phi,\epswit/2)$ (Lemma \ref{lem:certify-witness}), and let $\pset$ be the corresponding embedding. Let $L_W = L \cap W$. Recall the set $\estar$ from Line \ref{line:estar-witness} and let $\psetcrit$ be the set of all paths in $\pset$ that contain at least one edge in $\estar$. Then, the following holds:
\begin{enumerate}
	\item $|E_W(L_W,R \cup S)| \geq |L|\alphaex / 2$.
	\item $|\psetcrit| = \Omega(|L|\alphaex)$.
\end{enumerate} 
\end{claim}

\paragraph{Proof of First Claim Property:} Lemma \ref{lem:embed-witness} guarantees that $|L| \geq \epswit n$. Combined with Equation \ref{eq:vw} we have $$|L_W| \geq |L| - \epswit n / 2 \geq |L|/2.$$ Since $W$ is an expander, it contains no isolated vertices, so we clearly have $|E_W(L_W,V_W)| \geq |L_W|$. Thus, by the expansion of $W$, 
\begin{equation}
\label{eq:lws}
|E_W(L_W, R \cup S)| \geq \alphaex |E_W(L_W,V_W)| \geq 
\alphaex |L_W| \geq \alphaex |L|/2. 
\end{equation}

\paragraph{Proof of Second Claim Property}
Let $\pset$ be the embedding of $W$ into $G$. Let $\efull = E(L,R) \setminus \estar = \{ e \in E(L,R) \mid \kappa(e) = 1/\phi'\}$. Note that $E(L,V \setminus L)$ is the disjoint union of $\estar, \efull$ and $E(L,S)$. Consider any path in $\pset$ that corresponds to an edge in $E_W(L_W, V \setminus L_W)$ in $W$. We will categorize these by the first edge on the path that goes from $L$ to $V \setminus L$: if that edge is in $\estar$ then we put $P$ in $\psetstar$; if that edge is in $\efull$ then we put $P$ in $\psetfull$; and if that edge is in $E(L,S)$ then we put $P$ in $\psets$. By the first property of this claim we have 
$$|\psetstar| + |\psetfull| + |\psets| \geq \alphaex|L|/2.$$

Now, by Lemma \ref{lem:certify-witness} $\pset$ has vertex-congestion $1/\phi$ (and hence edge congestion $1/\phi$), so $|\psetfull| \leq |\efull|/\phi$ and $|\psets| \leq |S| / \phi$. But now, recall from Lemma \ref{lem:embed-witness} that $\kappa(E(L,R)) + |S|/(2\phi') \leq |L|$.  By definition of $\phi'$ in Line \ref{line:phi-witness} of Algorithm \ref{alg:rwitness} this implies 
$$|\psets| \leq \frac{S}{\phi} = \frac{S}{2\phi'} \cdot \frac{2\alphaex}{\log^2(n)} \leq \frac{2|L|\alphaex}{\log^2(n)}$$
Similarly, note that $\kappa(E(L,R)) \leq \kappa(\efull) \leq |\efull| / \phi'$, so doing out the same algebra as above we have 
$$|\psetfull| \leq |\efull|/\phi \leq \frac{|L|\alphaex}{\log^2(n)}.$$ Combining the equations above we have
$$|\psetstar| \geq \alphaex|L|/2 - |\psetfull| - |\psets| = \Omega(|L|\alphaex) - o(|L| \alphaex) - o(|L| \alphaex) = \Omega(|L| \alphaex)$$
This completes the proof, as $\psetstar \subseteq \psetcrit$, where $\psetcrit$ is the path set in the lemma statement.

\paragraph{Back to Proof of Property 2 of Lemma \ref{lem:witness-potential-main}}
Let $c$ be the cost function corresponding to $\kappa$ and $c'$ to $\kappa'$: so $c(e) = \log(d\kappa(e))$ and $c'(e) = \log(d\kappa'(e))$. Let $\pset'$ be the min-cost embedding such that $c'(\pset') = \Pi(G,\kappa')$. Note that since $\pset'$ is a valid embedding into $G$, we have that $\Pi(G,\kappa) \leq c(\pset')$. Now, observe that $c(e) = c'(e) - 1$ for all $e \in \estar$ and $c(e) = c'(e)$ for all other edges. By the second property of Claim \ref{claim:potential} we know that $\pset'$ contains at least $|L|\alphaex$ paths that go through $\estar$. We also know from Lemma \ref{lem:embed-witness} that $|L| \geq n\epswit$. We thus have that the desired: 
$$\Pi(G,\kappa') = c'(\pset') \geq c(\pset') + n \alphaex \epswit \geq \Pi(G,\kappa) + n \alphaex \epswit \geq \Pi(G,\kappa) + n^{1-o(1)}$$
\end{proof}

\begin{corollary}
\label{cor:witness-numtimes}
In any execution of Algorithm \rwitness, the total number of times that $\embedwitness$ in Line \ref{line:embed-witness} returns a cut is  $\Ohat(1/\phi)$.
\end{corollary}

\begin{proof}
First we argue that whenever $\embedwitness(G,\kappa,...)$ is called, $\Pi(G,\kappa)$ is finite. Firstly, note that $\kappa$ only affects the magnitude of $\Pi(G,\kappa)$, not whether it is finite or infinite. Thus, if $\Pi(G,\kappa)$ is finite the first time $\embedwitness$ is called in a phase, it will be finite every time $\embedwitness$ is called in that phase. Now, before running $\embedwitness$ for the first time in a phase we always call $\certifywitness$ in Line \ref{line:phase-begin-witness}, and by Observation \ref{obs:infinite-potential}, if $\Pi(G,\kappa)$ were infinite, then $\certifywitness$ would return a sparse cut and terminate the entire algorithm.

Thus, every time $\embedwitness(G,\kappa,...)$ is called, $\Pi(G,\kappa)$ is finite, and by Lemma \ref{lem:witness-potential-main}, it increases by at least $n^{o(1)}$. This completes the proof when combined with the fact that the potential starts at 0 and never decreases (Observation \ref{obs:witness-potential-initial}), and that if finite the potential is always $\Ohat(n/\phi)$ (Lemma \ref{lem:witness-potential-upper}).
\end{proof}

\begin{corollary}
\label{cor:witness-total-kappa}
Throughout the execution of Algorithm $\rwitness$ we have $\kappa(E_0) = \Ohat(n/\phi)$. Note that Lemma \ref{lem:embed-witness} requires this of the input capacity function $\kappa$, so this corollary ensures this input assumption is always valid. (Recall from Definition \ref{def:kappa-zero-witness} that the upper bound counts $\kappa(e)$ for \emph{all} edges $e \in E_0$, including those that were deleted from $G$.) 
\end{corollary}

\begin{proof}
$\kappa$ only changes in Line \ref{line:kappa-increase-witness}, so the corollary follows directly from Property 1 of Lemma \ref{lem:witness-potential-main} and Corollary \ref{cor:witness-numtimes}.
\end{proof}

\begin{lemma}
\label{lem:witness-num-phases}
The total number of phases is at most $\Ohat(1/\phi)$
\end{lemma}

\begin{proof}
Recall that $E_0$ is the original edge set of the graph. At any given time during the execution of the algorithm, let $\edel$ contain all edges that were deleted from $E_0$ by the adversary. Note that if $e \in \edel$, then the algorithm will never increase it's capacity, so $\kappa(e)$ is the capacity of the edge right before it was deleted. 

Consider potential function $\Phidel = \sum_{e \in \edel} \kappa(e)$. Clearly $\Phidel$ starts at time $0$ and can only increase. By Corollary \ref{cor:witness-total-kappa}, $\Phidel$ is always $\Ohat(n/\phi)$. We now complete the proof by showing that every phase that does not terminate the algorithm increases $\Phidel$ by $\Omegahat(n)$.

Consider any phase, and let $W_0$ be the witness returned by \embedmatching\ in that phase (Line \ref{line:return-witness}) before any deletions have been processed in this phase, and let $w$ be the edge-weight function for $W_0$. The witness $W_0$ is then pruned as edges in $G$ are deleted. (Although pruning is technically done through the intermediary of unweighted graph $\wu$, we will conceive of it as applying directly to the weighted version, as the two are equivalent.) Let $K$ be the total capacity of all edges deleted from $G$ by the adversary in this phase. By Property \ref{lem-em-val} of Lemma \ref{lem:embed-witness}, the total weight of edges in $\estar$ that are deleted from $W_0$ (Line \ref{line:witness-delete-estar}) is at most $O(K\log(n))$. All these edges are then inputted as adversarial deletions to the pruning algorithm. Let $P$ be the final set of vertices pruned from $W_0$ before the phase ends. By Theorem \ref{thm:pruning}, we have 
$$\vol_{W_0}(P) = w(E_{W_0}(P,V)) = K \cdot \log(n) \cdot n^{o(1)} = \Ohat(K).$$ 

Since $P$ was the pruned set when the phase ended, we must have $\vol_{W_0(P)} \leq n/50$ (see \ref{line:pruning-ends}). 
Combining with the above equation we get $n = \Ohat(K)$, so $[\textrm{increase in $\Phidel$}] = K = \hat{\Omega}(n)$, as desired.
\end{proof}

\paragraph{Correctness Analysis of Algorithm $\rwitness$} 
We now prove that the algorithm satisfies all the properties of Theorem \ref{thm:robust witness}. Recall that the algorithm maintains a witness until at some point it terminates and returns a cut. A cut is only returned by $\certifywitness$ (Line \ref{line:phase-begin-witness}) and by Lemma \ref{lem:certify-witness}, this cut is $\phi n^{o(1)}$-vertex-sparse and ($1/n^{o(1)}$)-vertex-balanced, as desired. 

The algorithm only returns a witness via subroutine $\embedwitness$ (Line \ref{line:embed-witness}). Let $W_0$ be the witness returned, before deletions are processed in this phase. By Lemma \ref{lem:embed-witness}, $W_0$ clearly satisfies all the properties of Theorem \ref{thm:robust witness}. $W_0$ then undergoes pruning (Theorem \ref{thm:pruning}) in Lines \ref{line:pruning-begins}-\ref{line:final-line-witenss}. Let $W$ denote the pruned witness. All the relevant properties of $W$ remain the same under pruning except the expansion factor, the size of $V(W)$, and the weighted degrees in $W$. Corollary \ref{cor:pruning} guarantees that the expansion factor of $W$ remains $1/n^{o(1)}$. Letting $P$ be the pruned set before termination, we know that  $\vol_{W_0}(P) \leq n/50$ (Line \ref{line:pruning-ends}). We know that $W_0$ had $n - n^{o(1)}$ vertices of weighted degree $\geq 3/4$ before pruning (Properties \ref{lem-em-size} and \ref{lem-em-degree} of Lemma \ref{lem:embed-witness}), so it is easy to see that after $n/50$ volume is pruned away, there are still at least $9n/10$ vertices in $W$ and at most $2n/25 \leq |V(W)|/10$ of them have degree $\leq 1/2$, so $W$ is a large $\phi$ witness, as desired.

Theorem \ref{thm:pruning} also requires that the witness is decremental within each phase, which is clearly true because within a phase the witness changes only via pruning. Finally, Lemma \ref{lem:witness-num-phases} shows that the total number phases is $\Ohat(n/\phi)$, as desired.

\paragraph{Running Time Analysis of Algorithm $\rwitness$}
We now show that Algorithm \ref{alg:rwitness} has running time $\Ohat(n/\phi^2)$, as required by Theorem \ref{thm:robust witness}. Since $\phi' = \phi / n^{o(1)}$, the subroutines 
\noindent $\embedmatching$ and $\certifywitness$ both require $\Ohat(m/\phi)$ time. Since Lemma \ref{lem:embed-matching} guarantees that the witness returned in Line \ref{line:return-witness} has expansion $1/n^{o(1)}$, Corollary \ref{cor:pruning} guarantees that the total run-time of pruning within a single phase is $\Ohat(m)$. Each phase thus requires $\Ohat(m/\phi)$ time, plus another $\Ohat(m/\phi)$ time for every call to $\embedmatching$ that returns a cut (since this can happen multiple times within a single phase). The total running time is thus $\Ohat(m/\phi) \cdot ([\textrm{\# of phases}] + [\textrm{\# of invocations to $\embedmatching$ that return a cut}])$. By Lemma \ref{lem:witness-num-phases} and Corollary \ref{cor:witness-numtimes}, both of those terms are $\Ohat(1/\phi)$, so the total running time is $\Ohat((m/\phi) \cdot (1/\phi)) = \Ohat(m/\phi^2)$, as desired.

\section{Directed Expander Pruning}
\label{sec:pruning}

In this section, we present the implementation and analysis of an pruning procedure for directed graphs. Our main result of the section is summarized in the theorem below.

\begin{restatable}[Directed Expander Pruning]{thm}{expanderPruning}
\label{thm:pruning} There is a deterministic
algorithm with the following input: a directed unweighted decremental multi-graph
$W=(V,E)$ with $n$ vertices and $m$ edges that is initially a $\phi$-expander
and a parameter $L \ge 1$. The algorithm maintains an incremental set $P\subseteq V(W)$ using $\tilde{O}\left(\frac{m n^{1/L}}{\gamma_{_L}(\phi)}\right)$ total update time such that for $\overline{P} = V \setminus P$, we have that $W[\overline{P}]$ is a $\gamma_{_L}(\phi)$-expander and $\vol_{W}(P)\le O\left(\frac{t n^{1/L}}{\gamma_{_L}(\phi)}\right)$ after $t$ updates, where $\gamma_{_L}(\phi) = \phi^{3^{O(L)}}$.
\end{restatable}
To ease working with the theorem above, let us introduce the following corollary.

\begin{corollary}
\label{cor:pruning}
Say that the graph given in Theorem \ref{thm:pruning} is initially a $\phi$-expander for $1/\phi = n^{o(1)}$. Then, there exists a setting for $L$ such that $L = \omega(1)$ and $1/\gamma_{_L}(\phi) = 1/\phi^{3^{O(L)}} = n^{o(1)}$. Note that the running time of Theorem \ref{thm:pruning} is then $\Ohat(m)$.
\end{corollary}

\begin{proof}[Proof of Corollary]
We start by specifying the constant inside the big-O notation: say that $1/\gamma_{_L}(\phi) \leq 1/\phi^{3^{cL}}$ for some constant $c$. Note that since $1/\phi = n^{o(1)}$ we have $\log_{1/\phi}(n) = \omega(1)$. Now, set $L = \frac{1}{2c} \cdot \log_3\log_{1/\phi}(n) = \omega(1)$. We have $3^{cL} = \sqrt{\log_{1/\phi}(n)}$. Thus $1/\phi^{3^{cL}} = 1/\phi^{\sqrt{\log_{\phi}(n)}} = n^{1/\sqrt{\log_{\phi}(n)}}$, which is $n^{o(1)}$ because $\log_{\phi(n)} = \omega(1)$.
\end{proof}

The proof strategy for \Cref{thm:pruning} follows on a high-level previous approaches (see for example \cite{NanongkaiS17, NanongkaiSW17}): we first provide a simple pruning procedure that is given an expander $W$ and a batch $B$ of edges that where deleted from $W$ and finds either a sparse cut in $W \setminus B$ of size roughly $|B|$ or certifies that $W \setminus B'$ is still an expander where $|B'| \ll |B|$ which can then be applied recursively. We call this kind of procedure \emph{one-shot pruning} and the algorithm and analysis of such a procedure is the main result of \Cref{subsec:oneShotPruning}. Using this sub-routine, we can then show how to give a \emph{dynamic pruning procedure}. This reduction is described in 
\Cref{subsec:DynamicExpanderPruning} where we also prove \Cref{thm:pruning}.

\subsection{One-Shot Pruning}
\label{subsec:oneShotPruning}

Let us begin the description of one-shot pruning by defining the concept of a near out-expander and near expander, both natural generalizations of the definition of an expander.

\begin{defn}[Near Out-Expander]  \label{def:nearOutExpander}
Let $G=(V,E)$ be a directed weighted graph. We say that $A \subseteq V$ is a \emph{near $\phi$-out-expander in $G$} if
\[
    \forall S \subset A, \vol_{G}(S)\le\vol_{G}(A)/2: \delta_{G}^{out}(S)\ge\phi\vol_{G}(S).
\]
\end{defn}
\begin{defn}[Near Expander] \label{def:nearExpander}
 Let $G=(V,E)$ be a directed weighted graph. We say that $A \subseteq V$ is a \emph{near $\phi$-expander in $G$} if $A$ is a near $\phi$-out-expander in $G$ and $\rev{G}$.
\end{defn}
\noindent
Given \Cref{def:nearExpander}, we can now state the guarantees of our one-shot pruning procedure. 

\begin{lem}
[Large Sparse Cut or Almost Expander]\label{lem:local cut} 
Given an unweighted multi-graph $W=(V,E)$, a boundary $P \subseteq V$, and a core $\overline{P} = V \setminus P$ where we let the boundary edges be edges between boundary and core denoted by $B = E_W(P, \overline{P}) \cup E_W(\overline{P}, P)$ and have that $E = E(W[\overline{P}]) \cup B$, i.e. the graph $W$ consists of edges between vertices in the core and boundary edges. Further, given some conductance parameter $\phi \in (1/n^2,1)$ such that $\overline{P}$ is a near $\phi$-expander in $W$ and the set of boundary edges $B$ has size at most $\phi m / 100$.

Then, there exists a deterministic algorithm that takes an integer $z$, and returns either
\begin{enumerate}
\item a set $B' \subseteq B$ of size at most $2z$ such that $\overline{P}$ is a near $\frac{\phi^2}{24}$-expander in the graph $W \setminus (B \setminus B')$, or
\item \label{item:largeCutSmallCond} a set $P' \subseteq \overline{P}$ where $\phi z/16 < \vol_W(P') \leq \vol_W(\overline{P})/2$ and 
\[
    \min\{\delta^{out}_{W[\overline{P}]}(P'), \delta^{in}_{W[\overline{P}]}(P')\} \leq \phi \cdot \vol_W(P') .\]
\end{enumerate}
The algorithm has running time $O\left(\frac{|B|\log n}{\phi}\right)$.
\end{lem}

Let us give such a deterministic algorithm that satisfies the guarantees stated above. We therefore start by setting up a flow problem $\Pi_{out} = (\Delta_{out}, T_{out}, c_{out})$ such that if the flow is feasible, we have that $\overline{P}$ is a near $\frac{\phi^2}{24}$-out-expander in $W$ as defined in \Cref{def:nearOutExpander} and otherwise we obtain a cut $P'$ as described in item \ref{item:largeCutSmallCond}.

Before we set up $\Pi_{out}$, let us define a slightly modified graph $W^{out} = (V^{out}, E^{out})$ of $W$ that is more convenient to work with. Of utmost importance in our flow problem are the edges $B^{out} = E_W(\overline{P}, P)$ that is the edges leaving $\overline{P}$. The graph $W^{out}$ differs from $W$ in the $B^{out}$ edges which are mapped to distinct endpoints in the boundary and then reversed so that they can inject flow using these edges.

More formally, we let $P^{out}$ be a set of vertices where there is a vertex associated with each edge in $B^{out}$ and let $\pi$ be the bijective mapping from edges in $B^{out}$ to $P^{out}$. We let $R^{out}$ be the set containing for every edge $(u,v) \in B^{out}$, the reversed edge after the head $v$ was mapped to $\pi(u,v)$, i.e. the vertex in $P^{out}$ associated with the edge $(u,v)$. That is $(u,v) \in B^{out}$ if and only if $(\pi(u,v), u) \in R^{out}$. Finally, we can define the graph $W^{out} = (V^{out} = V \cup P^{out}, E^{out} = (E \setminus B^{out}) \;\cup\;R^{out})$.

We can then set-up the flow problem $\Pi_{out} = (\Delta_{out}, T_{out}, c_{out})$ on the graph $W^{out}$ by setting 
\[
    \Delta_{out}(u) = 
    \begin{cases} 
        4/\phi & \mbox{if } u \in P^{out}\\
        0 & \mbox{if } u \in \overline{P}
    \end{cases}
\]
so that we have that all sources $u$ of the flow problem are in the boundary $P^{out}$ contributing with $1/\phi$ units of flow which gives in particular that $\Delta(V^{out}) = 4 \cdot \delta^{out}_W(\overline{P})/ \phi$. We let the sink function be defined $T_{out}(u) = \deg_{W^{out}}(u) = \deg_W(u)$ for all $u \in \overline{P}$ and otherwise $0$, and define the capacity $c_{out}(e) = 24/ \phi^2$ for each edge $e \in E^{out}$. 

We then invoke \Cref{lem:local flow} on the problem $\Pi_{out}$ with $z$ as given, $\overline{\Delta} = 4/\phi$ and $h=\frac{12 \cdot 40 \log n}{\phi}$. We have that the constraint one the parameters in \Cref{lem:local flow} is satisfied since $\Delta(V^{out}) = 4 \cdot \delta^{out}_W(\overline{P})/ \phi$ as seen earlier and by our assumption that $\delta^{out}_W(\overline{P}) + \delta^{in}_{W}(\overline{P}) \leq \frac{\phi^2}{24} \vol_W(\overline{P})$.

\noindent
Thus, in time $O\left(\frac{|B^{out}|) \log n}{\phi}\right)$, we obtain either 
\begin{enumerate}
    \item a pre-flow $f$ with total excess at most $z$, or \label{item:algReturnsFlow}
    \item a cut $S$ such that $\phi z/4<\vol_{W^{out}}(S)\le |E(W^{out})|/2$ satisfying $c(E_{W^{out}}(S,V\setminus S))\le\Delta_{out}(S)-T(S)-z+c_{out}(E_{W^{out}}(S,V)\cup E_{W^{out}}(V,S))\cdot\frac{40\log n}{h}$ where we use that the total capacity is bounded by $\sum_{e \in E^{out}} c_{out}(e) < n^4$. \label{item:algReturnsCut}
\end{enumerate}
We now state two claims and show how they establish the lemma. We then prove these two claims.

\begin{restatable}{claim}{algReturnsFlow}
\label{clm:algReturnsFlow}
If the algorithm ends with scenario \ref{item:algReturnsFlow}, then we find a set of edges $B'' \subseteq B^{out}$ of size at most $z$, such that $\overline{P}$ is a near $\frac{\phi^2}{24}$-out-expander in $W \setminus (B^{out} \setminus B'')$.
\end{restatable}

\begin{restatable}{claim}{algReturnsCut}
\label{clm:algReturnsCut}
If the algorithm ends with scenario \ref{item:algReturnsCut}, then we can find a set $P' \subseteq \overline{P}$ where $\phi z/16 < \vol_W(P') \leq \vol_W(\overline{P})/2$ and 
\[
    \delta^{out}_{W[\overline{P}]}(P') \leq \phi \cdot \vol_W(P').
\] 
\end{restatable}

Given the two claims, we obtain \Cref{lem:local cut} almost as a corollary.

\begin{proof}[Proof of \Cref{lem:local cut}]
It is then not hard to see that if we run the above algorithm on $W$ and $\rev{W}$, that we either have scenario \ref{item:algReturnsCut} for at least one of the problems and therefore by \Cref{clm:algReturnsCut} can return a cut $P'$ that satisfies the guarantees. 

Otherwise, both algorithms end in scenario \ref{item:algReturnsFlow} in which case by \Cref{clm:algReturnsFlow}, we have that $\overline{P}$ is a near $\frac{\phi^2}{24}$-out-expander in $W \setminus (B^{out} \setminus B'')$ for some set $B''$ and a near $\frac{\phi^2}{24}$-out-expander in the reverse graph of $W \setminus (B^{in} \setminus B''')$ for some set $B'''$ where $B^{in} = E_W(P, \overline{P})$. It is straight-forward to verify that this implies that $\overline{P}$ is a near $\frac{\phi^2}{24}$-expander in $W \setminus (B \setminus B')$ where $B' = B'' \cup B'''$ and that $B'$ is of size at most $2z$, so we can return $B'$. This establishes the lemma.
\end{proof}

It remains to prove the two claims. Without further due, let us give their proofs.

\algReturnsFlow*

\begin{proof}
The key ingredient of this claim is a simple insight: if the at most $z$ boundary edges $Z$ which induced the excess flow would not have existed, then $f$ would be a feasible flow, certifying that $\overline{P}$ is a near-$\frac{\phi^2}{24}$ expander in the graph $W \setminus (B^{out} \setminus Z)$.

Let us now prove this more formally: we have from \Cref{rem:excess at source and path decomposition} that the excess flow of the flow problem $\Pi^{out}$ remains at the sources. Let $S$ be the set of (source) vertices that have excess flow in $\Pi^{out}$ and observe that $S \subseteq P^{out}$ by definition. 

Then, let us create a new flow problem $\Pi' = (\Delta', T_{out}, c_{out})$ where we set $\Delta'(s)$ for every vertex $s$ in $S$ to $0$ but leave everything else as in $\Pi^{out}$. Clearly, the flow $f$ is now a feasible flow for $\Pi'$ by construction. We construct $B'' = \pi^{-1}(S)$.

Finally, we prove that $\overline{P}$ is a near $\frac{\phi^2}{24}$-out-expander in $W \setminus (B^{out} \setminus B'')$ if $f'$ is feasible by contraposition. Let us therefore assume that $\overline{P}$ is not a near $\frac{\phi^2}{24}$-expander in $W \setminus (B^{out} \setminus B'')$ for any set $B'' \subseteq B^{out}$. By \Cref{def:nearOutExpander} there exists a cut $P' \subseteq \overline{P}$, such that $\vol_{W \setminus (B^{out} \setminus B'')}(P') \leq \vol_{W \setminus (B^{out} \setminus B'')}(\overline{P})/2$ and 
\begin{equation}\label{eq:smallLeavingExpander}
    \delta_{W \setminus (B^{out} \setminus B'')}^{out}(P') < \frac{\phi^2}{24} \vol_{W \setminus (B^{out} \setminus B'')}(P').
\end{equation}
However, we have by assumption of the lemma, that $\overline{P}$ is a near $\phi$-expander in $W$ and therefore we have by \Cref{def:nearExpander} that
\begin{equation}\label{eq:largeLeavingExpander}
    \delta_W^{out}(P') \geq \phi \vol_W(P').
\end{equation}
But clearly, we have
\[
\left|E(P', (V \cup P^{out}) \setminus P') \cap (B^{out} \setminus B'')\right| \geq \delta_W^{out}(P') - \delta_{W \setminus (B^{out} \setminus B'')}^{out}(P')
\]
and by the inequalities \ref{eq:smallLeavingExpander} and \ref{eq:largeLeavingExpander}, we obtain
\begin{align*}
\delta_W^{out}(P') - \delta_{W \setminus (B^{out} \setminus B'')}^{out}(P') &> \phi \vol_W(P') - \frac{\phi^2}{24} \vol_{W \setminus (B^{out} \setminus B'')}(P') \\
&\geq (1 - \frac{\phi}{24}) \phi \vol_W(P') > \phi \cdot \vol_W(P')/2.
\end{align*}
But since for each edge $e$ in $E(P', (V \cup P^{out}) \setminus P') \cap (B^{out} \setminus B'')$, there is a vertex $\pi(e) \in P^{out}$ that induces $4/\phi$ units of flow into $P'$ in the flow problem $\Pi'$, the total amount of flow that enters $P'$ is more than $2 \cdot \vol_W(P')$. However, the total sink capacity is $\vol_W(P')$ and the amount of flow that can be routed out of $P'$ in $\Pi'$ is bounded by 
\[
\sum_{e \in E(P', (V \cup P^{out}) \setminus P') \cap (B^{out} \setminus B'')} c_{out}(e) = \delta_{W \setminus (B^{out} \setminus B'')}^{out}(P') \cdot 24/\phi^2 < \vol_{W}(P')
\]
where we use equation \ref{eq:smallLeavingExpander} in the last step. Thus, we derived a contradiction since the flow $f$ cannot route all flow entering $P'$ to sinks in the flow problem $\Pi'$, but then $f$ cannot be feasible.
\end{proof}

It remains to prove the second claim.

\algReturnsCut*
\begin{proof}
Recall that the flow algorithm returns a cut $S$, with $\phi z/4\le\vol_{W^{out}}(S)\le |E(W^{out})|/2$, such that
\begin{equation}\label{eq:guaranteeCutByFlowAlg}
c(E_{W^{out}}(S,V\setminus S))\le\Delta_{out}(S)-z+c_{out}(E_{W^{out}}(S,V)\cup E_{W^{out}}(V,S))\cdot\frac{40\log n}{h}.
\end{equation}
We let $P' = S \cap \overline{P}$. Then,
\[
    \delta_{W[\overline{P}]}(P') = |E_{W[\overline{P}]}(P', \overline{P} \setminus P')| \leq |E_{W^{out}}(S, V \setminus S)| 
\]
where the inequality follows since $P' \subseteq S,\overline{P} \setminus P' \subseteq V \setminus S$ and the fact that $W$ and $W^{out}$ only differ in the boundary edges. Further, by the setup of the flow problem $\Pi^{out}$ and equation \ref{eq:guaranteeCutByFlowAlg},
\begin{align}
|E_{W^{out}}(S, V \setminus S)| &= \frac{\phi^2}{24} c(E_{W^{out}}(S, V \setminus S))\nonumber\\
&\leq \frac{\phi^2}{24}\left( \Delta_{out}(S)+c_{out}(E_{W^{out}}(S,V)\cup E_{W^{out}}(V,S))\cdot\frac{40\log n}{h}\right).\label{eq:upperBoundCapacityCut}
\end{align}
Further, 
\begin{equation}\label{eq:smallerThanVol1}
    \Delta_{out}(S) = \sum_{s \in S \cap P^{out}} 4/\phi \leq 4/\phi \cdot \vol_{W^{out}}(S)
\end{equation}
and we have
\begin{equation}\label{eq:smallerThanVol2}
    c_{out}(E_{W^{out}}(S,V)\cup E_{W^{out}}(V,S)) \leq 24/\phi^2 \cdot \vol_{W^{out}}(S).
\end{equation}
Using \ref{eq:smallerThanVol1} and \ref{eq:smallerThanVol2} in equation \ref{eq:upperBoundCapacityCut}, we obtain that
\[
|E_{W^{out}}(S, V \setminus S)| \leq \frac{\phi^2}{24}\left( 4/\phi \cdot \vol_{W^{out}}(S) + 24/\phi^2 \cdot \vol_{W^{out}}(S) \cdot\frac{40\log n}{h}\right) = \phi \vol_{W^{out}}(S)/4.
\]
This implies that a $(1-\phi)$-fraction of the edges incident to $S$ are not in the cut $(S, V \setminus S)$ and therefore for  $P' = S \cap \overline{P}$, we have $\vol_{W^{out}}(P') \geq \frac{1-\phi}{2} \vol_{W^{out}}(S) \geq \vol_{W^{out}}(S)/4$ since each edge internal to $S$ has at least one endpoint in $\overline{P}$ and therefore in $P'$. On closer inspection, it is not hard to verify that $\vol_W(P') \geq \vol_{W^{out}}(P')$ since edges in the core are not changed, and no edges are added in $W^{out}$ to the boundary but only some edges are reversed. Combined, we obtain the desired inequality
\[
\delta_{W[\overline{P}]}(P') \leq \vol_W(P').
\]
Since we have by the guarantees of the flow algorithm that $\phi z/4\le\vol_{W^{out}}(S)$, we further have that $\vol_{W^{out}}(P') > \phi z/16$. 
\end{proof}

\subsection{Dynamic Expander Pruning}
\label{subsec:DynamicExpanderPruning}

Using the sub-routine from last section, we can now give a straight-forward prove of \Cref{thm:pruning} which is restated for convenience.

\expanderPruning*

To prove the above theorem, let us start by giving an algorithm. In our algorithm, we have $2L + 3$ levels, and for each level $\ell = 0, 1, 2, \dots, 2L + 2 = L_{max}$, we maintain a set $P_{\ell} \subseteq V$ and sets $B_{\ell}, D_{\ell} \subseteq E^{0}$ (where $E^0$ is the set of edges of $W$ at stage $0$). Each of these sets is initially empty. We also have a conductance parameter $\phi_{\ell}$ associated with each level $\ell$ which we define $\phi_{\ell} = (\phi/96)^{3^{L_{max} - \ell}}$. For convenience, let us denote by $X_{\geq \ell}$ the union $\bigcup_{j \geq \ell} X_j$ where $X$ can be $P$, $B$ or $D$ and similarly for $>, \leq$ and $<$. We further assume for the rest of the section that $n^{1/L}$ is an integer.

\begin{algorithm2e}
\label{alg:deletePruning}
\caption{$\textsc{DeletePruning}(e, t)$}
\KwIn{The $t^{th}$ update to $W$, i.e. $e$ is the edge that was deleted from $W^{t-1}$ to derive $W^t$.}
\KwOut{Recomputes the sets $P_{\ell}$ to produce a new version of vertices that when pruned, leave an expander.}
\lFor{$\ell \geq 0$}{
    Add $e$ to $D_{\ell}$.\label{lne:startSetupDeletePruning}
}
Let $j$ be the largest integer such that $t$ is divisible by $n^{(j-1)/L}$.\;
\For(\label{lne:addPjAndBjLoop}){$\ell < j$}{      
    $P_{j} \gets P_j \cup P_{\ell}$  \label{lne:addPj}\;
    $B_j \gets B_j \cup B_{\ell}$ \label{lne:addBj}\;
    $P_{\ell} \gets \emptyset$; $B_{\ell} \gets \emptyset$; $D_{\ell} \gets \emptyset$; \label{lne:lastLineToDecreaseP}
}

\For(\label{lne:forLoop}){$\ell = j$ \textbf{down to} $1$}{ 
    \Repeat(\label{lne:whileLoop}){the algorithm returned a set $B'$ of edges}{
        $W_{\ell} \gets ((V \setminus P_{\geq \ell}) \cup \{s\}, E(W[V \setminus P_{\geq \ell}]) \cup \pi^{out}(s, B_{\ell} \cup D_{\ell}) \cup \pi^{in}(s, B_{\ell} \cup D_{\ell}))$.\;
        Run the algorithm from \Cref{lem:local cut} on $W_{\ell}$ with $P = \{s\}$ and $z = \max\{ 0, n^{(\ell -1)/L} - 1\}$ and $\phi_{\ell}$. \label{lne:oneShotPruningInDynPruning}\;
        \If{the algorithm returns a cut $P'$}{
            \If(\tcp*[h]{If $P'$ is out-sparse.}\label{lne:ifPOutsparse}) {$\delta^{out}_{W[V \setminus P_{\geq \ell}]}(P') \leq \phi_{\ell} \cdot \vol_{W[V \setminus P_{\geq \ell}]}(P')$}{
                $B_{\ell} \gets B_{\ell} \cup E_{W[V \setminus P_{\geq \ell}]}(P', (V \setminus P_{\geq \ell}) \setminus P')$ \label{lne:outSparse}
            }\Else(\tcp*[h]{If $P'$ is in-sparse.}){
                $B_{\ell} \gets B_{\ell} \cup E_{W[V \setminus P_{\geq \ell}]}((V \setminus P_{\geq \ell}) \setminus P', P')$ \label{lne:inSparse}
            }
            $P_{\ell} \gets P_{\ell} \cup P'$\label{lne:incrPell}\;
        }
    }
    Set $B_{\ell-1}$ to the set of edges in $B'$ after the edges with tail in $s$ where mapped by $(\pi^{in})^{-1}$ and the edges with head in $s$ where mapped by $(\pi^{out})^{-1}$ \label{if:edgesMissingToBeExpander}
}
\end{algorithm2e}

\paragraph{Algorithm.} Now, let us give a formal description. At every stage $t$ where an edge $(u,v)$ is deleted from $W$, we invoke the procedure $\textsc{DeletePruning}(e = (u,v), t)$ given in \Cref{alg:deletePruning}. In the algorithm, we first add the edge $(u,v)$ to the set $D_{\ell}$ for every $\ell \geq 0$. We then find $j$, to be the largest index such that $t$ is divisible by $n^{j/L}$. We then add for all $\ell < j$, $P_{\ell}$ to $P_{j}$ and then set every $P_{\ell} = D_{\ell} = \emptyset$. We then want to do one-shot pruning to reduce the number of edges in $B_{\ell} \cup D_{\ell}$ significantly. However, \Cref{lem:local cut} requires that the graph one-shot pruning is executed upon has all edges that are due to removal have to be in the boundary, we use a simple trick: we add a special vertex $s$ to the graph and split every edge $(u,v)$ in $B_{\ell} \cup D_{\ell}$ into two edges $(u,s)$ and $(s,v)$. We use function $\pi$ to denote this transform on a set of edges, i.e. $\pi^{out}(s, E') = \{ (u,s) | (u,v) \in E'\}$ and analogously $\pi^{in}(s, E') = \{ (s,v) | (u,v) \in E'\}$. This gives us the special graph $W_{\ell}$ of interest, defined by
\[
    W_{\ell} = \left((V \setminus P_{\geq \ell}) \cup \{s\}, E(W[V \setminus P_{\geq \ell}]) \cup \pi^{out}(s, B_{\ell} \cup D_{\ell}) \cup \pi^{in}(s, B_{\ell} \cup D_{\ell})\right).
\]
We then invoke the algorithm in \Cref{lem:local cut} on $W_{\ell}$ with boundary $\{s\}$, $\phi_{\ell-1}$ and $z =  n^{\ell/L}/8 - 1$. The algorithm then returns either a cut $P'$ in which case we add $P'$ to $P_{\ell}$ and $E_W(P_{\geq \ell}, V \setminus P_{\geq \ell})$ and in $E_W(V \setminus P_{\geq \ell},P_{\geq \ell})$ to $B_{\ell}$, update the graph $W_{\ell}$ accordingly and rerun the pruning algorithm. When the algorithm returns a set of edges $B'$, we set $B_{\ell-1}$ to $B'$ and return.

Throughout the algorithm, we maintain $P = P_{\geq 0}$.

\paragraph{Analysis.} We start the analysis by proving the following claim that establishes correctness of our algorithm.

\begin{claim} \label{clm:correctnessDynPruning}
For every $\ell \geq 0$, at any stage $t$, after the for-loop starting in \Cref{lne:forLoop} finishes iteration $\ell + 1$, the set $V \setminus P_{>\ell}$ is a near $\frac{\phi_{\ell+1}^2}{24}$-expander in $W[V \setminus P_{>\ell}] \cup B_{\ell} \cup D_{\ell}$ and remains so for the rest of the stage. Further, every invocation of the algorithm described in \Cref{lem:local cut} in \Cref{lne:oneShotPruningInDynPruning} occurs with valid parameters.
\end{claim}
\begin{proof}
Initially, we have that $W$ is a $\phi$-expander and since every set $P_{\ell}$ is empty, we have that the invariant is certainly satisfied after the initial stage.

Let us now take the inductive step. We first observe that letting $W^t$ be the graph at the current stage $t$, and $W^{t-1}$ be the graph from the previous stage, then it is clear that since we added $e$ to every $D_{\ell}$ that the invariant is still true after \Cref{lne:startSetupDeletePruning}.

Let $j$ be as chosen in \Cref{alg:deletePruning}, then we have that for all levels $\ell > j$, that the sets $P_{\ell}$ remain unaffected by the algorithm. Additionally, for every $\ell \geq j$, $B_{\ell}$ is monotonically increasing during the stage (in fact for $\ell > j$ it remains unchanged). It is not hard to see that thus the invariant for every level $\ell \geq j$ remains true. We also observe that for the first iteration of the for-loop in \Cref{lne:forLoop}, we always correctly invoke the described in \Cref{lem:local cut} with valid parameters since the invariant remains true for $j$.

For levels $0 \leq \ell < j$, observe that the relevant sets $P_{\ell}, B_{\ell}$ and $D_{\ell}$ are set to the empty set in the for-loop starting in \Cref{lne:addPjAndBjLoop}. Then, for each such level $\ell$, there is a loop iteration $\ell + 1$, where the repeat-loop leaves after certifying that $V \setminus P_{\geq \ell + 1}$ is a near $\phi_{\ell}$-expander in $W[V \setminus P_{\geq \ell + 1}] \cup B'$. The algorithm then enters the if-case in \Cref{if:edgesMissingToBeExpander} and sets $B_{\ell} = B'$ thus the above invariant is certainly satisfied for level $\ell$. The for-loop iteration for $\ell$ again only adds edges to $B_{\ell}$ so the claim remains true for the rest of the algorithm and in particular every time the \Cref{lne:oneShotPruningInDynPruning} is entered, thus the algorithm described in \Cref{lem:local cut} is invoked with valid parameters.
\end{proof}

In order to establish efficient running time, it is crucial to show that the sets $P_{\ell}$ for every level $\ell$ are sparse cuts. We therefore first prove this invariant which roughly establishes that no vertex in $P_{\ell}$ is strongly-connected to a vertex that is outside the set.

\begin{invariant} \label{inv:disconnectingEdges}
For any $i \geq 0$, at the end of any stage $t$ and after any for and repeat-loop iteration in \Cref{alg:deletePruning}, we have that 
\begin{enumerate}
    \item $P_{i} \subseteq (V \setminus P_{> i})$, and
    \item $B_{i}$ is a subset of the edges incident to at least one vertex in $P_{i}$, and
    \item there exists a partition of $P_{i}$ into sets $P^{out}_{i}$ and $P^{in}_{i}$ such that 
    \[
    E_{W[V \setminus P_{> i}] \setminus B_{i}}(P^{out}_{i}, V \setminus (P_{> i} \cup P^{out}_{i})) = E_{W[V \setminus P_{> i}] \setminus B_{i}}(V \setminus (P_{> i} \cup P^{in}_{i}), P^{in}_{i}) = \emptyset
    \]
\end{enumerate}
Additionally, after every iteration for index $i$ of the for-loop starting in \Cref{lne:addPjAndBjLoop}, we have that the sets $P_{i}$ and $B_i$ for $\ell \leq i < j$ are empty. 
\end{invariant}
\begin{proof}
Properies 1 and 2, are straight-forward to verify from the algorithm. Let us therefore focus on Property 3, which we prove by induction on the repeat-loop iterations. 

In the base case, i.e. before the first execution of algorithm $\textsc{DeletePruning}(e, t)$, we have that sets $P_{\ell}$ are initialized to the empty sets, so \Cref{inv:disconnectingEdges} is vacuously true after stage $0$. 

Let us now take the inductive step. Let us start by analyzing the for-loop starting in \Cref{lne:addPjAndBjLoop}. Let us focus on the loop iteration for $\ell = i$. Here, we have that since \Cref{inv:disconnectingEdges} was satisfied at the start of the loop, we can partition $P_j$ into $P_j^{out}$ and $P_j^{in}$ with the properties described above. Similarly, we can do the same for $P_i$ which is partitioned into $P_i^{out}$ and $P_i^{in}$. Now, let us prove that Property 3 holds for $P^{out} = P_j^{out} \cup P_i^{out}$ and $P^{in} = P_j^{in} \cup P_i^{in}$ in the graph $W[V \setminus P_{> j}] \setminus (B_j \cup B_{i})$.

Now, for the sake of contradiction, let us assume that there is some edge leaving $P^{out}$ in the graph. We certainly have that the edge cannot leave a vertex in $P_j^{out}$, since $P_j^{out}$ has no out-going edges in the graph $W[V \setminus P_{> j}] \setminus B_{j} \supseteq W[V \setminus P_{> j}] \setminus (B_{j} \cup B_i)$. Thus, the vertex with a leaving edge has to be in $P_i^{out}$. But there are no edges leaving $P_i^{out}$ in $W[V \setminus P_{> i}] \setminus B_i \subseteq W[V \setminus P_{> i}] \setminus (B_j \cup B_i)$. But since $P_{i'}$ are empty for $i < i' < j$, we have that the edge must enter a vertex in $P_j$, and in order to be in the cut, it can only be in $P_j^{in}$. But we have that $E_{W[V \setminus P_{> j}] \setminus B_{j}}(V \setminus (P_{> j} \cup P^{in}_{j}), P^{in}_{j}) = \emptyset$, thus we derive a contradiction. A similar argument establishes the claim for $P^{in}$. Thus, at the end of the for-loop, $P_j$ and $B_j$ satisfy the invariant.

To prove the second statement, we simply observe that the sets $P_{i'}$ and $B_{i'}$ where not touched for indices $i < i' < j$ and the sets for $\ell = i$ are explicitly set to the empty set in the loop iteration. 

For the for-loop starting in \Cref{lne:forLoop}, let us consider an iteration $\ell$ and take the inductive step. We have by our claim that after every repeat-loop ends, the \Cref{inv:disconnectingEdges} holds at that step. Further, we know by the first statement of the claim, that $B_{\ell - 1}$ is empty before the for-loop enters the if statement in \Cref{if:edgesMissingToBeExpander} and adding edges to $B_{\ell-1}$ can not violate the Invariant.

For the repeat-loop starting in \Cref{lne:whileLoop}, we have that every time the algorithm \Cref{lem:local cut} computes a cut $P'$, we either add all out-edges or in-edges of $P'$ in $W[V \setminus P_{\geq \ell}]$ to $B_{\ell}$ so reusing the argument an almost identical argument as for the for-loop starting in \Cref{lne:addPjAndBjLoop}, we can again obtain that the Invariant remains satisfied, even though we add $P'$ to $P_{\ell}$. This completes the proof.
\end{proof}

Next, let us prove a simple claim, that holds a useful corollary.

\begin{claim} \label{clm:smallerCutsThanWL}
Whenever the algorithm enters \Cref{lne:outSparse} then
\[
B_{\ell} \cup E_{W[V \setminus P_{\geq \ell}]}(P', (V \setminus P_{\geq \ell}) \setminus P') \subseteq B_{\ell} \cup E_{W_{\ell}[V \setminus P_{\geq \ell}]}(P', (V \setminus P_{\geq \ell}) \setminus P').
\]
An analogous claim holds for \Cref{lne:inSparse}.
\end{claim}
\begin{proof}
The cut $E_{W_{\ell}[V \setminus P_{\geq \ell}]}(P', (V \setminus P_{\geq \ell}) \setminus P')$ clearly contains all edges in the cut $E_{W[V \setminus P_{\geq \ell}]}(P', (V \setminus P_{\geq \ell}) \setminus P')$ but for the edges whose endpoints where mapped to $s$ by the functions $\pi^{out}$ and $\pi^{in}$ since the vertex $s$ is excluded in the induced graphs considered above. But this implies that 
\[
    B_{\ell} \cup D_{\ell} \subseteq \left(E_{W[V \setminus P_{\geq \ell}]}(P', (V \setminus P_{\geq \ell}) \setminus P')\right) \setminus \left(E_{W_{\ell}[V \setminus P_{\geq \ell}]}(P', (V \setminus P_{\geq \ell}) \setminus P')\right).
\]
Further, since $W$ refers to the current graph, and $D_{\ell}$ is a subset of edge deletions to the graph $W$ up to the current stage, we have that we have in fact
\[
    B_{\ell} \subseteq \left(E_{W[V \setminus P_{\geq \ell}]}(P', (V \setminus P_{\geq \ell}) \setminus P')\right) \setminus \left(E_{W_{\ell}[V \setminus P_{\geq \ell}]}(P', (V \setminus P_{\geq \ell}) \setminus P')\right).
\]
But since we consider the sets including the union with $B_{\ell}$, the claim follows.
\end{proof}
\begin{corollary} \label{cor:sparseCutInDynPruning}
We augment the set $B_{\ell}$ in \Cref{lne:outSparse} and \Cref{lne:inSparse} by at most $\phi_{\ell} \cdot \vol_{W[V \setminus P_{\geq \ell}]}(P')$ edges.
\end{corollary}
\begin{proof}
This follows straight-forwardly from the guarantee of the algorithm of \Cref{lem:local cut} combined with the insight that a selected cut in the graph $W$ is even smaller than in the graph $W_{\ell}$ that the algorithm was invoked upon by \Cref{clm:smallerCutsThanWL}.
\end{proof}

Next, let us argue about the size of the sets $B_i$ and $D_i$. We establish the following invariant.

\begin{invariant} \label{inv:DandBAreSmall}
At the end of any stage $t$, for any $i \geq 0$, we have $t' = t \mod n^{i/L}$ and $t'' =  \lfloor t' / n^{(i-1)/L} \rfloor$, we have that 
\begin{align*}
    |D_{i}| &\leq t'\\
    |B_{i}| &\leq 6^i \left(n^{i/L} - 1 + t'' \cdot n^{(i-1)/L}\right) + \phi_{i} \vol_{W[V \setminus P_{>i}] \cup D_i}(P_{i}).
\end{align*}
In particular, we have, $|D_i| < n^{i/L}$ and $|B_i| < 6^i (3 n^{i/L}) + \phi_{i+1} \vol_{W[V \setminus P_{>i}] \cup D_i}(P_{i})$.
\end{invariant}
\begin{proof}
Let us prove the invariant by induction on the stage $t$.
\begin{itemize}
    \item \uline{Base case $t=0$:} Observe that the invariant is initially satsified since all sets $D_i$ and $B_i$ are initialized to the empty set.
    \item \uline{Inductive step $t-1 \mapsto t, t>0$:} Let us conduct a case analysis for the sets $D_i$ and $B_i$ for a level $i$. We distinguish by the following cases:
    \begin{itemize}
        \item \uline{$t$ is not divisible by $n^{(i-1)/L}$:} Then, we have that $j < i$ in \Cref{alg:deletePruning}. The algorithm therefore simply increases the set $D_i$ by a single edge in \Cref{lne:startSetupDeletePruning} and no further affects any of the sets. Observe that when $t$ is not divisible by $n^{(i-1)/L}$ then, $t''$ did not change since the last stage, and therefore all remaining bounds still hold.
        
        \item \uline{$t$ is divisible by $n^{(i-1)/L}$ but not by $n^{i/L}$:}  In this case, we have that $j = i$ is chosen in \Cref{alg:deletePruning}. Observe that in this case $t'$ increases by one and as before we add a single edge to $D_i$ and leave $D_i$ untouched for the rest of the algorithm. However, $t''$ is increased by one from the last stage, so we have at the beginning of the stage 
        \[
        |B_{i}| \leq 6^i \left(n^{i/L} - 1 + (t''-1) \cdot n^{(i-1)/L} \right) + \phi_{i} \vol_{W[V \setminus P_{>i}] \cup D_i}(P_{i}) 
        \]
        by the induction hypothesis. 
        
        Next observe that in the for-loop starting in \Cref{lne:addPjAndBjLoop}, we add all $B_{\ell}$ for $\ell < j = i$, to $B_i$. However, by the induction hypothesis on the last stage and the insight that the sets $B_{\ell}$ remain unchanged until this point in the algorithm, we conclude that $B_i$ is increased by at most 
        \begin{align*}
        \sum_{\ell < i} |B_{\ell}| 
        &= \sum_{\ell < i} 6^{\ell} (3 n^{\ell/L}) + \phi_{\ell} \vol_{W[V \setminus P_{>\ell}] \cup D_\ell}(P_{\ell}) \\
        &< 6^{i} (n^{(i-1)/L}) + \sum_{\ell < i} \phi_{i} \vol_{W[V \setminus P_{>\ell}] \cup D_\ell}(P_{\ell}))
        \end{align*}
        and since all $P_{\ell}$ are disjoint from $P_i$ and pairwise disjoint, we have that after the for-loop terminates, we have that 
        \[
        |B_{i}| \leq 6^i \left(n^{i/L} - 1 + t'' \cdot n^{(i-1)/L}  \right) + \phi_{i} \vol_{W[V \setminus P_{>i}] \cup D_i}(P_{i})
        \]
        i.e. the invariant is satisfied.
        
        Finally, for the rest of the algorithm, $B_{i}$ is only changed in the first iteration of the for-loop starting in \Cref{lne:forLoop} where whenever some edges are added to $B_i$, by \Cref{cor:sparseCutInDynPruning}, $P_i$ increases significantly so that the right-hand side of the equation remains larger throughout.
        
        \item \uline{$t$ is divisible by $n^{i/L}$:} In this case, we have that since $t = n^{i/L}$, that we choose $j > i$ in \Cref{alg:deletePruning}. Thus, we enter the for-loop starting in \Cref{lne:addPjAndBjLoop} with $\ell = i$, and set $D_i, B_i$ and $P_i$ to the empty set. Since the algorithm does not revisit the set $D_i$ afterwards, the invariant follows for $D_i$. For the remaining two sets, two iterations of the for-loop starting in \Cref{lne:forLoop} are relevant: the iteration where $\ell = i+1$ and the iteration where $\ell = i$. In the former iteration, the algorithm invokes repeatedly the algorithm from \Cref{lem:local cut} and only leaves the repeat-loop once it finds a set of size $B'$ of size at most $2z$ where $z = \max\{0, n^{(i-1)/L}-1\}$. It is not hard to verify that the invariant is thus satisfied at this point. The for loop with $\ell = i$ ensures by \Cref{cor:sparseCutInDynPruning} that the invariant remains enforced.
    \end{itemize}
    This exhausts all cases, and thereby concludes the proof.
\end{itemize}
\end{proof}

Using this invariant, we can further derive a straight-forward upper bound on the size of $P_i$. 

\begin{claim}\label{clm:PRemainsSmall}
Throughout the algorithm, for any level $i$, we have 
\[
\vol_{W[V \setminus P_{>i}] \cup D_i}(P_{i}) \leq 6^{i+2}  \frac{n^{i/L} - 1}{\phi_{i}}.
\]
\end{claim}
\begin{proof}
Let us assume that, for the sake of contradiction, at some point of the algorithm, during some stage $t$, for some $i \geq 0$, we have
\[
\vol_{W[V \setminus P_{>i}] \cup D_i}(P_{i}) > 6^{i+2} \frac{(n^{i/L} - 1)}{\phi_{i}}.
\]
We observe first that $P_i$ is increased in size only in \Cref{lne:incrPell} and after the violation has occurred, the set $P_i$ is only further increased while the sets $B_i$ and $D_i$ remain unchanged.

By \Cref{inv:disconnectingEdges}, at the end of the stage, we thus have that we can find $P_i^{out}$ and $P_i^{in}$ to form a partition of $P_i$ such that 
\begin{equation}\label{eq:cutIsEmpty}
    E_{W[V \setminus P_{> i}] \setminus B_{i}}(P^{out}_{i}, V \setminus (P_{> i} \cup P^{out}_{i})) = \emptyset.
\end{equation}
Now, let us assume that $\vol_{W[V \setminus P_{>i}] \cup D_i}(P_{i}^{out}) \geq \vol_{W[V \setminus P_{>i}] \cup D_i}(P_{i}^{in})$. And further, let us observe that, at the end of each stage, by \Cref{clm:correctnessDynPruning}, the set $V \setminus P_{>i}$ is a near $\frac{\phi_{i+1}^2}{24}$-expander in $W[V \setminus P_{>i}] \cup B_{i} \cup D_{i}$. Thus, by \Cref{def:nearExpander}, we have that
\begin{equation}\label{eq:cutIsFull}
     |E_{W[V \setminus P_{>i}] \cup B_{i} \cup D_{i}}(P^{out}_i, V \setminus (P_{>i} \cup P^{out}_i))| \geq \frac{\phi_{i+1}^2}{24} \cdot \vol_{W[V \setminus P_{>i}] \cup D_i}(P_{i}^{out})  
\end{equation}
But equations \ref{eq:cutIsEmpty} and \ref{eq:cutIsFull} imply that $B_i \cup D_i$ is of size at least $\frac{\phi_{i+1}^2}{24} \cdot \vol_{W[V \setminus P_{>i}] \cup D_i}(P_{i}^{out})$.

However, by \Cref{inv:DandBAreSmall}, we have for $i = 0$, that $B_i \cup D_i$ is of size $0$ which gives a contradiction and for $i > 0$, that at the end of the stage, the size is bounded by
\[
|B_i \cup D_i| \leq 6^i 3n^{i/L} + \phi_{i} \vol_{W[V \setminus P_{>i}] \cup D_i}(P_{i}) \leq \frac{\phi_{i}}{2} \vol_{W[V \setminus P_{>i}] \cup D_i}(P_{i}) 
\]
where we use in the last inequality that $\vol_{W[V \setminus P_{>i}] \cup D_i}(P_{i}) \geq 6^{i+1} \frac{n^{i/L}}{\phi_{i}}$.

But, since $\phi_i < \phi^2_{i+1}/96$, we have that 
\[ 
|B_i \cup D_i| \leq \frac{\phi_{i}}{2} \vol_{W[V \setminus P_{>i}] \cup D_i}(P_{i}) < \frac{\phi_{i+1}^2}{48} \cdot \vol_{W[V \setminus P_{>i}] \cup D_i}(P_i) \leq \frac{\phi_{i+1}^2}{24} \cdot \vol_{W[V \setminus P_{>i}] \cup D_i}(P_i^{out}) 
\]
Thus, we have derived a contradiction on the size of the set $B_i \cup D_i$. The case where $\vol_{W[V \setminus P_{>i}] \cup D_i}(P_{i}^{out}) < \vol_{W[V \setminus P_{>i}] \cup D_i}(P_{i}^{in})$ can be established analogously.
\end{proof}

Finally, we can prove \Cref{thm:pruning} which is restated below for convenience.

\expanderPruning*
\begin{proof}
We have correctness of the algorithm, following from \Cref{clm:correctnessDynPruning} and \Cref{inv:DandBAreSmall}, where the former states that after each stage $V \setminus P_{>\ell}$ is a near $\frac{\phi_{\ell+1}^2}{24}$-expander in $W[V \setminus P_{>\ell}] \cup B_{\ell} \cup D_{\ell}$. So in particular, for $\ell = 0$, we have $V \setminus P_{>0}$ is a near $\frac{\phi_{1}^2}{24}$-expander in $W[V \setminus P_{>0}] \cup B_{0} \cup D_{0}$ where the latter states that $B_0$ and $D_0$ are empty sets. Thus, $W[V \setminus P_{>0}]$ is a $\frac{\phi_{\ell+1}^2}{24}$-expander and therefore certainly a $\phi_0$-expander.

For the running time, we observe that the invocations of the algorithm from \Cref{lem:local cut} dominate the costs of the for-loop starting in \Cref{lne:forLoop}. This follows since we can construct $W_{\ell}$ straight-forwardly from $W$ using the same running time as the algorithm from \Cref{lem:local cut} and afterwards, updating sets $P_{\ell}, B_{\ell}$ and $B_{\ell-1}$ can easily be done in the time that the algorithm requires to output these sets. The running time outside of the for-loop can be at most factor $L_{max}$ larger than the time spent in the loop (plus $m$) since we move every item in a set $P_{\ell}, B_{\ell}$ eventually to a higher level. But there are at most $L_{max}$ levels.

Therefore, let us bound the running time of the for-loop iterations. Let us fix a level $\ell$ and focus on the total time spend in the for-loop on iterations $\ell$.

We observe that the sets $P_{\ell}$ is monotonically increasing between stages that are divisible by $n^{\ell/L}$ but bounded in size by \Cref{clm:PRemainsSmall}. But every time the algorithm from \Cref{lem:local cut} runs and finds a cut $P'$, we add $\Omega(\phi_{\ell} n^{(\ell-1)/L})$ to the volume of $P_{\ell}$. Thus, there can be at most $m/n^{\ell/L} \cdot \frac{6^{\ell+2} n^{\ell/L}/\phi_{\ell}}{\Omega(\phi_{\ell}n^{(\ell-1)/L})} = O(\frac{m}{\phi_{\ell}^2 n^{(\ell -1)/L}})$ invocations of the algorithm where a cut $P'$ is reported. On the other hand, since we enter the for-loop for $\ell$ only every $n^{(\ell-1)/L)}$ iterations, there can also only be a total of $O(m/ n^{(\ell-1)/L})$ invocations ending in a set of edges $B'$ since we leave the repeat-loop once such a set is obtained.

We further observe that every invocation of the algorithm runs in time $O(|B_{\ell} \cup D_{\ell}|/\phi_{\ell})$ since there are at most two boundary edges for every edge in $|B_i \cup D_i|$. But by \Cref{inv:DandBAreSmall} and \Cref{clm:PRemainsSmall}, we have that $B_{\ell} \cup D_{\ell}$ never exceeds size $O(6^{\ell} \frac{n^{\ell/L}}{\phi_{\ell}})$. Thus, each invocation runs in time $O(\frac{m n^{1/L}}{\phi_{\ell}^4 } 6^{\ell})$. The total running time follows now straight-forwardly by summing over the levels, multiplying by factor $L_{max}$ and setting $\phi_{\ell}$.
 
Finally, to prove the claim on the volume of $P$, let $j'$ be the smallest index such that $t < n^{(j'-1)/L}$. Then, we observe that in algorithm \Cref{alg:deletePruning}, we have never chosen $j \geq j'$ in the previous or current stage. But this implies that every set $P_{j''}$ for $j'' \geq j'$ has not been changed since initialization of the algorithm. Thus, the total size of $P = P_{\geq 0}$ can be upper bound by this insight and \Cref{clm:PRemainsSmall} by
\[
    \sum_{j < j'} 6^{j+2}  \frac{n^{j/L} - 1}{\phi_{j}} \leq  O(6^{j'}  \frac{n^{j'/L}}{\phi_{0}}) = O(6^{O(L)} \frac{t \cdot n^{O(1/L)}}{\phi_0})
\]
by summing over the levels $\ell$ along with the upper bound provided by \Cref{clm:PRemainsSmall}.
\end{proof}

\section{Directed Cut-matching Game}

\label{sec:CMG}

Consider the following process between the \emph{cut player} and the
\emph{matching player}. The process starts with an empty directed
graph $W=(V,\emptyset)$ with $n$ vertices. In \emph{round} $i$
starting from $1$, the cut player chooses two disjoint sets $A_{i},B_{i}\subset V$
where $|A_{i}|=|B_{i}|\ge n/4$, then the matching player chooses
two directed (fractional) perfect matchings $\Mto_{i}$ and $\Mback_{i}$
that match vertices from $A_{i}$ to $B_{i}$ and back. Then, we set
$W\gets W\cup\Mto_{i}\cup\Mback_{i}$ and proceed with round $i+1$.
We call this process a \emph{cut-matching game}.

For any number $d\ge1$, we say that an edge is $1/d$-integral if
its weight is a non-negative multiple of $1/d$. A fractional matching
or a graph is $1/d$-integral if it consists of only $1/d$-integral
edges.
\begin{thm}
[Deterministic Cut-matching Game for Directed Graphs]\label{thm:CMG}Suppose
that, for every $i$, $\Mto_{i}$ and $\Mback_{i}$ are $1/d$-integral
for some integer $d\ge1$. There is a deterministic algorithm for
the cut player that takes $\Ohat(nd)$ time to output each $(A_{i},B_{i})$
in the cut-matching game such that after $R=O(\log n)$ rounds $W=(\Mto_{1}\cup\Mback_{1})\cup\dots\cup(\Mto_{R}\cup\Mback_{R})$
must be a $\alphacmg$-expander, where $\alphacmg=1/n^{o(1)}$ is
a parameter we will refer to it other parts of the paper. Moreover,
the weighted in-degree and out-degree of each vertex in $W$ is at
least $1$. 
\end{thm}

\Cref{thm:CMG} is proved by extending the fast deterministic cut-matching game in undirected graphs by Chuzhoy et.~al~\cite{ChuzhoyGLNPS_det_cut}. The proof is not too hard because the most technique in \cite{ChuzhoyGLNPS_det_cut} can be generalized to directed graphs. The only crucial new ingredient is in the analysis about entropy function.

We review the previous work on the cut-matching game below. The framework was first introduced by Khandekar, Rao and Vazirani
\cite{KhandekarRV09} and has been used in numerous algorithms for
computing sparse cuts \cite{KhandekarRV09,NanongkaiS17,SaranurakW19,GaoLNPSY19}
and beyond (e.g.~\cite{ChekuriC13,RackeST14,ChekuriC16,ChuzhoyL16}).
There is also a line of works which focuses on the quality
of the cut-matching game itself (i.e.~the guarantee of the cut player)
and describe our contribution. For simplicity, we assume that $d=1$. 
\begin{itemize}
\item \textbf{(Undirected Matching Player, Randomized Cut Player): }The
first work is by Khandekar, Rao and Vazirani \cite{KhandekarRV09}.
They require the matching player to choose an \emph{undirected} perfect
matching $M_{i}$ at each round $i$. Then, they show a \emph{randomized
}algorithm for the cut player that takes $O(n\log^{2}n)$ time in
each round $i$ to output $(A_{i},B_{i})$ and guarantees that after
$R=O(\log^{2}n)$ rounds, $\Psi(W)\ge\Omega(1)$ and so $\Phi(W)\ge\Omega(1/\log^{2}n)$.
Then, Orecchia et.~al~\cite{OrecchiaSVV08} show a slower randomized
algorithm which takes $\tilde{O}(n)$ time per round but after $R=O(\log^{2}n)$
rounds, they improve the sparsity guarantee to $\Psi(W)\ge\Omega(\log n)$. 
\item \textbf{(Directed Matching Player, Randomized Cut Player): }Louis
\cite{Louis10} generalizes the result by \cite{KhandekarRV09} and
shows that even when the matching players give two directed perfect
matchings $\Mto_{i}$ and $\Mback_{i}$, there is a randomized algorithm
for the cut player with same guarantee as in \cite{KhandekarRV09}.
As every undirected matching $M_{i}$ can be thought as two directed
matchings $\Mto_{i}$ and $\Mback_{i}$ such that $(u,v)\in\Mto_{i}$
iff $(v,u)\in\Mback_{i}$, this setting of directed matchings is a
strict generalization.
\item \textbf{(Undirected Matching Player, Deterministic Cut Player): }In
the attempt to reduce the number of $O(\log^{2}n)$ rounds, Khandekar~et.~al~\cite{KhandekarKOV2007cut}
show that, when the matching player chooses an \emph{undirected} perfect
matching $M_{i}$ at each round $i$, there is a \emph{deterministic
exponential-time} algorithm for the cut player (by simply finding
a sparsest cut in $W_{i-1}$). Then, after $R=O(\log n)$ rounds,
they guarantee $\Psi(W)\ge\Omega(1)$. The novel component of this
work is the potential analysis based on \emph{entropy}. Later, it
is observed in \cite{GaoLNPSY19} that finding approximate sparsest
cuts also works: they show a \emph{deterministic }$\tilde{O}(n^{2})$-time
algorithm for the cut player where $\Psi(W)\ge1/\log^{O(1)}n$ after
$R=O(\log n)$ rounds. Finally, Chuzhoy~et.~al \cite{ChuzhoyGLNPS_det_cut}
give a deterministic $\Ohat(n)$-time algorithm for the cut player
where $\Psi(W)\ge1/n^{o(1)}$ after $R=O(\log n)$ rounds. This in
turns imply a wide range of applications in undirected graphs. We
note that both \cite{GaoLNPSY19,ChuzhoyGLNPS_det_cut} use the same
potential analysis based on entropy. 
\end{itemize}
We can see that, in contrast to \Cref{thm:CMG}, all previous cut-player
algorithms either are randomized, or require undirected matchings,
or both. We describe our cut-player algorithm in \Cref{sec:cut player algorithm}.
The idea for proving \Cref{thm:CMG} is by generalizing two components
of the previous works, and then combining the two. 

First, we generalize the deterministic $\Ohat(n)$-time implementation
of Chuzhoy~et.~al~\cite{ChuzhoyGLNPS_det_cut} for the cut player
to work in directed graphs. Although the result in \cite{ChuzhoyGLNPS_det_cut}
was stated for undirected graphs, most of the tools from \cite{ChuzhoyGLNPS_det_cut}
readily generalizes to directed graphs. We sketch how to do this in
\Cref{sec:cut implementation}. 

Second, we generalize the potential analysis based on entropy by Khandekar~et.~al~\cite{KhandekarKOV2007cut}
to work with directed matchings. Although the idea is similar, our
analysis is more involved. At a very high level, the reason is that,
while each undirected matching $M_{i}$ can be viewed as a collection
of directed cycles of length 2 (and hence a directed calculation by
hand is possible), the union two directed matchings of $\Mto_{i}\cup\Mback_{i}$
can be a collection of directed cycles of \emph{arbitrary length}.
The detail of our analysis is shown in \Cref{sec:CMG round}.

\paragraph{Preliminaries about Sparsity of Cuts.}

In this section, it is more convenient to work with the notion of
\emph{sparsity} instead of \emph{conductance}. Sparsity measures expansion
of a cut like conductance but, for sparsity, we compare the cut size
to the number of vertices in the cut.
\begin{defn}
[Sparsity]A directed weighted graph $G=(V,E)$ has sparsity $\Psi(G)\ge\psi$
if, for any set $S\subset V$ where $|S|\le|V\setminus S|$, $\min\{\delta^{in}(S),\delta^{out}(S)\}\ge\psi|S|$.
The sparsity of a cut $(S,V\setminus S)$ is $\Psi_{G}(S)=\min\{\delta^{in}(S),\delta^{out}(S)\}/\min\{|S|,|V\setminus S|\}$.
\end{defn}

Note that, in the graph with maximum weighted degree $d$, we have
$\Phi(G)\le\Psi(G)\le d\cdot\Phi(G)$. Also, $\Psi(H)\le\Psi(G)$
for any subgraph $H$ of $G$.

\subsection{The Cut Player Algorithm}

\label{sec:cut player algorithm}

To describe the algorithm of the cut player for \Cref{thm:CMG}, we
need the following subroutine:
\begin{thm}
\label{thm: directed cut player}There is a deterministic algorithm,
that we call $\cutorcert$, that, given a directed $n$-vertex $1/d$-integral
graph $G=(V,E)$ and maximum weighted degree $O(\log n)$, returns
one of the following: 
\begin{itemize}
\item either a cut $(A,B)$ in $G$ such that $|A|,|B|\ge n/10$ and $w(E_{G}(A,B))\le n/100$;
or 
\item a subset $S\subset V$ of at least $n/2$ vertices and $\Psi(G[S])\ge1/\gamma$. 
\end{itemize}
The running time of the algorithm is $O\left(nd\gamma\right)$ where
$\gamma=n^{o(1)}$.
\end{thm}

As this subroutine is the generalization of Theorem 1.5 of \cite{ChuzhoyGLNPS_det_cut}
to directed weighted graphs and almost all tools are readily generalized,
we only sketch the proof for completeness in \Cref{sec:cut implementation}.

Now, we describe the algorithm of the cut player for \Cref{thm:CMG}
which is a generalization of the algorithm in \cite{KhandekarKOV2007cut}
to directed graphs. Initialize $W_{0}=\emptyset$ as an $n$-vertex
empty graph. Starting from $i=1$. While the algorithm $\cutorcert$
running on $W_{i-1}$ returns the cut $(A,B)$ where $w(E_{W_{i-1}}(A,B))\le n/100$
and $|A|,|B|\ge n/10$, we do the following. Let $A_{i}$ and $B_{i}$
be arbitrary subsets where $|A_{i}|=|B_{i}|\ge n/4$ and $(A_{i},B_{i})$
does not cross $(A,B)$ (i.e.~either $A\subseteq A_{i}$ or $B\subseteq B_{i}$).
Then, the matching player gives us two directed $1/d$-integral perfect
matchings $\Mto_{i}$ and $\Mback_{i}$ that matches vertices from
$A_{i}$ to $B_{i}$ and back. Then, $W_{i}\gets W_{i-1}\cup\Mto_{i}\cup\Mback_{i}$.
This finishes the round $i$. Then, we set $i\gets i+1$. 

Otherwise, $\cutorcert$ returns a subset $S\subseteq V$ of at least
$n/2$ vertices, such that $\Psi(W_{i-1}[S])\geq1/\gamma$. Now, we
call the last round. Let $T\subseteq V\setminus S$ be an arbitrary
set where $|T|=|S|$. The cut player chooses $A_{i}$ and $B_{i}$
by setting $(A_{i},B_{i})\gets(S,T)$. Then, the matching player again
gives us the perfect matchings $\Mto_{i}$ and $\Mback_{i}$. Finally,
set $W_{i}\gets W_{i-1}\cup\Mto_{i}\cup\Mback_{i}$ and terminate.
Let $W=W_{i}$ denote the graph after the last iteration.

Now, we are ready to prove \Cref{thm:CMG}. First, we bound the number
of rounds:
\begin{lem}
\label{lem:bound iter}There are at most $O(\log n)$ rounds in the
above process.
\end{lem}

The proof of \Cref{lem:bound iter} is the main contribution of this
section and is shown later in \Cref{sec:CMG round}. Next, we claim
that after the process is terminated, then $\Psi(W)\ge\Omega(1/\gamma)$.
This follows because $\Psi(W)\ge\Psi(W_{i-1}[S]\cup\Mto_{i}\cup\Mback_{i})\ge\Omega(1/\gamma)$
where the last inequality is by the following observation (which is
a generalization of Observation 2.3 in \cite{ChuzhoyGLNPS_det_cut}):
\begin{prop}
\label{obs: exp plus matching is exp}Let $G=(V,E)$ be an $n$-vertex
(weighted) graph where $\Psi(G)\ge\psi$, and let $G'$ be another
graph that is obtained from $G$ by adding to it a new set $V'$ of
at most $n$ vertices, and two perfect (fractional) matching $\Mto$
and $\Mback$, matching vertices from $V'$ to another set $V''\subseteq V$
and vice versa where $|V''|=|V'|$. Then $\Psi(G') = \Omega(\psi)$. 
\end{prop}

As the weighted degree of each vertex in $W$ is at most $O(\log n)$,
we have that $\Phi(W)\ge\Psi(W)/O(\log n)\ge\Omega(1/\gamma\log n)=1/n^{o(1)}$.
Observe further that the weighted in-degree and out-degree of each
vertex in $W$ is at least $1$. To see this, consider $W_{i-1}$
before the last round. Observe that weighted in-degree and out-degree
of each vertex is integral, because $W_{i-1}$ is a union of perfect
matchings. However, if a vertex $u$ has either zero in-degree or out-degree
in $W_{i-1}$, then $u$ can not be in the set $S$ where $\Psi(W_{i-1}[S])\geq1/\gamma$.
But, the perfect matching $\Mto_{i}$ and $\Mback_{i}$ in the last
round must contribute exactly 1 to both the weighted in-degree and
out-degree of $u$. 

Therefore, we conclude that, in each round, the cut player takes $O\left(nd\gamma\right)=\Ohat(nd)$
time. After $O(\log n)$ rounds, $W$ is a $1/n^{o(1)}$-expander
and each vertex in $W$ has weighted in-degree or out-degree at least
$1$. This completes the proof of \Cref{thm:CMG}.

\subsection{Bounding the Number of Rounds}

\label{sec:CMG round}

We prove \Cref{lem:bound iter} in this section. Consider the following
process. Initially, each vertex $u$ has a unit of mass initialized
at $u$ itself. 

At round $i$, we are given the $1/d$-integral perfect matchings
$\Mto_{i}$ and $\Mback_{i}$. Observe that $\Mto_{i}$ is the average
of exactly $d$ integral perfect matching $\Mto_{i,1},\dots,\Mto_{i,d}$.
Similarly, $\Mback_{i}$ is the average of $\Mback_{i,1},\dots,\Mback_{i,d}$.
Let $D_{i}$ be a uniformly random number from $\{1,\dots,d\}$. The
mass on each vertex is distributed as follows: 
\begin{itemize}
\item For each $u\in A\cup B$, $1/2$-fraction of the mass at $u$ stays
at $u$ and $1/2$-fraction of the mass from $u$ is sent to $v$
where $(u,v)$ the unique outgoing edge of $u$ in $\Mto_{i,D_{i}}\cup\Mback_{i,D_{i}}$.
\item For each $u\in V\setminus(A\cup B)$, all of the mass at $u$ stays
at $u$.
\end{itemize}

Observe that, at round $i$, the mass is moved only between $A_{i}$
and $B_{i}$ and there are exactly 1 unit of mass on every vertex
after each round. Let $p_{i}(u,v)$ denote the expected mass that
starts from $u$ and ends at $v$ after the $i$-th round. From the
above process, we have that $p_{0}(u,u)=1$ for all $u\in V$ and
$p_{0}(u,v)=0$ for all $u\neq v$. Observe that $0\le p_{i}(u,v)\le1$
for all $u,v,i$, and $\sum_{v\in V}p_{i}(u,v)=1$, $\sum_{u\in V}p_{i}(u,v)=1$. 

Let $\Pto_{i}(u)$ denote the random variable where $\Pr[\Pto_{i}(u)=v]=p_{i}(u,v)$
for all $v\in V$, i.e., the distribution of $\Pto_{i}(u)$ is the
distribution of mass starting from $u$ after the $i$-th round. Similarly,
let $\Pback_{i}(v)$ denote the random variable where $\Pr[\Pback_{i}(v)=u]=p_{i}(u,v)$
for all $v\in V$. That is, the distribution of $\Pback_{i}(v)$ is
the distribution of mass of each vertex that ends at $v$ after the
$i$-th round. For any distribution $X=(x_{1},\dots,x_{n})$ where
$p(x)=\Pr[X=x]$, the \emph{entropy} of $X$ is $H(X)=\sum_{x}p(x)\log\frac{1}{p(x)}.$ 

The \emph{potential} after round $i$ is defined as 
\[
\Phi_{i}=\sum_{u\in V}H(\Pto_{i}(u))+H(\Pback_{i}(u)).
\]

From the definition of entropy, observe the following simple fact:
\begin{prop}
$\Phi_{0}=0$ and $\Phi_{i}\le O(n\log n)$ for all $i$. 
\end{prop}

Our main goal is to show that after each round $i$, we have $\Phi_{i}\ge\Phi_{i-1}+\Omega(n)$.
So there can be only $O(\log n)$ rounds. We will show that this is
true even if $D_{i}$ is fixed. We formalize this below. Let $Z$
be a random variable. The \emph{entropy of $X$ conditioned on the
value of $Z=z$} is defined as $H(X\mid Z=z)=\sum_{x}p(x\mid z)\log\frac{1}{p(x\mid z)}.$
It is well-known that fixing some random variable never increases
the entropy: 
\begin{fact}
$H(X\mid Z=z)\le H(X)$
\end{fact}

Let $\Phi_{i,z}=\sum_{u\in V}H(\Pto_{i}(u)\mid D_{i}=z)+H(\Pback_{i}(u)\mid D_{i}=z).$
As $\Phi_{i,z}\le\Phi_{i}$ by the above fact, we can bound the number
of rounds to be $O(\log n)$, proving \Cref{lem:bound iter}, once
we can prove the following: 
\begin{lem}
\label{lem: potential increase}$\Phi_{i,z}\ge\Phi_{i-1}+\Omega(n)$
for any $z\in[d]$. %

\end{lem}

As our goal is to lower bound $\Phi_{i,z}$ for every $z$, from now
on, we will assume that $D_{i}=z$ is fixed for some $z$. For notational
convenience, below we will assume $\Mto_{i}=\Mto_{i,z}$ and $\Mback_{i}=\Mback_{i,z}$
and avoid writing ``given $D_{i}=z$'' in the expressions. As $i$
will be fixed below, we also write $p_{i-i},\Pto_{i-1},\Pback_{i-1}$
as $p,\Pto,\Pback$ respectively, and write $p_{i},\Pto_{i},\Pback_{i}$
as $p',\Pto',\Pback'$ respectively. For any sets $S,T\subseteq V$,
we define $p(S,T)=\sum_{u\in S,v\in T}p(u,v)$ and $p'(S,T)$ is similarly
defined.

As $\Mto_{i}$ and $\Mback_{i}$ are now assumed to be integral, $\Mto_{i}\cup\Mback_{i}$
forms a collection $\cset$ of disjoint directed cycles that partition
$A_{i}\cup B_{i}$. Indices of vertices in each cycle $C=(c_{1},\dots,c_{|C|})\in\cset$
are such that $c_{1},c_{3},c_{5},\dots,c_{|C|-1}\in A_{i}$ and $c_{2},c_{4},c_{4},\dots,c_{|C|}\in B_{i}$.
In particular, $|C|$ is even. How the mass moves in at round $i$
can be described as follows: for every $C=(c_{1},\dots,c_{|C|})\in\cset$,
$u\in V$, and $1\le j\le|C|$ 
\[
p'(u,c_{j})=\frac{p(u,c_{j})+p(u,c_{j-1})}{2}
\]
where we define $c_{0}=c_{|C|}$. Observe that $p'(u,C)=p(u,C)$.
First, we show the entropy never decreases.
\begin{lem}
\label{lem:ent non_dec}For all $u\in V$, $H(\Pto'(u))\ge H(\Pto(u))$
and $H(\Pback'(u))\ge H(\Pback(u))$.
\end{lem}

\begin{proof}
We will prove that $H(\Pto'(u))\ge H(\Pto(u))$ for all $u$. The
proof for $H(\Pback'(u))\ge H(\Pback(u))$ is symmetric. 

Fix $u$ from now. For each cycle $C\in\cset$, let $H_{C}(\Pto(u))=\sum_{v\in C}p(u,v)\log\frac{1}{p(u,v)}$
be the sum of the terms in $H(\Pto(u))$ restricted to only vertices
in $C$. Similarly, we let $H_{C}(\Pto'(u))=\sum_{v\in C}p'(u,v)\log\frac{1}{p'(u,v)}$.
It suffices to show that $H_{C}(\Pto'(u))\ge H_{C}(\Pto(u))$ for
each $C\in\cset$. Fix $C$ from now. Recall the binary entropy function
$h(x)=x\log\frac{1}{x}+(1-x)\log\frac{1}{(1-x)}$ where $h:[0,1]\rightarrow[0,1]$.
Let us verify the following equality:
\begin{claim}
$H_{C}(\Pto'(u))+\sum_{j=1}^{|C|}p'(u,c_{j})\cdot h(\frac{p(u,c_{j})}{p(u,c_{j})+p(u,c_{j-1})})=H_{C}(\Pto(u))+p(u,C)$
\end{claim}

\begin{proof}
We have
\begin{align*}
 & H_{C}(\Pto'(u))+\sum_{j=1}^{|C|}\frac{p(u,c_{j})+p(u,c_{j-1})}{2}\cdot h(\frac{p(u,c_{j})}{p(u,c_{j})+p(u,c_{j-1})})\\
= & H_{C}(\Pto'(u))+\sum_{j=1}^{|C|}\left(\frac{p(u,c_{j})}{2}\log\frac{p(u,c_{j})+p(u,c_{j-1})}{p(u,c_{j})}+\frac{p(u,c_{j-1})}{2}\log\frac{p(u,c_{j})+p(u,c_{j-1})}{p(u,c_{j})}\right)\\
= & \sum_{j=1}^{|C|}\frac{p(u,c_{j})}{2}\log\frac{2}{p(u,c_{j})}+\sum_{j=1}^{|C|}\frac{p(u,c_{j-1})}{2}\log\frac{2}{p(u,c_{j-1})}\\
= & \sum_{j=1}^{|C|}p(u,c_{j})(\log\frac{1}{p(u,c_{j})}+1)\\
= & H_{C}(\Pto(u))+p(u,C)
\end{align*}
\end{proof}
So, it remains to show that $\sum_{j=1}^{|C|}p'(u,c_{j})\cdot h(\frac{p(u,c_{j})}{p(u,c_{j})+p(u,c_{j-1})})\le p(u,C)$.
To show this, let $Y$ be random variable where 
\[
\Pr[Y=\frac{p(u,c_{j})}{p(u,c_{j})+p(u,c_{j-1})}]=p'(u,c_{j})/p'(u,C)
\]
Observe that $E(h(Y))=\sum_{j=1}^{|C|}\frac{p'(u,c_{j})}{p'(u,C)}\cdot h(\frac{p(u,c_{j})}{p(u,c_{j})+p(u,c_{j-1})})$
and $E(Y)=\sum_{j=1}^{|C|}\frac{p'(u,c_{j})}{p'(u,C)}\cdot\frac{p(u,c_{j})}{p(u,c_{j})+p(u,c_{j-1})}=1/2$.
By Jensen's inequality, we have $E(h(Y))\le h(E(Y))=h(1/2)=1$. So
$\sum_{j=1}^{|C|}p'(u,c_{j})\cdot h(\frac{p(u,c_{j})}{p(u,c_{j})+p(u,c_{j-1})})\le p'(u,C)=p(u,C)$
as desired. This completes the proof of \Cref{lem:ent non_dec}.
\end{proof}
\Cref{lem:ent non_dec} already implies that $\Phi_{i,z}\ge\Phi_{i-1}$.
Next, to show that the potential increase is $\Omega(n)$, we need
to exploit the fact that the cut $(A,B)$ is a sparse cut. More precisely,
let $(A,B)$ be a cut of $W_{i-1}$ returned by \Cref{thm: directed cut player}
where $w(E_{W_{i-1}}(A,B))\le n/100$ and $|A|,|B|\ge n/10$. Recall
that we choose $A_{i}$ and $B_{i}$ where $|A_{i}|=|B_{i}|\ge n/4$
and $(A_{i},B_{i})$ does not cross $(A,B)$. 

Suppose that $|A|\le|B|$. We will show that $\sum_{u\in V}H(\Pto_{i}(u))\ge\sum_{u\in V}H(\Pto_{i-1}(u))+\Omega(n)$.
If $|A|\ge|B|$, we can show that $\sum_{u\in V}H(\Pback_{i}(u))\ge\sum_{u\in V}H(\Pback_{i-1}(u))+\Omega(n)$
by symmetry. So we will assume $|A|\le|B|$ from now.

As $|A|\le|B|$, we can choose $(A_{i},B_{i})$ such that $A\subseteq A_{i}$
and $B_{i}\subseteq B$. Observe that each $1/d$-integral edge $e\in W_{i-1}$
has mass going through it exactly once with amount $1/2d=w(e)/2$.
As $w(E_{W_{i-1}}(A,B))\le n/100$, we have $p(A,B)\le n/200$. As
$B_{i}\subseteq B$, we have $p(A,B_{i})\le n/200\le|A|/20$. By averaging
argument, there at least $|A|/2\ge n/20$ vertices $u\in A$ such
that $p(u,B_{i})\le1/10$ (otherwise, $p(A,B_{i})>\frac{|A|}{2}\cdot\frac{1}{10}$
which is a contradiction). We call these vertices in $A$ \emph{interesting}
vertices. Note that, for each interesting $u\in A$, we have $p(u,A_{i})=p(u,V)-p(u,B_{i})>9/10$. 

Fix an interesting vertex $u$. Consider the collection $\cset$ of
cycles forming by $\Mto_{i}\cup\Mback_{i}$. We say that a cycle $C\in\cset$
is \emph{good} (w.r.t.~$u$) if $p(u,A_{i}\cap C)\ge2p(u,B_{i}\cap C)$.
Observe the following:
\begin{prop}
\label{prop:good mass for interesting}For every interesting vertex
$u\in A$, $\sum_{C:good}p(u,A_{i}\cap C)\ge1/2$.
\end{prop}

\begin{proof}
For each $v\in A_{i}$, there is a unique cycle from $\cset$ containing
$v$. So $\sum_{C\in\cset}p(u,A_{i}\cap C)=p(u,A_{i})>9/10$. Assume
for contradiction that $\sum_{C:good}p(u,A_{i}\cap C)<1/2$. Then,
we have 
\begin{align*}
p(u,B_{i}) & \ge\sum_{C:bad}p(u,B_{i}\cap C)\\
 & >\sum_{C:bad}p(u,A_{i}\cap C)/2\\
 & >(9/10-1/2)/2=2/10.
\end{align*}
But $u$ is interesting, so $p(u,B_{i})\le1/10$, which is a contradiction.
\end{proof}
\begin{lem}
\label{lem:ent inc}For every interesting vertex $u\in A$ and good
cycle $C$ w.r.t.~$u$, $H_{C}(\Pto'(u))\ge H_{C}(\Pto(u))+\Omega(p(u,C))$. 
\end{lem}

\begin{proof}
The proof is the extension of \Cref{lem:ent non_dec}. Let $C=(c_{1},\dots,c_{|C|})$.
Recall that $H_{C}(\Pto'(u))+\sum_{j=1}^{|C|}p'(u,c_{j})\cdot h(\frac{p(u,c_{j})}{p(u,c_{j})+p(u,c_{j-1})})=H_{C}(\Pto(u))+p(u,C)$.
It suffices to prove that $\sum_{j=1}^{|C|}p'(u,c_{j})\cdot h(\frac{p(u,c_{j})}{p(u,c_{j})+p(u,c_{j-1})})\le(1-\Omega(1))p(u,C)$. 

Let $Z$ be random variable that is similarly defined as the random
variable $Y$ from \Cref{lem:ent non_dec}. For odd $1\le j\le|C|$,
we set 
\[
\Pr[Z=\frac{p(u,c_{j})}{p(u,c_{j})+p(u,c_{j-1})}]=p'(u,c_{j})/p'(u,C)
\]
and, for even $1\le j\le|C|$, we set
\[
\Pr[Z=\frac{p(u,c_{j-1})}{p(u,c_{j})+p(u,c_{j-1})}]=p'(u,c_{j})/p'(u,C).
\]
Observe that $E(Z)=\sum_{j:odd}p(u,c_{j})/p'(u,C)=p(u,A_{i}\cap C)/p(u,C)\ge2/3$
because $C$ is good. Recall from \Cref{lem:ent non_dec} that $\frac{1}{p(u,C)}\sum_{j=1}^{|C|}p'(u,c_{j})\cdot h(\frac{p(u,c_{j})}{p(u,c_{j})+p(u,c_{j-1})})=E(h(Y))$.
However, as $h(\frac{p(u,c_{j})}{p(u,c_{j})+p(u,c_{j-1})})=h(\frac{p(u,c_{j})}{p(u,c_{j})+p(u,c_{j-1})})$
for any $j$, so we have that $E(h(Y))=E(h(Z))$. By Jensen's inequality,
we have $E(h(Z))\le h(E(Z))\le h(2/3)\le1-\Omega(1)$. Therefore,
we conclude that 
\[
\frac{1}{p(u,C)}\sum_{j=1}^{|C|}p'(u,c_{j})\cdot h(\frac{p(u,c_{j})}{p(u,c_{j})+p(u,c_{j-1})})\le1-\Omega(1).
\]
This completes the proof of \Cref{lem:ent inc}.
\end{proof}
Finally, we summarize the argument above and prove \Cref{lem: potential increase}.
Recall that we assume that the cut $(A,B)$ on $W_{i-1}$ found by
\Cref{thm: directed cut player} is such that $|A|\le|B|$. Then, we
have shown that there are $n/20$ interesting vertices. For each interesting
vertex $u\in A$, combining \Cref{prop:good mass for interesting}
and \Cref{lem:ent inc}, we have 
\[
H(\Pto'(u))\ge H(\Pto(u))+\sum_{C:good}\Omega(p(u,C))=H(\Pto(u))+\Omega(1).
\]
As $H(\Pto'(u))\ge H(\Pto(u))$ and $H(\Pback'(u))\ge H(\Pback(u))$
for all $u\in V$ by \Cref{lem:ent non_dec}. We have $\Phi_{i,z}\ge\Phi_{i-1}+\frac{n}{20}\cdot\Omega(1)$. 

If $|A|\ge|B|$, the proof is symmetric. We choose $(A_{i},B_{i})$
such that $A_{i}\subseteq A$ and so $p(A_{i},B)\le n/200\le|B|/20$.
We say that a vertex $u\in B$ is interesting if $p(A_{i},u)\le1/10$.
There must be at least $|B|/2\ge n/20$ interesting vertices using
the same agrument. We say that a cycle $C\in\cset$ is \emph{good}
(w.r.t.~$u$) if $p(B_{i}\cap C,u)\ge2p(A_{i}\cap C,u)$ and can
prove that $\sum_{C:good}p(B_{i}\cap C,u)\ge1/2$ for every interesting
$u$. We also have $H_{C}(\Pback'(u))\ge H_{C}(\Pback(u))+\Omega(p(C,u))$.
All these imply that $\Phi_{i,z}\ge\Phi_{i-1}+\frac{n}{20}\cdot\Omega(1)$
as well. This completes the proof of \Cref{lem: potential increase},
which in turn proves \Cref{lem:bound iter}.

\subsection{Implementation of $\protect\cutorcert$ in Directed Graphs}

\label{sec:cut implementation}

In this section, we sketch the proof of \Cref{thm: directed cut player}.
First, we state the version of \Cref{thm: directed cut player} for
only \emph{unweighted} graphs.
\begin{thm}
\label{thm: directed cut player unweight}There is a deterministic
algorithm that, given a directed $n$-vertex
unweighted graph $G=(V,E)$ and maximum weighted degree $O(\log n)$,
returns one of the following: 
\begin{itemize}
\item either a cut $(A,B)$ in $G$ such that $|A|,|B|\ge n/4$ and $|E_{G}(A,B)|\le n/1000$;
or 
\item a subset $S\subset V$ of at least $n/2$ vertices and $\Psi(G[S])\ge1/\gamma$. 
\end{itemize}
The running time of the algorithm is $O\left(n\gamma\right)$ where
$\gamma=n^{o(1)}$.
\end{thm}

\paragraph{\Cref{thm: directed cut player} follows from \Cref{thm: directed cut player unweight}.}

Given \Cref{thm: directed cut player unweight} above, the proof of
\Cref{thm: directed cut player} is quite straightforward. There are
two steps: (1) making the graph unweighted, (2) reducing the maximum
degree. 

For the first step, as the input graph $G$ of \Cref{thm: directed cut player}
is $1/d$-integral, we can scale up all $1/d$-integral edges to unweighted
edges. Let $G'$ denote the resulting graph. As the weighted minimum
and maximum in-degree/out-degree in $G$ is $1$ and $O(\log n)$
respectively, $G'$ has $O(nd\log n)$ unweighted edges and has minimum
and maximum in-degree/out-degree $d$ and $O(d\log n)$ respectively. 

For the second step, we apply the standard ``degree reduction''
technique. (See Section 5.2 of \cite{ChuzhoyGLNPS_det_cut}) to $G'$
and obtain $G''$. The idea to obtain $G''$ is to replace each vertex
in $G$ by a constant-degree expander with $O(d\log n)$ vertices.
It is easy to show that, when we compute call \Cref{thm: directed cut player unweight}
on $G''$, we can obtain a corresponding cut in $G'$ as an output
of \Cref{thm: directed cut player} with the same balanced and sparsity
in linear time. This argument is formally shown in Lemma 5.4 of \cite{ChuzhoyGLNPS_det_cut}.
Although the proof was for undirected graphs, the proof generalizes
seamlessly to directed graphs. 

\paragraph{Proof of \Cref{thm: directed cut player unweight}. }

\Cref{thm: directed cut player unweight} is exactly the directed-graph
version of Theorem 1.5 from \cite{ChuzhoyGLNPS_det_cut}. The proof
of Theorem 1.5 of \cite{ChuzhoyGLNPS_det_cut} only needs the techniques
from Section3 and 4 in \cite{ChuzhoyGLNPS_det_cut}, and not any other
sections. Below, we sketch the idea how to modify such ideas from
\cite{ChuzhoyGLNPS_det_cut} in Sections 3 and 4. The modification
is as follows:
\begin{itemize}
\item Section 3 of \cite{ChuzhoyGLNPS_det_cut} describes algorithms that,
given a set of vertices $A_{1},\dots A_{k}$ and $B_{1},\dots B_{k}$
where $|A_{i}|=|B_{i}|$, either compute an embedding of matchings
between $A_{i}$ and $B_{i}$ for all $i$ with some small number
of fake edges, or return a balanced sparse cut. Their first algorithm
is based on Even-Shiloach tree and their second algorithm is based
on push-relabel flow algorithm. As both algorithms readily work on
directed graphs, the statement of their result in Section 3 can be
generalized to directed graphs without technical modification.
\item Section 4 of \cite{ChuzhoyGLNPS_det_cut} describes a recursive algorithm
for the undirected version of \Cref{thm: directed cut player unweight}.
We need three simple modifications. First, they employ the undirected
expander pruning from \cite{SaranurakW19} to identify the large vertex
set $S$ where $\Psi(G[S])\ge1/n^{o(1)}$. We can replace this subroutine
in a black-box manner with our directed expander pruning from \Cref{thm:pruning}
(when all the edge deletions are even given in one batch). As the
quality and running time of \Cref{thm:pruning} directed graphs is
only $n^{o(1)}$ factor worse than the algorithm of \cite{SaranurakW19}
for undirected, this only affects our final guarantee in \Cref{thm: directed cut player unweight}
by $n^{o(1)}$ factor. The second modification is the following. The
algorithm in Section 4 of \cite{ChuzhoyGLNPS_det_cut} use a simple
observation that a union of sparse cuts is also sparse. While this
is true for undirected graphs, this is not true in directed graphs
because a sparse cut can be sparse either because of few out-going
edges or because of few in-coming edges. Fortunately, we can show
that there is a large subset of the union whose sparsity is at most
twice. This is formally stated and proved below in \Cref{prop:union sparse cut}.
Lastly, the recursive algorithm in Section 4 of \cite{ChuzhoyGLNPS_det_cut}
needs the cut-matching game of Khandekar et.~al~\cite{KhandekarKOV2007cut}
which works for only undirected graphs (i.e.~the matching player
inserts undirected matchings). But we have generalized the analysis
of this cut-matching game to work even when the matching players inserts
directed matchings in \Cref{sec:CMG round}. With these three technical
modification, we can prove \Cref{thm: directed cut player unweight}
by following the same steps of the algorithm shown in Section 4 of
\cite{ChuzhoyGLNPS_det_cut}.
\end{itemize}
\begin{prop}
\label{prop:union sparse cut}Let $G_{1},G_{2},\dots,G_{k}$ be a
sequence of weighted directed graphs obtained by the following process.
For each $i$, there is a set $S_{i}\subset V(G_{i})$ such that $G_{i+1}=G_{i}[V(G_{i})\setminus S_{i}]$,
$|S_{i}|\le|V(G_{i}|/2$, and $\Psi_{G_{i}}(S_{i})\le\psi$. Suppose
$|\bigcup_{i}S_{i}|\le|V(G_{1})|/2$. Then, there is a set $S\subseteq\bigcup_{i}S_{i}$
where $|S|\ge|\bigcup_{i}S_{i}|/2$ such that $\Psi_{G_{1}}(S)\le3\psi$. 
\end{prop}

\begin{proof}
For each $i$, we say that $S_{i}$ is \emph{out-sparse} if $w\left(E\left(S_{i},V(G_{i})\setminus S_{i})\right)\right)\le\psi|S_{i}|$
and $S_{i}$ is \emph{in-sparse} if $w\left(E\left(V(G_{i})\setminus S_{i}),S_{i}\right)\right)\le\psi|S_{i}|$.
Let $S^{out}$ and $S^{in}$ be the union of out-sparse sets $S_{i}$
and the union of in-sparse sets $S_{i}$ respectively. We assume w.l.o.g.
that $|S^{out}|\ge|S^{in}|$, otherwise the proof is symmetric. Note
that $|S^{out}|\ge|\bigcup_{i}S_{i}|/2$ and $S^{out}\subseteq\bigcup_{i}S_{i}$. 

First, we claim that $w\left(E\left(S^{out},S^{in}\right)\right)\le\psi(|S^{out}|+|S^{in}|)$.
To see this, suppose that $S_{1}$ is out-sparse. Then, we have 
\begin{align*}
w\left(E\left(S^{out},S^{in}\right)\right) & \le w\left(E\left(S_{1},S^{in}\right)\right)+w\left(E\left(S^{out}\setminus S_{1},S^{in}\right)\right)\\
 & \le w\left(E\left(S_{1},V(G_{1})\setminus S_{1}\right)\right)+w\left(E\left(S^{out}\setminus S_{1},S^{in}\right)\right)\\
 & \le\psi|S_{1}|+\psi(|S^{out}\setminus S_{1}|+|S^{in}|)
\end{align*}
where the last inequality is because $S_{1}$ is out-sparse and because
we can continue the same argument on $S_{2}$ and $w\left(E\left(S^{out}\setminus S_{1},S^{in}\right)\right)$.
If $S_{1}$ is in-sparse the argument is the symmetric. Next, observe
that $w\left(E\left(S^{out},V(G_{1})\setminus(S^{out}\cup S^{in})\right)\right)\le\psi|S^{out}|$.
To see this, we write $S^{out}=S_{j_{1}}\cup\dots\cup S_{j_{k'}}$
where, for each $i$, $S_{j_{i}}$ is an out-sparse cut and $j_{i}<j_{i+1}$.
Then, we have 
\[
w\left(E\left(S^{out},V(G_{1})\setminus(S^{out}\cup S^{in})\right)\right)\le\sum_{i}w\left(E\left(S_{j_{i}},V(G_{j_{i}})\setminus S_{j_{i}}\right)\right)\le\psi\sum_{i}|S_{j_{i}}|=\psi|S^{out}|.
\]
Therefore, we have
\begin{align*}
w\left(E\left(S^{out},V(G_{1})\setminus S^{out}\right)\right) & \le w\left(E\left(S^{out},V(G_{1})\setminus(S^{out}\cup S^{in})\right)\right)+w\left(E(S^{out},S^{in})\right)\\
 & \le\psi|S^{out}|+\psi(|S^{out}|+|S^{in}|)\\
 & \le3\psi|S^{out}|.
\end{align*}
As $|S^{out}|\le|\bigcup_{i}S_{i}|\le|V(G_{1})|/2$, so $\Psi_{G_{1}}(S^{out})\le3\psi$.
\end{proof}

\section{Achieving Almost Path-Length Query-Time}

\newcommand{\vertexinpath}{\textrm{{\sc Vertex-In-Path}}}
\newcommand{\curvertex}{\textrm{CurVertex}}
\newcommand{\curpath}{\textrm{CurPath}}
\newcommand{\curpointer}{\textrm{CurPointer}}

\label{sec:query}

In this section, we show how our decremental SCC Algorithm (Algorithm \ref{alg:main}) responds to queries. By Proposition \ref{thm:lacki}, we only need to show how to answer SCC path-queries in $\gstar[\vstar \setminus \Shat]$. Since the algorithm explicitly maintains the connected components of $\gstar[\vstar \setminus \Shat]$ (these are precisely the sets in $\mathcal{C}$), the query can easily determine in $O(1)$ time whether two vertices belong to the same SCC in $\gstar[\vstar \setminus \Shat]$. All that remains is to show that if $u$ and $v$  belong to the same SCC $G$ in $\gstar[\vstar \setminus \Shat]$, then the agorithm can efficiently return a simple path from $u$ to $v$ in $G$. (A path in the other direction can be returned using an analogous argument.)

Since $G = (V,E)$ is an SCC in $\gstar[\vstar \setminus \Shat]$, we know that the algorithm makes some call $\scchelper(G)$. Let $W$ be the large witness maintained in line \ref{line:scchelper-rwitness}. Since the algorithm also maintains data structure $\pathtowitness(G,W,\phistar)$ (Line \ref{line:scchelper-pathtowitness}), we can in $O(\log(n))$ time find vertices $w$, $w'$ in $W$ such that $u$ is contained in an in-directed tree $T$ rooted at $w$ and $v$ is contained in an out-directed tree $T'$ rooted at $w'$ (see guarantees of Theorem \ref{thm:path-to-witness}). Finally, we can find use $\shortoracle(W)$, maintained in Line \ref{line:scchelper-shortoracle}, to find a path $P_W$ from $w$ to $w'$ in $E(W)$, where $|P_W| = n^{o(1)}$. Note that the path $P_W$ uses the edges of witness $W$, NOT the edges of $G$. The total time spent up to this point is only $n^{o(1)}$.

For convenience, we relabel vertices a bit. Let $u = v_1$. Let ${v_2, \ldots, v_{k-1}}$ be the edges in $W$ on $P_W$; so $v_2 = w$ and $v_{k-1} = w'$. Let $v_k = v$. Since $|P_W| = n^{o(1)}$, we also have $k = n^{o(1)}$.

We first consider a naive procedure query, and show that while it successfully returns a path, it is not efficient enough. We can use $T$ and $T'$ to find paths $P_u$ and $P_v$ in $G$, which are respectively from $u$ to $w$ and from $w'$ to $v$. Now, let $\pset$ be the embedding of $W$ into $G$, which is explicitly maintained by the call to $\rwitness$ in Line \ref{line:scchelper-rwitness} of the algorithm. We can use $\pset$ to convert the path $P_W = (v_2, ..., v_{k-1})$ into a path in $G$. Each edge $(v_i,v_{i+1}) \in P_W \subseteq E(W)$ corresponds to some path $v_2 \rightarrowtail v_3$ in $\pset$, so concatenating these yields a path $P_G \subset E(G)$ from $v_2$ to $v_{k-1}$. We then return the $u-v$ path $P = P_u \circ P_G \circ P_v$. Note that $P_u$ and $P_v$, as well as the paths in $\pset$, can be as long $\Ohat(1/\phistar) = \Ohat(n^{1/3})$, so $P$ can be quite long. At first glance this does not seem to be a problem, because it is not hard to check that the time spent to find $P$ is $O(|P|)$. The issue is that the path $P$ might not be simple. Say, for example, that the first edge of $P_u$ is $(u,z)$ and the before-last edge of $P_v$ is $(y,z)$. Then almost all of $P$ consists of a long cycle from $z$ to $z$. Of course, we can always extract a simple path $P' \subseteq P$, but in the example above $P'$ will be the path $(u,z) \circ (z,v)$. We thus spent as much as $\Ohat(n^{1/3})$ returning a path of length 2. 

In order to achieve almost path-length query time, we thus need a more clever query procedure. We start with some notation. Let $P_1$ and $P_{k-1}$ be the paths described above from $v_1$ to $v_2$ and from $v_{k-1}$ to $v_k$; these paths are both contained in acyclic trees, so they are simple. Similarly, for $2 \leq i \leq k-2$, let $P_i$ be the path in $\pset$ from $v_i$ to $v_{i+1}$; these are all simple because they correspond to paths in $\pset$, which form the path decomposition of a flow (see \Cref{rem:excess at source and path decomposition}). In this terminology, the naive query procedure is to look at all of the edges in all of the $P_i$, and concatenate them. We now show a different method that allows us to effectively throw away long cycles without having to look at all the edges on the cycle.

\subsection{Improved Query Procedure}

\paragraph{Minor additions to the data structures used by Algorithm \ref{alg:main}}
Recall that the paths $P_1, \ldots, P_k$ all come from $\pathtowitness$ and $\rwitness$. Our query procedure requires these two algorithms to construct slightly more powerful data structures. Both the additions are light-weight, and will only increase the total update time of these algorithms by a $O(\log(n))$ factor, which is subsumed in the $\Ohat$-notation.

Recall that $\rwitness$ (Theorem \ref{thm:robust witness}) explicitly constructs all the flow paths in embedding $\pset$. These paths can be stored as doubly linked lists.  For the query procedure to work, we also have $\rwitness$ build a simple data structure on each path $P$: build a balanced binary search tree on the vertices in $P$, and let each node in the tree have a pointer to the corresponding node in list $P$. This can clearly be done in $O(|P|\log(n))$ time. Note that we do not need to maintain these data structures dynamically, because within each phase of $\rwitness$, individual paths in the embedding never change; the embedding changes only via deleting entire paths. Every time $\rwitness$ enter a new phase, it computes a new embedding from scratch, at which point we can again construct our data structure on each path $P$ with only $O(\log(n))$ overhead.

Recall that $\pathtowitness$ (Theorem \ref{thm:path-to-witness}) maintains a forest of trees. Firstly, for each vertex $x \in V$, we maintain a pointer to the corresponding node in the tree that contains $x$, and vice versa: these pointers never change, only incur $O(1)$ overhead. We also maintain a top tree on each tree in the forest: see e.g. the paper by Alstrup et al. for a nice overview \cite{AlstrupHLT05}. These trees can perform link and cut operations in $O(\log(n))$ time, maintaining them incurs at most a $O(\log(n))$ multiplicated overhead in the update time. (In fact, the proof of Theorem \ref{thm:path-to-witness} in Section \ref{sec:forest} already uses link-cut trees, so in our case using top-trees incurs no additional overhead.) The key operation we need from top trees is that given any vertices $x,y \in V$, we can {\bf 1)} Given any $x,y \in V$, determine whether they are in the same tree. This is done by using the pointers to the respective nodes of $x$ and $y$ in the forest and checking if they have the same root. {\bf 2)} check if $y$ is on the path between $x$ and the root. Letting $r$ be the root vertex, this is done by checking if $\dist(x,y) + \dist(y,r) = \dist(x,r)$; see Lemma 5 of \cite{AlstrupHLT05} for details on how the top-trees can be used to return distances in the tree.

The above data structures lead to the following claim

\begin{claim}
Let $T_1$ and $T_k$ be the trees maintained for $v_1$ and $v_k$ by $\pathtowitness$, say that paths  $P_2, \ldots, P_{k-1}$ are stored as doubly linked lists, and say that we also have the augmented data structures described above. Then, given any vertex $x \in V$ and any index $i$ with $1 \leq i \leq k$, it is possible to answer the following query $\vertexinpath(x,i)$ in $O(\log(n))$ time:
\begin{enumerate}
	\item If $x \notin P_i$ return False
	\item If $x \in P_i$, returns True and also returns a pointer to the node corresponding to $x$ in the path $P_i$: for $P_2, \ldots, P_{k-1}$ this means the node in the doubly linked list $P_i$, and for $P_1, P_k$ this means the node in the corresponding tree $T_1, T_k$.
\end{enumerate}
\end{claim}

\begin{proof}
The claim follows directly from the augmented data structures. If $2 \leq i \leq k-1$, then the binary search tree on $P_i$ allows us to search for $x$ in $O(\log(|P_i|)) = O(\log(n))$ time; if the node is found, then we follow the pointer from the binary search tree to the path.

Say $i=1$; the case $i=k$ is analogous. As mentioned above, the top tree on $T_1$ allows us to check if $x$ is in $T_1$, and if yes, to determine is $x$ is on path $P_1$ by checking if it is ancestor for $v_1$. The pointer to the node $x$ in the tree comes from the fact that we store pointers to and from every vertex in $G$ and the corresponding nodes in the forest.
\end{proof}

\paragraph{The Algorithm:}

Say that the first edge on $P_1$ is $(v_1, z)$. To avoid exploring a long cycle through $z$ (see example above), before continuing from $z$ the algorithm checks if $z$ is on one of the other paths $P_i$. If not, it can safely continue. If yes, let $P_j$ be the path that contains $z$ with maximum $j$. Then, instead of continuing the search from $P_1$, the algorithm continues from $P_j$. This guarantees that there can be no cycle through $j$, because $P_j$ is simple, and no later path contains $z$.

The pseudocode in Algorithm \ref{alg:query} formalizes the intuition above.

\begin{algorithm2e}
	\label{alg:query}
\caption{Finding a path from $v_1$ to $v_k$. Recall the paths $P_1,\ldots,P_k$ defined above.}

\SetKwBlock{RepeatUntilEnd}{Repeat Until $\mathbf{\curvertex = v_k}$}

Initialize $\curvertex \gets v_1$  \;

Initialize $\curpath \gets 1$ \;

Initialize $\curpointer$ to point to $v_1$ in $P_1$ \tcp*[f]{always points to $\curvertex$ in $P_{\curpath}$} \;

Initialize $\Pstar \gets \emptyset$ \tcp*[f]{$\Pstar$ is returned at the end, and will always be simple} \;

\RepeatUntilEnd(\label{query:main-loop}){
Do $\vertexinpath(\curvertex, i)$ for all $i > \curpath$ \;

\If(\tcp*[f]{no cycle through \curvertex}){none of the $\vertexinpath$ return True} {
	
Let $z$ be the vertex after $\curvertex$ on path $P_{\curpath}$ (can find $z$ by following $\curpointer$ and then taking the next edge in the path/tree) \;

Add edge $(\curvertex,z)$ to $\Pstar$ \;

$\curvertex \gets z$; adjust $\curpointer$ to point to $z$ \;
} %

\Else{
Let $i$ be the largest index such that $\vertexinpath(\curvertex, i)$ returns True \;

$\curpath \gets i$ \;

Set $\curpointer$ to be pointer returned by $\vertexinpath$	\;
} %

} %

Return $\Pstar$	
\end{algorithm2e}

\paragraph{Analysis}
Firstly, note that when we execute the main loop in Line \ref{query:main-loop}, we cannot land in the else statement twice in a row, since the else statement always switches to the highest-indexed path that contains $\curvertex$. So for every two iterations of the loop, we execute the if statement at least once, and hence add an edge to $\Pstar$.

Consider the (possibly non-simple) path $P = P_1 \circ P_2 \circ \ldots \circ P_k$. It is easy to see that in every iteration of the main loop, the algorithm jumps forward in $P$: it either goes forward one vertex in some path $P_i$ (the if statement), or it jumps from $\curvertex$ to another copy of $\curvertex$ on a later path (the else statement). In other words, the vertices of $\pstar$ for a subsequence of the vertices in $P$. Thus, the algorithm eventually reaches $v_k$ and terminates.

We now argue that the returned path $\pstar$ is simple. Consider any vertex $x \in \pstar$ and consider the first time we added $x$ to $\pstar$; say that at this time $\curpath = i$. We argue that $x$ will never be reached again. The \emph{first case} is that $\vertexinpath(x,j)$ returns False for all $j > i$. In this case $\pstar$ will never again reach $x$, because as argued in the above paragraph, $\pstar$ only moves forward along $P$; it cannot reach $x$ a second time in $P_i$ because each $P_i$ is simple and $x$ is not contained in any of the later $P_j$. The \emph{second case} is that $x \in P_j$ for some $j \geq i$. Let $j$ be the largest index such that $j \geq i$. Then the else-statement of the main loop switches to $P_j$ without adding any vertices to $\pstar$ and in the next iteration we are the first case, so there is no cycle through $x$.

For the running time analysis, note that each iteration of the main loop executes $k = n^{o(1)}$ instances of $\vertexinpath$, each of which takes $O(\log(n))$ time, so the running time is $\Ohat($\# iterations of main loop$)$. We argued above that for every two iterations of the while loop at least one vertex is added to $\pstar$. We thus have a running time of $\Ohat(\Pstar)$, as desired.

\section{Deterministic SSSP in Decremental Graphs}

In this section, we prove one of our main results: \Cref{thm:main_sssp}. Recall that our decremental SSR/SCC result combines our new expander-based framework with earlier techniques for decremental SCC in \cite{Lacki11,ChechikHILP16}. Our decremental SSSP results uses the new framework in a similar way, but now combines it with earlier tools for decremental SSSP in \cite{GutenbergW20a,nearOptDenseSSSP}. In particular, we start with the following proposition, which essentially combined Proposition \ref{thm:lacki} and Theorem \ref{thm:path-to-witness}.

\begin{restatable}{prop}{PGWNProp}
	\label{prop:PGWN}
Let $G=(V,E,w)$ be a weighted decremental graph, and $s \in V$ a fixed source. Let $\mathcal{A}$ be a data structure given some integer $d >0$, that processes edge deletions to $E$ and after every edge deletion ensures that {\bf 1)} $G$ is strongly-connected and has diameter at most $d$ and {\bf 2)} supports path queries between any two vertices in $G$ that returns a path of length $\Ohat(d)$ in almost-path-length query time and runs in total update time $T(m,n,d)$ (here we assume $T(m_1, n_1, d_1) + T(m_2, n_2, d_2) \leq T(m,n,d)$ for all choices $m,n,d$ and $m_1, m_2, n_1, n_2, d_1, d_2$ such that $m = m_1 + m_2$, $n = n_1 + n_2$ and $d = d_1 + d_2$). At any time the data structure may perform the following operation: it finds and outputs a $\Ohat(1/d)$-sparse cut $(L,S,R)$ where $|L| \leq |R|$ and replaces $G$ with $G[R]$; here we only require the algorithm to output $L$ and $S$ explicitly. (This sparse-cut operation is not an adversarial update, but is rather something the data structure can do of its own accord at ay time.)

Then, there exists a deterministic data structure $\mathcal{B}$ that can report $(1+\epsilon)$-approximate distance estimates and corresponding paths from $s$ to any vertex $v \in V$ in the graph $G$ in almost-path-length query time and has total update time $\Ohat((T(m,n,\delta) + n^3/\delta + n^2 \delta + mn^{2/3})\log W/\epsilon)$ for any choice of $\delta, \epsilon > 0$. (Note that the data structure can cause $V(G)$ to shrink over time via sparse-cut operations, so it only has to answer queries for vertices $u,v$ in the \emph{current} graph.)
\end{restatable}

It is straight-forward to obtain \Cref{thm:main_sssp} from the proposition, and \Cref{thm:robust witness}.

\begin{proof}[Proof of \Cref{thm:main_sssp}.]
We now show how to implement the data structure $\mathcal{A}$ required by the setup of \Cref{prop:PGWN}, with $T(m,n,\delta) = \Ohat(m\delta^2)$ as follows. Given the graph $G$, we can invoke the algorithm described in \Cref{thm:robust witness} with parameter $\phi = \hat{\Theta}(1/\delta)$, such that the algorithm  maintains a $\phi$-short-witness $W$ that restarts up to $\Ohat(1/\phi) = \Ohat(\delta)$ times. Whenever $W$ starts a new \emph{phase}, we use the data structures from \Cref{thm:path-to-witness} and 	\Cref{thm:short-path-oracle} on $G$ and $W$ until the phase ends. We forward the sparse cuts $(L,S,R)$ found in the algorithm from \Cref{thm:path-to-witness} and  \Cref{thm:robust witness} and update $G$ accordingly. Thus after the algorithm from \Cref{thm:robust witness} terminates, the graph $G$ contains only a constant fraction of the vertices that the algorithm in \Cref{thm:robust witness} was initialized upon. We then repeat the above construction and note that after at most $O(\log n)$ times, the graph $G$ is the empty graph. 

We note that to obtain a path between any two vertices in the current graph $G$, we can query the data structures from \Cref{thm:path-to-witness} and \Cref{thm:short-path-oracle} as we described in \Cref{sec:ingredient} to obtain such a path of length $\Ohat(\delta)$ in almost-path-length time. We further observe that if we set $\phi$ to $\frac{1}{\delta n^{o(1)}}$, for a large enough subpolynomial factor $n^{o(1)}$, then we can ensure that vertices in $G \setminus W$ are at all times at most $\delta/3$ away from some vertex in $W$ by \Cref{thm:path-to-witness}, have that any two vertices in $W$ are at distance at most $\delta/3$ to each other in $G$ by \Cref{thm:short-path-oracle} and \Cref{thm:robust witness}, and again, that there exists a path to every vertex in $G \setminus W$ to a vertex in $W$ of length at most $\delta/3$. But this implies that any two vertices in $G$ are at all times at distance at most $\delta$ and therefore the diameter of $G$ is upper bounded by $\delta$, as required.

The total update time of the data structure $\mathcal{A}$ is at most $\Ohat(m/\phi^2) = \Ohat(m\delta^2)$ by adding the running time of \Cref{thm:robust witness} with the running time induced by the algorithms in \Cref{thm:path-to-witness} and 	\Cref{thm:short-path-oracle} which are restarted in $\Ohat(\delta)$ phases. 

We thus derive an algorithm $\mathcal{B}$ as specified in \Cref{prop:PGWN}, where we use the above data structure $\mathcal{A}$ and where we set $\delta = n^{1/3}$ which gives total update time
\begin{align*}
&\Ohat((T(m,n,\delta) + n^3/\delta + n^2 \delta + mn^{2/3})\log W/\epsilon)= n^{2+2/3 + o(1)} \log W/\epsilon.
\end{align*}
\end{proof}

The rest of this section is dedicated to prove \Cref{prop:PGWN}. We therefore introduce necessary notation in the next subsection, then introduce the abstraction of an approximate topological order which we reduce the problem to and finally prove that an approximate topological order can be maintained efficiently.

\subsection{Additional Preliminaries}

Given two partitions $\mathcal{A}$ and $\mathcal{B}$ of a universe $V$, we say $\mathcal{A}$ is a \textit{refinement} of $\mathcal{B}$ if and only if for every set $A \in \mathcal{A}$ there exists a set $B \in \mathcal{B}$ such that $A \subseteq B$.

Throughout the section, we let $u \leadsto_G v$ denote that $u$ reaches $v$ in $G$, and $u \rightleftarrows_G v$ that $u$ and $v$ are strongly-connected, i.e. that $u$ reaches $v$ and $v$ reaches $u$. We call the tuple $(\mathcal{V}, \tau)$ the generalized topological order of $G$, if $\mathcal{V}$ is the set of SCCs in $G$ and $\tau : \mathcal{V} \rightarrow [1,n]$ is a function that maps each SCC $X$ in $\mathcal{V}$ to a number $\tau(X)$ such that no other $Y \in \mathcal{V}$ has $\tau(Y) \in [\tau(X), \tau(X) + |X| - 1]$. Thus, $\tau$ establishes a one-to-one correspondence between SCCs in $X$ and intervals of size $|X|$ in $[1,n]$. In a decremental graph $G$, we have that a generalized topological order has the property that each $\mathcal{V}$ is a refinement of its earlier versions, since SCCs decompose over time. 

We say that $(\mathcal{V}, \tau)$ has the \emph{nesting} property, if for any set $X \in \mathcal{V}$ and a set $Y \subseteq X$ that was in $\mathcal{V}$ at an earlier stage, that $\tau(X) \in [\tau(Y), \tau(Y) + |Y| - |X|]$. Thus, the interval $[\tau(X), \tau(X) + |X| -1]$ associated with $X$ is contained in the interval $[\tau(Y), \tau(Y) + |Y| - 1]$ associated with $Y$.

Given a partition $\mathcal{V}$ of $V$, we let $G / \mathcal{V}$ denote the multi-graph of $G$ after contracting vertices that are in the same set $X \in \mathcal{V}$, where we keep all edges, i.e. also self-loops and parallel edges. Abusing notation slightly, we refer to $\mathcal{V}$ as the \textit{node} set of the graph $G / \mathcal{V}$.

For convenience, we define $\mathcal{T}(X, Y, (\mathcal{V}, \tau))$ to be the function that takes as parameters two SCCs $X,Y \in \mathcal{V}$ and a generalized topological order $(\mathcal{V}, \tau)$ of $G$, and define the function
\[
 \mathcal{T}(X ,Y, (\mathcal{V}, \tau)) = \begin{cases}
                        \tau(Y) - (\tau(X) + |X| - 1) & \text{if } \tau(X) < \tau(Y) \\
                        \mathcal{T}(Y, X) & \text{otherwise}
                     \end{cases}
\]
For any path $P$ in $G$, we let $\mathcal{T}(P, (\mathcal{V}, \tau))$ denote the total topological distance traversed by $P$ in the topological order $ (\mathcal{V}, \tau)$.
Formally,
\[
 \mathcal{T}(P, (\mathcal{V}, \tau)) = \sum_{(X,Y) \in P / \mathcal{V}} \mathcal{T}(X,Y, (\mathcal{V}, \tau)).
\]

\subsection{SSSP via Approximate Topological Orders}

We now introduce the concept of an approximate topological order which we define similar to \cite{nearOptDenseSSSP} and which we implement similar to \cite{GutenbergW20a}. The main idea of an approximate topological order is as follow: consider the generalized topological order $(\mathcal{V}, \tau)$ of a graph $G$. Then $G / \mathcal{V}$ is a directed acyclic graph by definition. But this implies that for any (shortest) $s$-to-$t$ path $\pi_{s,t}$ in $G$ we have that  every edge $(X,Y)$ on $\pi_{s,t} / \mathcal{V}$ in $G / \mathcal{V}$ has $\tau(X) < \tau(Y)$. Since further $\tau$ maps to numbers between $1$ and $n$, we have thus that summing along the topological difference of the edges of  $\pi_{s,t} / \mathcal{V}$, we that $\mathcal{T}(\pi_{s,t}, (\mathcal{V}, \tau))$ is at most $n$.

Next, let us assume that the sum of diameters of all SCCs in $\mathcal{V}$ is at most $\epsilon \delta$, then for any shortest path $\pi_{s,t}$, we can upper bound the difference in weight between $\pi_{s,t}/\mathcal{G}$ path in $G / \mathcal{V}$ as opposed to $\pi_{s,t}$ in $G$ by an additive term of $\epsilon \delta$. So, if $\pi_{s,t}$ is of weight at least $\delta$, the additive term can be subsumed in a multiplicative error of $(1\pm\epsilon)$.

Now, the gist of this set-up is that given this upper bound on $\mathcal{T}(\pi_{s,t}, (\mathcal{V}, \tau))$, we can implement a fast SSSP data structure as follows. We know that on a path of length $\delta$ in $G / \mathcal{V}$ there are at most $\delta / 2^i$ edges that have topological order difference more than $2^i n/\delta$ by the pigeonhole principle for any $i$. But this implies that adding an additive error of $\epsilon 2^i$ on each such edge would only amount to an $(1+\epsilon)$ multiplicative error of a shortest path of length $\delta$. But this allowance for a significant additive error can be exploited to speed-up the SSSP data structure significantly because it allows for vertices to consider the neighbors that are close in topological order difference more closely while being more lenient when passing updates to vertices that are far in terms of topological order difference. 

Before we state a data structure from \cite{nearOptDenseSSSP} that exploits this very efficiently, let us now state more formally the construct of an approximate topological order. Here, we point out one last issue: we cannot assume that SCCs in $G$ have small diameter in general. Therefore we maintain the generalized topological order on a graph $G'$ initialized to $G$ where we, additionally to adversarial edge updates to $G$, also take vertex separators $S$ such that edges incident to $S$ are deleted from $G'$. This ensures that all SCCs in $G'$ have small diameter. Relating back to $G$ (where no separator was deleted) we have that $\mathcal{T}(\pi_{s,t}, (\mathcal{V}, \tau))$ might be increased by this operation since some edge $(X,Y)$ on $\pi_{s,t} /\mathcal{V}$ with $X$ or $Y$ containing a separator vertex $S$, such that $(X,Y)$ might now go "backwards" in the topological order, i.e. have $\tau(X) > \tau(Y)$. This increases $\mathcal{T}(\pi_{s,t}, (\mathcal{V}, \tau))$ by up to $2n-2$ for every separator vertex since we might move along $(X,Y)$ all the way back in the topological order and then forward again. However, by choosing small separators, we can still bound $\mathcal{T}(P, (\mathcal{V}, \tau))$ by a non-trivial upper bound.

Without further due, let us give the formal definition of an approximate topological order.

\begin{restatable}{defn}{ATO}
\label{defn:ATOdecomposition}
Given a decremental weighted digraph $G=(V,E,w)$ and parameter $\eta \leq n$ and $\nu \leq W$, we say a dynamic tuple $(\mathcal{V}, \tau)$ where $\mathcal{V}$ partitions $V$, and $\tau : \mathcal{V} \rightarrow [1,n]$, is an $\mathcal{ATO}(G, \eta, \nu)$ if at each stage
\begin{enumerate}
    \item $\mathcal{V}$ forms a refinement of all earlier versions of $\mathcal{V}$ and $\tau$ is a \emph{nesting} function, i.e. $\tau$ initially assigns each set in $X$ in the initial version of $\mathcal{V}$ a number $\tau(X)$, such that no other set $Y$ in $\mathcal{V}$ has $\tau(Y)$ in the interval $[\tau(X), \tau(X) + |X| - 1]$. If some set $Y \in \mathcal{V}$ is split at some stage into disjoint subsets $Y_1, Y_2, .., Y_l$, then we let $\tau(Y_1) = \tau(Y)$ and $\tau(Y_{i+1}) = \tau(Y_i) + |Y_i|$. We then return a pointer to each new subset $Y_i$ such that all vertices in $Y_i$ can be accessed in time $O(|Y_i|)$. The value $\tau(X)$ for each $X \in \mathcal{V}$ can be read in constant time. \label{prop:TauUpdate}
    \item each set $X$ in $\mathcal{V}$ has weak diameter $\mathbf{diam}(X, G) \leq \frac{|X| \eta \nu}{n}$, and \label{prop:ContractLittle}
    \item At each stage, for any vertices $s, t \in V$, the shortest-path $\pi_{s,t}$ in $G$ satisfies $\mathcal{T}(\pi_{s,t}, (\mathcal{V}, \tau)) = \Ohat
    \left(\frac{n^2}{\eta} + n\cdot \frac{\mathbf{dist}_G(s,t)}{\nu}\right)$. \label{prop:TauTotal}
\end{enumerate}
\end{restatable}

Here, we captured in Property \ref{prop:TauUpdate}, that the vertex sets in $\mathcal{V}$ decompose over time, that $\tau$ is \emph{nesting} and that all sets are easily accessible. In Property \ref{prop:ContractLittle}, we capture that the sum of diameters of the vertex sets in $\mathcal{V}$ is small. It is not hard to see that by summing the upper bound on the diameter of all such sets $X$ in $\mathcal{V}$, we get that the sum of diameters is bounded by $\eta \nu$. Finally, we give an upper bound for the topological order difference for any shortest-path in $G$.

The main result of the next section, shows that we can maintain an $\mathcal{ATO}$ using data structure $\mathcal{A}$ from \Cref{prop:PGWN}.

\begin{restatable}{lemma}{AtoDecomp}
\label{lem:ATOdecomposition}
Given a decremental weighted digraph $G=(V,E,w)$, parameters $\eta \leq n, \nu \leq W$, and a data structure $\mathcal{A}$ as described in \Cref{prop:PGWN} that can for each SCC $X$ in $\mathcal{V}$ at any point return a path between any two vertices $u,v \in X$ of length $\Ohat(\frac{|X|\eta \nu}{n})$ in near-linear time. Then, we can deterministically maintain a $\mathcal{ATO}(G, \eta, \nu)$ in total update time $\Ohat(T(m,n,\eta) + mn^{2/3})$.
\end{restatable}

From \cite{nearOptDenseSSSP}, we now obtain the following theorem. Note that we slightly modified the theorem from \cite{nearOptDenseSSSP} to adapt it to the simplified definition of an $\mathcal{ATO}$ that we use for this paper. However, the adaption is obtained straight-forwardly and we refer the reader to \cite{nearOptDenseSSSP} to verify. 

\begin{restatable}[see \cite{nearOptDenseSSSP}, Theorem 5.1]{theorem}{ssspSimple}
\label{thm:SSSPEfficient}
Given $G=(V,E,w)$, a decremental weighted digraph, a source $r \in V$, an approximation parameter $\epsilon > 0$, and access to $(\mathcal{V}, \tau)$ an $\mathcal{ATO}(G, \eta, \nu)$.

Then, there exists a deterministic data structure that maintains a distance estimate $\widetilde{\mathbf{dist}}(r,v)$ for every vertex $v \in V$ such that at each stage of $G$, $\mathbf{dist}_G(r,v) \leq \widetilde{\mathbf{dist}}(r,v)$ and if $\mathbf{dist}_G(r,v) \in [\eta \nu/\epsilon, 2\eta \nu/\epsilon)$, then 
\[
\widetilde{\mathbf{dist}}(r,v) \leq (1+\epsilon)\mathbf{dist}_G(r,v)
\]
and the algorithm can for each such vertex $v$, report a path of length $(1+\epsilon)\mathbf{dist}_G(r,v)$ in the graph $G / \mathcal{V}$ in almost-path-length time. The total time required by this structure is 
\[
\Ohat\left(\frac{n^3}{\eta\epsilon} + \cdot \frac{n^2 \eta}{\epsilon}\right).
\]
\end{restatable}

We can now prove \Cref{prop:PGWN}.

\begin{proof}[Proof of \Cref{prop:PGWN}.]
For every $0 \leq i \leq \lg W$, where $W$ is the aspect ration of $G=(V,E,w)$, we maintain at level $i$, an $\mathcal{ATO}(G, \delta, 2^i)$ using \Cref{lem:ATOdecomposition}, and then running \Cref{thm:SSSPEfficient} on $G$ and the $\mathcal{ATO}(G, \delta, 2^i)$ from our source vertex $s$ to depth $\delta \cdot 2^i$. Thus, each such data structure maintains for every vertex $v$ at distance $[\delta \cdot 2^i/\epsilon', \delta \cdot 2^{i+1}/\epsilon')$ from $s$ an $(1+\epsilon')$-approximate distance estimate. We can therefore find for every vertex $v$ at distances larger than $\delta/\epsilon'$ from $s$ a distance estimate in some of these data structures that gives the right approximation, and since all data structures overestimate the distance, we can find the right distance estimate by comparing all distance estimates $\widetilde{\mathbf{dist}}(s,v)$. Finally, we can maintain a simple ES-tree in time $O(m \delta/\epsilon')$ to obtain exact distances from $s$ to every vertex at distance at most $\delta$.

It is not hard to verify that the total update time of all data structures is
\begin{align*}
&\sum_{0 \leq i \leq \lg W} \left( \Ohat\left(\frac{n^3}{\delta\epsilon'} + \cdot \frac{n^2 \delta}{\epsilon'}\right) + \Ohat(T(m,n,\delta) + mn^{2/3})\right) \\
&= \Ohat((T(m,n,\delta) + n^3/\delta + n^2 \delta + mn^{2/3})\log W/\epsilon').
\end{align*}
for $\epsilon'$ to be set $\epsilon' = \epsilon/n^{o(1)}$ which is again subsumed in the $\Ohat$-notation.

To answer path queries for a $s$-to-$v$ path $\pi_{s,v}$, we query the corresponding shortest path data structure where we found a $(1+\epsilon')$-approximation. This gives us the path $\widetilde{\pi_{s,v}}$ in $G / \mathcal{V}$ for some $\mathcal{ATO}$ $(\mathcal{V}, \tau)$. We then identify for every vertex $x$ on $\widetilde{\pi_{s,v}}$ the corresponding SCC in $\mathcal{V}$ and the two endpoints in $G$ of the incident edges on $\widetilde{\pi_{s,v}}$. We can then query for a path between these two vertices in the $\mathcal{ATO}$ data structure. Summing over all exposed paths, by \Cref{lem:ATOdecomposition}, we can extend the path $\widetilde{\pi_{s,v}}$ to a path in $G$ of length $(1+\epsilon')\mathbf{dist}_G(s,v) + \Ohat(\eta \nu)$. But we have that $\mathbf{dist}_G(s,v) \geq \delta/\epsilon'$. Thus, setting $\epsilon'$ to $\epsilon/2$ divided by the subpolynomial factor hidden in $\Ohat(\eta \nu)$, we obtain a path of length $(1+\epsilon)\mathbf{dist}_G(s,v)$. Since each piece on the path can be obtained in almost-path-length time, we can also construct the extension of path $\widetilde{\pi_{s,v}}$ to a path in $G$ in almost-path-length time. This completes the proof.
\end{proof}
\subsection{A Deterministic Algorithm to Maintain an Approximate Topological Order}

Finally, let us prove the main ingredient to achieve our result. 

\AtoDecomp*
\begin{proof}
We start the proof by partitioning the edge set $E$ of the initial graph $G$ into edge set $E^{heavy}$ and $E^{light}$. We assign every edge $e \in E$ to $E^{heavy}$ if its weight $w(e)$ is larger than $\nu$, and to $E^{light}$ if $w(e) \leq \nu$. 

We now describe our algorithm where we focus on the graph $G$ where the edge set $E^{heavy}$ is removed. As we will see later, there can only be few edges from $E^{heavy}$ on any shortest path. Let us start the proof by giving an overview and then a precise implementation. We finally analyze correctness and running time. 

\paragraph{Algorithm.} Our goal is subsequently to maintain an incremental set $\hat{S} \subseteq V$ such that every SCC $X$ in $G' = G \setminus E(\hat{S}) \setminus E^{heavy}$ has unweighted diameter at most $\frac{|X|\eta}{n}$. Since each edge weight is at most $\nu$ this will imply that every SCC $X$ in the weighted version of $G'$ has diameter at most $\frac{|X|\eta\nu}{n}$. 

We then maintain $(\mathcal{V}, \tau)$ as the generalized topological order of $G'$ using the data structure described in the theorem below which is a straight-forward extension of \Cref{thm:main_scc} using internally the algorithm by Tarjan \cite{tarjan1972depth} as described in \cite{GutenbergW20a, nearOptDenseSSSP}.

\begin{theorem}
\label{thm:SCCinDecrGraph}
Given a decremental digraph $G=(V,E)$, there exists a deterministic algorithm that can maintain the SCCs $\mathcal{V}$ of $G$. The algorithm can further be extended to maintain the generalized topological order $(\mathcal{V}, \tau)$ of $G$ where $\tau$ has the nesting property. The algorithm is deterministic and runs in total update time $mn^{2/3 + o(1)}$.
\end{theorem}

To maintain $G'$, we initialize a data structure $\mathcal{A}$ on every SCC $X$ in the initial set $\mathcal{V}$ on the graph $G'[X]$ with parameter $d = \frac{|X|\eta}{2n}$. Then, whenever such a data structure $\mathcal{A}$ that currently operates on some graph $G'[Y]$, announces a sparse cut $(L, S, R)$ and sets its graph to $G'[R]$, we add $S$ to $\hat{S}$ and then initialize a new data structure $\mathcal{A}'$ on $G'[L]$ with parameter $d = \frac{|L|\eta}{2n}$. Further, if the data structure $\mathcal{A}$ was initialized on a graph with vertex set at least twice as large as $R$, we delete $\mathcal{A}$, and initialize a new data structure $\mathcal{A}''$ on $G'[R]$ with $d = \frac{|R|\eta}{2n}$. This completes the description of the algorithm.

\paragraph{Correctness of the Algorithm.} We prove each property of the theorem individually:
\begin{itemize}
    \item \underline{Property \ref{prop:TauUpdate}}: It is straight-forward to see that since $(\mathcal{V}, \tau)$ is the generalized topological order of $G' \subseteq G$ and since it is maintained to satisfy the nesting property, that Property \ref{prop:TauUpdate} follows immediately. 
    \item \underline{Property \ref{prop:ContractLittle}:} Observe that $\mathcal{V}$ is the set of SCCs in $G'$. Further, observe that we maintain the data structures $\mathcal{A}_1, \mathcal{A}_2, \dots$ such that vertex set of all graphs that they run on spans all vertices in $V \setminus S$. For the vertices in $S$ we have that each $s \in \hat{S}$ forms a trivial SCC and therefore certainly satisfies the constraint. For each set $X$ that some data structure $\mathcal{A}$ runs upon, we have that the unweighted diameter is at most the $d$ that $\mathcal{A}$ was initialized with. Observe that we delete data structures if the size of the initial vertex set $Y$ is decreased by factor $2$. Thus, we have that the data structure $\mathcal{A}$ was initialized for some $d = \frac{|Y|\eta}{2n} \leq \frac{|X|\eta}{n}$. Since the largest edge weight in $G'$ is $\nu$, we thus have that for each SCC $X$ in $\mathcal{V}$, we have $\mathbf{diam}(X,G') \leq \frac{|X|\eta \nu}{n}$. Adding edges in $E(\hat{S})$ and $E^{heavy}$ can further only decrease the weak diameter and therefore we finally obtain that,
    \[
        \mathbf{diam}(X,G) \leq \frac{|X|\eta \nu}{n}.
    \]
    \item \underline{Property \ref{prop:TauTotal}}: In order to establish the last property, let us partition the set $\hat{S}$ into sets $S_1, S_2, \dots, S_{\lg n}$ where a vertex $s$ is in $S_i$ if it joined $\hat{S}$ after a data structure $\mathcal{A}$ announced it that was initialized on a graph $G'[Y]$ where $Y$ was of size $[n/2^{i+1},n/2^i)$. Since we delete data structures after their initial vertex set has halved in size, we have that are such data structure that added vertices to a set $S_i$ ran with $d \geq \frac{(n/2^{i+1}) \eta}{2n} = \frac{n \eta}{2^{i+2}}$. Since each such set of vertices $S$ that was added to $S_i$ is $\Ohat(1/d)$-sparse and we then only compute sparse cuts on the induced subgraphs of the cut, we further have that there are at most $\Ohat(n/d) = \Ohat(2^i/\eta)$ vertices in $S_i$ at the end of the algorithm. Further, we observe that every edge $(u,v)$ that was contained in the subgraph $G'[Y]$ when $\mathcal{A}$ was initialized has both endpoints in $Y$ and therefore by property \ref{prop:TauUpdate}, we have $|\tau(u) - \tau(v)| < |Y| \leq n/2^{i-1}$. 
    
    Now, let us fix any shortest path $\pi_{s,t}$ in $G$ (in the current version). Instead of analyzing $\mathcal{T}(\pi_{s,t}, (\mathcal{V}, \tau))$, let us analyze 
    \[
    \mathcal{T}'(\pi_{s,t}, (\mathcal{V}, \tau)) \stackrel{\text{def}}{=} \sum_{(u,v) \in \pi_{s,t}} \max\{ 0 , \tau(u) - \tau(v)\}.
    \]
    which only considers the edges on the path that go "backwards" in the topological order. However, it can be seen that for every path $\mathcal{T}(\pi_{s,t}, (\mathcal{V}, \tau)) \leq 2\mathcal{T}'(\pi_{s,t}, (\mathcal{V}, \tau)) + n$.

    For edges on $\pi_{s,t}$ in $E^{heavy}$, we observe that each such edge $(u,v)$ can contribute to $\mathcal{T}'(\pi_{s,t}, (\mathcal{V}, \tau))$ at most $n$ since $\tau(u) - \tau(v) \leq n$ (trivially since both numbers are taken from the interval $[1,n]$). Further, since each such edge adds weight at least $\nu$ to the shortest path, there are at most $\frac{\mathbf{dist}(s,t)}{\nu}$ such edges. Thus, the total contribution by all these edges is at most $n \frac{\mathbf{dist}(s,t)}{\nu}$.
    
    For the edges on $\pi_{s,t}$ in $E^{light}$, we observe that each edge $(u,v)$ that contributes to $\mathcal{T}'(\pi_{s,t}, (\mathcal{V}, \tau))$ is not in $G'$ since $(\mathcal{V}, \tau)$ is a generalized topological order of $G'$ and therefore directed "forwards" (recall the definition in \cref{sec:prelim}). Thus, each such edge is in $E(\hat{S})$ and therefore incident to some vertex $s$ in some $S_i$. But then it adds at most $n/2^{i-1}$ to $\mathcal{T}'(\pi_{s,t}, (\mathcal{V}, \tau))$ by our previous discussion. Since a path only visits each vertex once, and by our bound on the size of $S_i$, we can now bound the total contribution by
    \[
        \mathcal{T}'(\pi_{s,t}, (\mathcal{V}, \tau)) \leq \sum_{i} |S_i| n/2^{i-1} = \Ohat(n^2/\eta + n \frac{\mathbf{dist}(s,t)}{\nu})
    \]
\end{itemize}

\paragraph{Bounding the Running Time.} Observe that for any vertex $x \in V$, that between any two times that it part of a graph $G'[Y]$ that a data structure $\mathcal{A}$ is invoked upon and of graph $G'[X] \subseteq G'[Y]$, the set $X$ is of at most half the size of $Y$. This follows by the definition of data structure $\mathcal{A}$ which whenever a sparse cut $(L,S,R)$ is output, continues on the graph $G'[R]$ where $R$ is larger than $L$ while no data structure is thereafter initialized on a graph containing any vertex in $S$. 

But if the SCC that some vertex $x$ is contained in halves in size every time between two data structures $\mathcal{A}$ are initialized upon $x$, then we have that $x$ participates in at most $\lg n$ data structures over the entire course of the algorithm. Since each edge $(x,y)$ or $(y,x)$ for any $y \in V$ is only present in the induced graph containing $x$, we have that no data structure that is not initialized on a graph with vertex set contain $x$ has $(x,y)$ or $(y,x)$ in its graph. Thus, every edge only participates in $\lg n$ graphs. 

Finally, we observe that the distance parameter $d$ that each data structure $\mathcal{A}$ is upper bounded by $\eta/2$. Thus, by the (super-)linear behavior of the function $T(m,n,d)$, we have that the total update time for all data structures in $\Ohat(T(m,n,\eta))$. Further, we have by \Cref{thm:SCCinDecrGraph} that the data structure maintaining $(\mathcal{V}, \tau)$ can be implemented in time $\Ohat(mn^{2/3})$. The time required for all remaining operations is subsumed in both bounds. 

\paragraph{Returning the Paths.} For any SCC $X$ in $\mathcal{V}$, we have that there is a data structure $\mathcal{A}$ on $G'[X]$ that allows for SCC queries. Since by our previous discussion each such data structure runs with $d$ at most $\frac{|X|\eta}{n}$ and each edge on the path has weight at most $\nu$ (recall that $G'$ only contains edges of small weight), we can return the path from data structure $\mathcal{A}$ on query.   
\end{proof}

\section{Conclusion}

In this article, we provide three new algorithms for decremental graphs: 1) a deterministic algorithm with running time $mn^{2/3+o(1)}$ that can answer SCC and SSR queries, 2) a deterministic algorithm with running time $n^{2+2/3+o(1)}$ that maintains SSSP and 3) a randomized (but adaptive) algorithm that maintains matchings with near optimal running time $\Tilde{O}(m)$. 

Each of these algorithms is a significant improvement for the problem at hand, and especially the former two algorithms improve on the long-standing upper bound of $O(mn)$ by Even and Shiloach \cite{EvenS}.

Our progress motivates the following related open questions:
\begin{itemize}
    \item Can we find \emph{deterministic} algorithms for SSR, SCC and directed SSSP that run in near-linear time? For SSR and SCC such an algorithm is known when randomization is allowed \cite{BernsteinPW19}. For directed SSSP, even obtaining a \emph{randomized} (non-adaptive) algorithm with near-linear update time is a major open question (although this goal has been achieved for very dense graphs \cite{nearOptDenseSSSP}). We also point out that while a randomized near-linear update time algorithm exists for undirected SSSP \cite{HenzingerKN14_focs}, even in this setting, the current best deterministic algorithms have running time $mn^{1/2 + o(1)}$ and $\tilde{O}(n^2)$ \cite{BernsteinC16, BernsteinChechikSparse, gutenberg2020deterministic, bernstein2020fully}.
    \item Can we obtain deterministic algorithms for the directed decremental $(1+\epsilon)$-approximate All-Pairs Shortest-Path problem with near-optimal total running time $\Tilde{O}(mn)$? Such an algorithm is currently only known in the randomized setting \cite{bernstein2016maintaining}, however, the best deterministic algorithm runs in total update time $\Tilde{O}(mn^{2})$ \cite{demetrescu2004new}.
\end{itemize}

\section{Acknowledgements}
We are very grateful to Mira Bernstein for showing us how to lower-bound the increase of the entropy potential function in the directed cut-matching game, which is a crucial step in our framework. The first author would like to thank David Wajc for helping him work through the black-box in \cite{Wajc19}, which allows us to convert our dynamic algorithm for fractional matching into one for integral matching.
We are grateful to Julia Chuzhoy for allowing us to apply the short-path oracle on expanders in \Cref{sec:oracle} to our framework, which is crucial to obtain almost-path-length query time. 
This result is by directly translating the same subroutine for undirected graphs shown in \cite{ChuzhoyS20_apsp} to directed graphs using our new primitives for directed graphs.

\newpage
\bibliographystyle{alpha}
\bibliography{bibliography}

\newpage

\appendix

\section{Proofs Omitted From Main Body of Conference Submission}
\label{sec:full_version_proofs}

In this section, we fill in some of the proofs that were omitted in the main body of the paper.

\subsection{Analysis of Algorithm \ref{alg:main}}
\label{sec:scc_analysis_long} \label{sec:SCC_analysis_long}

In this section, we give the complete analysis of our decremental SCC algorithm (Algorithm \ref{alg:main}) from Section \ref{sec:ingredient}. We in particular show that it satisfies the bounds of Theorem \ref{thm:main_scc}.

\paragraph{Correctness Analysis}
We need to show that after the algorithm finishes processing an update, the sets $C_1, ..., C_k \in \mathcal{C}$ are precisely the SCCs of $\gstar[\vstar \setminus \Shat]$. 

First we show that each $G[C_i]$ is strongly connected in $\gstar[\vstar \setminus \Shat]$. We know that $\rwitness$ in Line \ref{line:scchelper-rwitness} maintains a large witness $W_i$ for $C_i$, since otherwise it would have decomposed $C_i$ into smaller parts (Line \ref{line:scchelper-terminate}). Similarly, the fact that $\pathtowitness(C_i,W_i,\phistar)$ did not decompose $C_i$ in Line \ref{line:scchelper-truncate} implies that every vertex in $V(G)$ is strongly connected to $W_i$ (see invariant in Theorem \ref{thm:path-to-witness}). Since $W_i$ is itself strongly connected (because it is an expander), all of $C_i$ is strongly connected.

We now show by induction that no pair $C_i$, $C_j \in \mathcal{C}$ are strongly connected. This clearly holds at the beginning since $\mathcal{C}$ starts with a single element. Now, there are two lines in which $\mathcal{C}$ can change: Line \ref{line:scchelper-terminate} and Line \ref{line:scchelper-truncate}. In both cases, $C$ is replaced with $C'$ and $C''$, where $C' = L$ and $C'' = R$ for some vertex-cut $(L,S,R)$ in $C$. By definition of vertex-cut, $L$ and $R$ are not strongly connected in $\gstar[C \setminus S]$. Since $S$ is added to $\Shat$, it is easy to check that $L$ and $R$ will also be not strongly connected in $\gstar[\vstar \setminus \Shat]$. 

Finally, we show that the input conditions to each of the subroutines is satisfied. Firstly, all updates to $\shortoracle(W)$ come from changes made to $W$ by $\rwitness(G,\phistar)$; the latter always ensures that $W$ is a witness, so it is always a $1/n^{o(1)}$-expander, as required by $\shortoracle(W)$. Secondly, $\rwitness(G,\phistar)$ always maintains a large witness, so in $\pathtowitness(G,W,\phistar)$, we always obey the promise that $V(W) \geq n/2$. 

\paragraph{Update-Time Analysis}
For any $v \in \vstar$, define $X(v)$  to be the number of calls $\scchelper(G)$ for which $v \in V(G)$. The key to our analysis is to show that $X(v) = n^{o(1)} \ \forall v \in \vstar$.  To see this, consider any call $\scchelper(G)$ for which $v \in V(G)$, other than the initial call $\scchelper(\gstar)$. This call could only have been created in Line \ref{line:scchelper-terminate} or Line \ref{line:scchelper-truncate} of an earlier call $\scchelper(G')$. It is easy to see from the algorithm that the call $\scchelper(G')$ must have terminated as soon as $\scchelper(G)$ was created. We now complete the claim by arguing that $|V(G)| \leq (1-\alpha)|V(G')|$, for some parameter $\alpha = 1/n^{o(1)}$. To see this consider two cases. The first is that $\scchelper(G)$ was created in Line \ref{line:scchelper-terminate} of $\scchelper(G')$. In this case $G$ is equal to $G'[L]$ or $G'[R]$ for some vertex cut $(L,S,R)$ in $G'$. Theorem \ref{thm:robust witness} guarantees that this vertex-cut is $(1/n^{o(1)})$-balanced, so we have the desired $|V(G)| \leq (1-1/n^{o(1)})|V(G')|$. The second case is that $\scchelper(G)$ was create in Line \ref{line:scchelper-truncate} of $\scchelper(G')$. In this case $G = G'[L]$ for some vertex cut $(L,S,R)$ in $G'$; by definition of vertex-cut, we have $|L| \leq |V(G')|/2$, as desired.

Now consider the total running time of the three subroutines in $\scchelper(G)$: $\rwitness(G,\phistar)$, $\shortoracle(W)$ and $\pathtowitness(G,W,\phistar)$. The first subroutine has a total update time of $\Ohat(|E(G)|/(\phistar)^2) = \Ohat(|E(G)|\cdot n^{2/3})$, where $n = \vstar$. The second has total update time $\Ohat(|E(G)|)$ (Theorem \ref{thm:short-path-oracle}), but it must be reset every time $\rwitness(G,\phistar)$ enters a new phase (Line \ref{line:scchelper-new-phase}): since the total number of phases is $\Ohat(1/\phistar)$ (Theorem \ref{thm:robust witness}), the total update time for $\shortoracle$ in the call to $\scchelper(G)$ is $\Ohat(|E(G)|/\phistar) = \Ohat(|E(G)|n^{1/3})$. Finally, $\pathtowitness(G,W,\phistar)$ has total update time $\Ohat|E(G)|/\phistar$ (Theorem \ref{thm:path-to-witness}); multiplying by $\Ohat(1/\phistar)$ phases yields total update time $\Ohat(|E(G)|\cdot n^{2/3})$.

The total update time for a single call $\scchelper(G)$ is thus $\Ohat(|E(G)|\cdot n^{2/3})$. This can clearly be upper bounded by $\Ohat(n^{2/3}\sum_{v \in V(G)} \deg(v))$, where $\deg(v)$ is the degree of $v$ in the main graph $\gstar$ at time zero (before any deletions). It is thus easy to check that the \emph{total} update time of all $\scchelper(G)$ is at most $\Ohat(n^{2/3}\sum_{v \in \vstar} \deg(v) \cdot X(v))$. Since we showed at the beginning of the proof that $X(v) = n^{o(1)}$, we have a total update time of $\Ohat(n^{2/3} \cdot n^{o(1)} \cdot \sum_{v \in \vstar}\deg(v)) = \Ohat(mn^{2/3})$, as desired.

The final component of the total update time is the quantity $O(m|\Shat|)$ from Proposition \ref{thm:lacki}, where $|\Shat|$ refers to the largest size that $\Shat$ ever reaches. We complete the proof by showing that we always have $|\Shat| = \Ohat(n^{2/3})$. To see this, not that $\scchelper(G)$ only adds to $\Shat$ in lines \ref{line:scchelper-terminate} or \ref{line:scchelper-truncate}. In either case, it adds the set $S$ from a vertex cut $(L,S,R)$ and in either case the vertices in $L$ join a new call $\scchelper(G[L])$. Moreover, the vertex cut is always $\Ohat(\phistar)$-sparse (by Theorems \ref{thm:robust witness} and \ref{thm:path-to-witness}), so we have $|S| = \Ohat(|L|\phistar)$. Thus, if we give a vertex a token every time in participates in some new $\scchelper(G)$, then we can charge every vertex in $\Shat$ to $\Omegahat(1/\phistar)$ tokens. Since we have $X(v) = n^{o(1)}$ for all $v \in \vstar$, we can conclude that the total number of tokens is $\Ohat(n)$, so $|\Shat| = \Ohat(n \phistar) = n^{2/3}$. 

\paragraph{Query-Time Analysis}
By Proposition \ref{thm:lacki}, all we need to show is that each $\scchelper(G)$ has almost path-length query time. Say that the query is from $u$ to $v$ in some $G \in \mathcal(C)$. Let $W$ be the witness maintained by $\rwitness(G,\phistar)$ in Line \ref{line:scchelper-rwitness}. We use $\pathtowitness(G,W,\phistar)$ to find paths $P_{uW} = u \rightarrow w_1$ and $P_{Wv} = w_2 \rightarrow v$ for some $w_1, w_2 \in W$ (if $u$ or $v$ are in $W$, the corresponding path is empty.) We then use $\shortoracle(W)$ to find a path $P_W = w_1 \rightarrow w_2$ in $E(w)$; using the embedding of $W$ into $G$, $P_W$ can easily be transformed into a path $P_G$ in $E(G)$. We then return the path $P = P_{uW} \circ P_G \circ P_{Wv}$. It is not hard to show that the resulting query time is $\Ohat(|P|)$. The issue that the path $P$ might not be simple. We can always find a simple path $P'$ inside $P$, but if $|P'| << |P|$, then the time we spent is not proportional to $P'$.

To guarantee that we return a simple path in almost path-length query-time, we need a more clever query procedure. The details are in Section \ref{sec:query}.

\subsection{Proof of Theorem \ref{thm:decremental-matching}}
\label{sec:matching_proof_long}

In this section, we show that our main theorem for decremental matching (Theorem \ref{thm:decremental-matching}) follows easily from Algorithm $\rmatching$ (Lemma \ref{lem:rmatching}).

\begin{proof}[Proof of Theorem \ref{thm:decremental-matching}]
	We start with a deterministic algorithm that maintains a \emph{fractional matching}. The algorithm is as follows. Initialize $\mustar = n$. The algorithm runs $\rmatching(G, \mustar)$ to maintain the matching $M$. When \rmatching\ terminates, multiply $\mustar$ by $(1-\eps)$ and again run $\rmatching(G, \mustar)$. Terminate when $\mustar < 1$.
	
	The algorithm runs $\rmatching(G, \mustar)$ $O(\log(n)/\eps)$ times, which yields the desired total update time of $O(m\log^3(n)/\eps^4)$. If $M$ is the matching maintained by some $\rmatching(G, \mustar)$, then by Lemma \ref{lem:rmatching},  $\val(M) \geq \mustar(1-5\eps)$. If $\mustar = n$, we clearly have a $(1-5\eps)$-approximate matching. Else, since $\mustar < n$, we know that $\rmatching(G,\mustar/(1-\eps))$ already terminated, so $\mu(G) \leq \mustar$, so $M$ is a $(1-6\eps)$ approximate matching.
	
	Finally, to obtain an \emph{integral} matching, we plug in the above result to the black-box result of Wajc \cite{Wajc19} for converting dynamic fractional matching into dynamic integral matching. Consider Theorem 3.7 \cite{Wajc19}. We  have just showed an algorithm with $T_f(n,m) = O(\log(n)/\eps^4)$. We set $\gamma = 3$; As indicated in Section 2 of \cite{Wajc19}, we then have $T_c(n,m) = O(1)$ using a simple randomized algorithm for $3\delta$ edge-coloring that works against an adaptive adversary. Finally, we set $d = O(\log(1/\eps)/\eps)$ as in Lemma 4.5 of \cite{Wajc19}. By Theorem 3.7 of \cite{Wajc19}, the update time of the resulting algorithm is then $O(T_f(n,m) \cdot T_c(n,m) + \log(n/\eps) \cdot \gamma \cdot d/\eps^3 = O(\log^3(n)/\eps^4 + \log(n/\eps) \cdot \log(1/\eps) / \eps^4$. Since we always set $\eps = \Omega(1/n)$, our amortized update time for integral matching is the same as for fractional matching: $O(\log^3(n)/\eps^4)$.
	
	Note that as a result of this conversion the algorithm becomes randomized, but still works against an adaptive adversary. 
\end{proof}

\section{Implementation of Flow Subroutines}

\label{sec:flow}

Throughout this paper, we use several flow subroutines for various
contexts (e.g.~expander pruning, embedding robust witness, finding
approximate matching, etc.). All these flow algorithms are based on
the same techniques which is the \emph{bounded height variant} of
push-relabel and blocking flow algorithms. This idea was used explicitly
many times before (e.g.~\cite{LangR04,OrecchiaZ14,HenzingerRW17,SaranurakW19}).
Our contribution in this section is only to show a uniform presentation
that all of our flow subroutines can be implemented using the same
framework, and to give proofs for completeness.

We start with introducing notations in \Cref{sec:flow notation},
then we describe the guarantee of the bounded height variant of push-relabel
and blocking flow algorithms in \Cref{sec:bounded height flow}.
The common framework for edge-capacitated flow problems is described
in \Cref{sec:flow framework} and then we apply the framework to obtain several useful subroutines in \Cref{sec:flow subroutine}.
Similarly, the common framework for vertex-capacitated flow problems is presented in \Cref{sec:flow framework vertex} and so we obtain several useful subroutines via the framework in \Cref{sec:flow subroutine vertex}

\subsection{Flow Notations}

\label{sec:flow notation}

The notation below is slight adjusted from \cite{SaranurakW19} because
we work with directed graphs instead of undirected graphs.

A \emph{flow problem} $\Pi=(\Delta,T,c)$ on a \emph{directed} graph
$G=(V,E)$ is specified by a source function $\Delta:V\rightarrow\mathbb{R}_{\geq0}$,
a sink capacity function $T:V\rightarrow\mathbb{R}_{\geq0}$, and
edge capacities $c:E\rightarrow\mathbb{R}_{\geq0}$. We say that $\Pi$
is \emph{integral} if $\Delta:V\rightarrow\mathbb{Z}_{\geq0}$, $T:V\rightarrow\mathbb{Z}_{\geq0}$,
and $c:E\rightarrow\mathbb{Z}_{\geq0}$. More generally, for any number
$d\ge1$, we say that $\Pi$ is $1/d$-integral if $\Delta:V\rightarrow\frac{1}{d}\mathbb{Z}_{\geq0}$,
$T:V\rightarrow\frac{1}{d}\mathbb{Z}_{\geq0}$, and $c:E\rightarrow\frac{1}{d}\mathbb{Z}_{\geq0}$.
We use \emph{mass} to refer to the substance being routed. For a vertex
$v$, $\Delta(v)$ specifies the amount of mass initially placed on
$v$, and $T(v)$ specifies the capacity of $v$ as a sink. For an
edge $e$, $c(e)$ bounds how much mass can be routed along the edge.

A \emph{routing} (or \emph{flow}) $f:E\rightarrow\mathbb{R}_{\ge0}$
is $1/d$-integral if $f:E\rightarrow\frac{1}{d}\mathbb{Z}_{\geq0}$.
$f(u,v)>0$ means that mass is routed in the direction from $u$ to
$v$. If $f(u,v)=c(u,v)$, then we say $(u,v)$ is \emph{saturated}.
For convenience, for each directed edge $(u,v)$, we let $f(v,u)=-f(u,v)$.
For $A,B\subset V$, let $f(A,B)=\sum_{(u,v)\in E\cap(A\times B)}f(u,v)$
be the total mass routing directly from $A$ to $B$. Given $\Delta$,
we also treat $f$ as a function on vertices, where $f(v)=\Delta(v)+\sum_{u}f(u,v)=\Delta(v)+f(V,v)-f(v,V)$
is the amount of mass ending at $v$ after the routing $f$. If $f(v)\ge T(v)$,
then we say $v$'s sink is \emph{saturated}. 

We say that $f$ is a \emph{feasible routing/flow} for $\Pi$ if $f(u,v)\leq c(u,v)$
for each edge $e=(u,v)$ (i.e.~obey edge capacities), $f(v,V)-f(V,v)=\sum_{u}f(v,u)\leq\Delta(v)$
for each $v$ (i.e.~the net amount of mass routed away from a vertex
can be at most the amount of its initial mass), and $f(v)\leq T(v)$
for each $v$ (i.e.~no excess flow on each vertex).

Given a flow problem $\Pi=(\Delta,T,c)$, a \emph{pre-flow} $f$ is
a feasible routing for $\Pi$ except the condition $\forall v:f(v)\leq T(v)$
may not be satisfied. As pre-flow may not obey sink capacity on vertices,
we define the \emph{absorbed} mass on a vertex $v$ as $\ab_{f}(v)=\min(f(v),T(v))$.
We have $\ab_{f}(v)=T(v)$ iff $v$'s sink is saturated. The \emph{excess}
on $v$ is $\ex_{f}(v)=f(v)-\ab_{f}(v)$. From the definition, when
there is no excess, $\forall v:\ex_{f}(v)=0$, then $f$ is a feasible
flow for $\Pi$. Intuitively, we think of $\max\{\Delta(v)-T(v),0\}$
as \emph{initial excess} at $v$, and $\ex_{f}(v)$ is the excess
at $v$ after routing $f$. Similarly, we think of $\min\{\Delta(v),T(v)\}$
as \emph{initial absorbed mass} at $v$ and $\ab_{f}(v)$ is the absorbed
mass at $v$ after routing $f$. For any $S\subseteq V$, we usually
write $x(S)=\sum_{v\in S}x(v)$ where $x$ can be from $\{\Delta,T,\ex_{f},\ab_{f}\}$
(e.g.~$\Delta(S)$, $\ab_{f}(S)$). We omit the subscript whenever
it is clear.

For any directed path $P$, let $|P|$ denote the number of edges
in $P$. A \emph{path-decomposition} of a pre-flow $f$ is a collection
$\pset_{f}$ of directed paths with \emph{value} $\val(P)>0$ associated
with each path $P\in\pset_{f}$ and $\sum_{P\in\pset_{f}\mid P\ni e}\val(P)=f(e)$
for all $e\in E$.

\subsection{Bounded Height Push-Relabel and Blocking Flow}

\label{sec:bounded height flow}

The following proposition is the key algorithmic component for the
whole section. 
\begin{prop}
\label{prop: preflow and label}There is an algorithm that, given
a directed $n$-vertex $m$-edge graph $G=(V,E)$, a height parameter
$h\ge1$, and a flow problem $\Pi=(\Delta,T,c)$, returns a preflow
$f$ together with labels on vertices $l:V\rightarrow\{0,\dots,h\}$
such that: 
\begin{enumerate}
\item \label{enu:cross for}If $l(u)>l(v)+1$ and $(u,v)\in E$, then the
mass on $(u,v)$ is saturated, i.e.,~$f(u,v)=c(u,v)$. 
\item \label{enu:cross back}If $l(u)>l(v)+1$ and $(v,u)\in E$, then the
mass on $(v,u)$ is empty, i.e.,~$f(v,u)=0$. 
\item \label{enu:no excess}If $l(v)<h$, then $v$ has no excess, i.e.~$\ex_{f}(v)=0$. 
\item \label{enu:full absorb}If $l(v)>0$, then $v$'s sink is saturated,
i.e.~$\ab_{f}(v)=T(v)$. 
\item \label{enu:monotone}After routing $f$, excess does not increase
and absorbed mass never decreases, i.e.~$\ex_{f}(v)\le\max\{\Delta(v)-T(v),0\}$
and $\ab_{f}(v)\ge\min\{\Delta(v),T(v)\}$ for all $v$. 
\end{enumerate}
The algorithm takes at most $O(mh\log m)$ time. If $\Pi$ is $1/d$-integral
for some number $d\ge1$, then so is $f$. If $\Pi$ is integral,
$T(v)\ge\deg(v)$ for all $v\in V$, and the algorithm can access
the adjacency list of every vertex, then the running time can be reduced
to $O(\Delta(V)h)$. 

\end{prop}

\begin{proof}
The statement simply summarizes the output that one can obtain from
performing blocking flow computations for $\approx h$ rounds (instead
of $\approx n$ rounds as when we want to solve the exact max flow
problem).

We explain this idea in more detail. Let us create a graph $G'$ from
the graph $G$ by adding a super source vertex $s$ and a super sink
vertex $t$. For each $u\in V$, we add an edge $(s,u)$ with capacity
$\Delta(u)$. For each $u\in V$, we add an edge $(u,t)$ with capacity
$T(u)$. Then, we run blocking flow for at most $h+2$ rounds until
the (unweighted) distance between $s$ and $t$ in the residual graph
$G'_{f}$ of $G'$ is at least $h+2$. Each blocking flow computation
takes $O(m\log m)$ time (even when the flow problem $\Pi$ is fractional).
This running time is possible by using the \emph{link-cut tree data
structure}. (See the detail in Section 6 of \cite{SleatorT83}, Page
387-389.) So the total running time is $O(mh\log m)$. This completes
the running time analysis.

Let $f$ be the flow in $G'$ obtained after the blocking flow computations.
We can define the vertex labeling $l:V\cup\{s,t\}\rightarrow\{0,\dots,h\}$
as follows. For $0\le i\le h$, we set $l(u)=h-i+1$ where $i$ is
the (unweighted) distance between $s$ and $u$ in $G'_{f}$. For
all vertices $u$ whose distance from $s$ is more than $h$, we set
$l(u)=0$. By definition, $l(s)=h+1$. Also, as the distance from
$s$ to $t$ in $G'_{f}$ is $h+2$, for each label $i\in\{0,\dots,h\}$,
there must exist a vertex with label $i$. Observe that, for any edge
$(u,v)$ where $|l(u)-l(v)|>1$, the residual capacity in $G'_{f}$
must be $c_{f}(u,v)=0$. So this implies \Cref{enu:cross for} and
\Cref{enu:cross back}.

Note that we can view $f$ as a preflow on $G$ by restricting $f$
to only edges of $G$. Observe that the flow value on $(u,t)$ in
$G'$ corresponds to the absorbed flow at $u$ in $G$, i.e.~$f(u,t)=\ab_{f}(u)$.
Also, the residual capacity $c_{f}(s,u)$ of $(s,u)$ in $G'_{f}$
corresponds to the excess at $u$ after routing $f$, i.e.~$c_{f}(s,u)=\Delta(u)-f(s,u)=\ex_{f}(u)$.
So if $l(u)<h$, then $c_{f}(s,u)=0$ and so $\ex_{f}(u)=0$. Also,
if $l(u)>0$, then $c_{f}(u,t)=0$ and so $\ab_{f}(u)=T(u)$. This
implies \Cref{enu:no excess} and \Cref{enu:full absorb}.

In fact, our algorithm will do some simple preprocessing. For each
$u$, we will assume that we start with the initial flow that go through
$(s,u)$ and $(u,t)$ with value $\min\{\Delta(u),T(u)\}$. So initially
there are $\min\{\Delta(u),T(u)\}$ units of mass absorbed at $u$
and the initial excess at $u$ is $\max\{\Delta(u)-T(u),0\}$. As
the blocking flow computations have a property that flow value incident
to $s$ and $t$ never decreases. This implies \Cref{enu:monotone}.
This completes the correctness of the algorithm with running time
$O(mh\log m)$. Note that the implementation of blocking flow using
link-cut tree also have the guarantee that flow value on each edge
is $1/d$-integral if the given flow problem $\Pi$ is $1/d$-integral.

Lastly, we need to show if $\Pi$ is integral, $T(v)\ge\deg(v)$ for
all $v\in V$, and the algorithm can access the adjacency list of
every vertex, then the running time is $O(\Delta(V)h)$. However,
the proposition here is simply the summary of the output by the Unit
Flow algorithm by Henzinger Rao and Wang \cite{HenzingerRW17} (see
also \cite{SaranurakW19}) where Unit Flow is a bounded height variant
of push-relabel algorithms. 
\end{proof}
\begin{lem}
\label{lem:path decomposition}Given a preflow $f$ from the algorithm
from \Cref{prop: preflow and label}, any path decomposition $\pset_{f}$
of $f$ satisfies $\sum_{P\in\pset_{f}}\val(P)|P|\le(\Delta(V)-\ex_{f}(V))h$.
Moreover, if the flow problem $\Pi$ is $1/d$-integral, then a decomposition
$\pset_{f}$ can be computed in time $O(\sum_{P\in\pset_{f}}|P|)=O(d\Delta(V)h)$. 
\end{lem}

\begin{proof}
By the definition of path decomposition, we have $\sum_{P\in\pset_{f}\mid P\ni e}\val(P)=f(e)$
for all $e\in E$. So 
\[
\sum_{P\in\pset_{f}}\val(P)|P|=\sum_{e\in E}\sum_{P\in\pset_{f}\mid P\ni e}\val(P)=\sum_{e\in E}f(e).
\]
Now, we want to show that $\sum_{e\in E}f(e)\le(\Delta(V)-\ex_{f}(V))h$.
Consider the bounded height blocking-flow algorithm or push-relabel
algorithm. The algorithm always sends the flow along a path of length
at most $h$ in the residual graph. The total amount of mass that
is sent out of the initial place that mass was placed is at most $\Delta(V)-\ex_{f}(V)$.
So even the total flow value over all edges ``without flow cancellation''
must be at most $(\Delta(V)-\ex_{f}(V))h$. As $\sum_{e\in E}f(e)$
is the total flow value over all edges ``after flow cancellation'',
we conclude that $\sum_{e\in E}f(e)\le\Delta(V)h$.

We can find a path decomposition $\pset_{f}$ of a flow $f$ in time
$O(\sum_{P\in\pset_{f}}|P|)$ as follows. As the flow problem $\Pi$
is $1/d$-integral, the flow value of each edge $f(e)$ is also $1/d$-integral
by \Cref{prop: preflow and label}. Let $H$ be a graph induced
by edges $e$ where positive flow value $f(e)>0$. We make $H$ unweighted
by scaling up all edges in $H$ by a factor of $d$. Then, we add
a dummy source to $H$ and performing the depth-first search on $H$
from the dummy source. Whenever a search reaches a sink $u$ (i.e.
$\ab_{f}(u)>0$) or the search cannot proceed from $u$ (i.e. $\ex_{f}(u)>0$),
we backtrack and output the corresponding path $P$ excluding the
dummy source. Note that $P$ is a directed \emph{simple} path in $H$
and corresponds to a flow path of value $1/d$. We remove the path
$P$ from $H$ and repeat.

Observe that each edge in $H$ is read at most twice and so the total
time is subsumed by the total time for outputting all paths which
is $O(\sum_{P\in\pset_{f}}|P|)$. As $\val(P)\ge1/d$, so $O(\sum_{P\in\pset_{f}}|P|)\le O(\sum_{P\in\pset_{f}}d\cdot\val(P)|P|)=O(d\Delta(V)h)$. 
\end{proof}
For convenience, we will use the following notation throughout this
section. 
\begin{defn}
Given a vertex labeling $l:V\rightarrow\{0,\dots,h\}$, let $V_{i}=\{u\mid l(u)=i\}$
for each $i$. Also, we define $V_{\ge i}=\{u\mid l(u)\ge i\}$ and
$V_{>i},V_{\le i},V_{<i}$ are defined similarly. 
\end{defn}

When the input graph to \Cref{prop: preflow and label} is bipartite
and all source/sink vertices are only on the left/right respectively,
we can additionally guarantee that the vertices from each level alternate
between the left and right side of the bipartite graph $G$. 
\begin{lem}
\label{lem:bipartite label}If the input graph from \Cref{prop:
preflow and label} is a bipartite graph $G=(L,R,E)$ where $R$ has
no initial mass ($\Delta(R)=0$), and $L$ cannot absorb mass ($T(L)=0$),
then the vertex labeling $l$ from \Cref{prop: preflow and label}
has additional property that $V_{h},V_{h-2},V_{h-4},\dots\subseteq L$
and $V_{h-1},V_{h-3},V_{h-5},\dots\subseteq R$. 
\end{lem}

\begin{proof}
This guarantee follows immediately when we use blocking-flow-based
algorithms. For Unit Flow (push-relabel-based algorithm), we can guarantee
this using very simple modification: we initialize by assigning all
all vertices in $L$ a label $1$ and all vertices in $R$ a label
$0$, and when ever we relabel a vertex $l(u)\gets l(u)+1$, we instead
set $l(u)\gets l(u)+2$. All the invariants of the push-relabel-based
algorithm will still be maintained because the graph is bipartite. 
\end{proof}

\subsection{The Common Framework for Edge-capacitated Graphs}

\label{sec:flow framework}

Given a flow problem $\Pi=(\Delta,T,c)$ for a graph $G=(V,E)$, all
the algorithms in \Cref{sec:flow subroutine} starts by calling
\Cref{prop: preflow and label} with parameter $h$ and obtain a
preflow $f$ and a vertex labeling $l:V\rightarrow\{0,\dots h\}$.

If the total excess after routing $f$ is at most $z$, then the algorithm
just returns $f$ and we are done. Otherwise, the algorithm will return
one of the \emph{level-$i$ cuts} $(V_{\ge i},V_{<i})$ for some $0<i\le h$.
The only main task we need to prove in each algorithm is to show that
there exists an index $i$ such that the level-$i$ cut satisfies
the requirement of the lemma.

There are two common arguments that will be used by all flow algorithms
in \Cref{sec:flow subroutine}. The first one will be used to lower
bound the ``size'' of both $V_{h}$ and $V_{0}$ by the total excess.
In our algorithm, the outputted cut $S$ will be such that $S\supset V_{h}$
and $S\cap V_{0}=\emptyset$, so the proposition below will be useful
to prove the balance of $(S,V\setminus S)$. 
\begin{prop}
\label{prop:excess give lower bound}If $\Delta(V)\le T(V)$, then
$\Delta(V_{h})\ge\ex_{f}(V)$ and $T(V_{0})\ge\ex_{f}(V)$.
\end{prop}

\begin{proof}
First, note that $\ex_{f}(V)=\ex_{f}(V_{h})$ because all vertices
with level below $h$ has no excess by \Cref{enu:no excess} of
\Cref{prop: preflow and label}. Also, $\ex_{f}(V_{h})\le\Delta(V_{h})-T(V_{h})$
because excess does not increase (\Cref{enu:monotone} of \Cref{prop:
preflow and label}). So $\ex_{f}(V)\le\Delta(V_{h})$. Similarly,
note that $T(V_{0})-\ab_{f}(V_{0})=T(V)-\ab_{f}(V)$ because all vertices
with level above $0$ are fully absorbed by \Cref{enu:full absorb}
of \Cref{prop: preflow and label}. Also, we have $T(V)-\ab_{f}(V)\ge\Delta(V)-\ab_{f}(V)=\ex_{f}(V)$
by the assumption that $T(V)\ge\Delta(V)$ and by the definition of
$\ex_{f}(V)$. So $\ex_{f}(V)\le T(V_{0})$. 
\end{proof}
The next common argument is for upper bounding the total capacity
$c(E(V_{\ge i},V_{<i}))$ of the level cut $(V_{\ge i},V_{<i})$.
The argument used for bounding capacity of ``consecutive-level''
edges from $E(V_{i},V_{i-1})$ will be different in each algorithm.
But, the argument for bounding total capacity of edges that are not
from $E(V_{i},V_{i-1})$ will be the same and is stated as follow: 
\begin{prop}
\label{prop:cap of skip}$c(E(V_{\ge i},V_{<i})\setminus E(V_{i},V_{i-1}))\le\Delta(V_{\ge i})+f(V_{i-i},V_{i})-\ab_{f}(V_{\ge i})-\ex_{f}(V_{\ge i})$. 
\end{prop}

\begin{proof}
First, note that each $(u,v)\in E(V_{\ge i},V_{<i})\setminus E(V_{i},V_{i-1})$
is ``skipping levels'', i.e.~$l(u)-l(v)>1$. So, by \Cref{enu:cross
for} of \Cref{prop: preflow and label}, we have $f(u,v)=c(u,v)$.
So the total capacity $c(E(V_{\ge i},V_{<i})\setminus E(V_{i},V_{i-1}))$
is at most the total mass going out of $V_{\ge i}$. Observe that
the total mass in-coming into $V_{\ge i}$ is $f(V_{i-i},V_{i})$
because any other edges $(u,v)\in E(V_{<i},V_{\ge i})\setminus E(V_{i-i},V_{i})$
is ``skipping levels'' and so $f(u,v)=0$ by \Cref{enu:cross back}
of \Cref{prop: preflow and label}. Therefore, the total mass going
out of $V_{\ge i}$ is at most 
\begin{align*}
\underset{\textrm{initial mass}}{\underbrace{\Delta(V_{\ge i})}}+\underset{\textrm{incoming mass}}{\underbrace{f(V_{i-i},V_{i})}}-\underset{\textrm{absorbed mass}}{\underbrace{\ab_{f}(V_{\ge i})}}-\underset{\textrm{excess}}{\underbrace{\ex_{f}(V_{\ge i})}}
\end{align*}
\end{proof}

\subsection{Flow Subroutines for Edge-capacitated Graphs}

\label{sec:flow subroutine}

In this section, we state the flow subroutines that will be used by
several places throughout our paper. To state some lemmas below, we
also define degree and volume w.r.t. capacity function $c$. Given
a graph $G=(V,E)$ with edge capacity $c$, for $u\in V$, let $\deg^{c}(u)=\sum_{(u,v)}c(u,v)+\sum_{(v,u)}c(v,u)$
denote the weighted degree w.r.t.~$c$. For $S\subset V$, let $\vol^{c}(S)=\sum_{u\in S}\deg^{c}(u)$
be the volume of $S$ w.r.t.~$c$.

The following remark will be used repeatedly. 
\begin{rem}
\label{rem:excess at source and path decomposition} Let $f$ be a
preflow returned by any of the algorithms below in this section. Note
that every algorithm below starts by calling the algorithm from \Cref{prop:
preflow and label} with parameter $h$ on some graph. Note that any
vertex $v$ with positive excess, i.e.~$\ex_{f}(v)>0$, must have
positive initial excess, i.e.~$\Delta(v)-T(v)>0$. This follows from
\Cref{enu:monotone} of \Cref{prop: preflow and label}. By simply
scanning vertices with initial excess and removing excess after routing
$f$, we obtain a \emph{feasible flow} $f'$ of value $\Delta(V)-\ex_{f}(V)$
from $f$. The time to remove these excess is obviously subsumed by
the algorithm because the algorithm at least need to read all vertices
with initial excess. Moreover, we can obtain a path decomposition
of $f'$ in additional time $O(d\Delta(V)h)$ if the flow problem
is $1/d$-integral by \Cref{lem:path decomposition}. 
\end{rem}

\subsubsection{Local Flow}

The algorithm below either sends most of the flow or finds a balanced
sparse cut in \emph{local} time. We need that the given flow problem
is integral and each vertex can absorb mass at least by its degree. 
\begin{lem}
[Local Flow]\label{lem:local flow}There is a deterministic algorithm
that, given access to the adjacency list of every vertex of a directed
$m$-edge graph $G=(V,E)$, parameters $z\ge0$ and $h\ge1$, and
an \emph{integral} flow problem $\Pi=(\Delta,T,c)$ with total capacity
$C=\sum_{e\in E}c(e)$ where 
\begin{enumerate}
\item $\forall v,\Delta(v)\le\overline{\Delta}\deg(v)$ and $T(v)\ge\deg(v)$
where $\deg(v)$ denote an unweighted degree of $v$ in $G$. 
\end{enumerate}
in $O(\Delta(V)\cdot h)$ time either 
\begin{itemize}
\item returns a preflow $f$ with total excess $\ex_{f}(V)\le z$, or 
\item returns a set $S$ where $z/\overline{\Delta}<\vol(S)\le\Delta(V)$
and $c(E(S,V\setminus S))\le\Delta(S)-T(S)-z+\vol^{c}(S)\cdot\frac{10\log C}{h}$. 
\end{itemize}
\end{lem}

\begin{proof}
We call \Cref{prop: preflow and label} with parameter $h$. By
the assumption of the lemma, the running time is $O(\Delta(V)h)$.
Suppose that $\ex_{f}(V)>z$ otherwise we are done. By \Cref{prop:excess
give lower bound}, we know $\vol(V_{h})\overline{\Delta}\ge\Delta(V_{h})>z$
because $\forall v\in V,\Delta(v)\le\overline{\Delta}\deg(v)$. So
$\vol(V_{h})>z/\overline{\Delta}$. Also, observe that $\vol(V_{\ge1})\le\Delta(V)$.
This is because all vertices in $V_{\ge1}$ are fully absorbed by
\Cref{enu:full absorb} of \Cref{prop: preflow and label}, so
$T(V_{\ge1})\le\ab_{f}(V_{\ge1})\le\Delta(V)$, and because $T(V_{\ge1})\ge\vol(V_{\ge1})$
as $T(v)\ge\deg(v)$ for all $v$.

By the ball growing argument, there is an index $0<i\le h$ such that
$c(E(V_{i},V_{i-1})\cup E(V_{i-1},V_{i}))\le\min\{\vol^{c}(V_{\ge i}),\vol^{c}(V_{<i})\}\cdot\frac{10\log C}{h}$.
Otherwise, $\vol^{c}(V_{\ge1})\ge(1+\frac{10\log C}{h})^{h}>2C$ which
is a contradiction. We fix such $i$. Set $S=V_{\ge i}$. As $0<i\le h$,
so we have $z/\overline{\Delta}<\vol(S)\le\Delta(V)$. We have, by
\Cref{prop:cap of skip}, that 
\begin{align*}
c(E(S,V\setminus S)) & =c(E(V_{i},V_{i-1}))+c(E(V_{\ge i},V_{<i})\setminus E(V_{i},V_{i-1}))\\
 & \le c(E(V_{i},V_{i-1}))+\Delta(V_{\ge i})+f(V_{i-i},V_{i})-\ab_{f}(V_{\ge i})-\ex_{f}(V_{\ge i})\\
 & <\Delta(S)-T(S)-z+c(E(V_{i},V_{i-1})\cup E(V_{i-1},V_{i})).
\end{align*}
By the choice of $i$, we are done. 
\end{proof}

\subsubsection{Global Flow}

The algorithm below either sends most of the flow or finds a balanced
sparse cut when the flow problem is fractional. This is needed because
the capacity of edges will be fractional when we maintain the robust
witness. 
\begin{lem}
[Global Flow]\label{lem:global flow}There is a deterministic algorithm
that, given a directed $m$-edge graph $G=(V,E)$, excess parameter
$z\ge0$, a height parameter $h\ge1$, and a flow problem $\Pi=(\Delta,T,c)$
with total capacity $C=\sum_{e\in E}c(e)$ where 
\begin{enumerate}
\item $\Delta(V)\le T(V)$, 
\item $\forall v\in V,\Delta(v),T(v)\le1$, 
\end{enumerate}
in $O(mh\log m)$ time, either 
\begin{itemize}
\item returns a preflow $f$ with total excess $\ex_{f}(V)\le z$, or 
\item returns a set $S\subset V$ where $|S|,|V\setminus S|>z$ and $c(E(S,V\setminus S))\le\Delta(S)-T(S)-z+\min\{\vol^{c}(S),\vol^{c}(V\setminus S)\}\cdot\frac{10\log C}{h}$. 
\end{itemize}
\end{lem}

\begin{proof}
We call \Cref{prop: preflow and label} with parameter $h$ in $O(mh\log m)$.
Suppose that $\ex_{f}(V)>z$. By \Cref{prop:excess give lower bound},
we know $|V_{h}|\ge\Delta(V_{h})>z$ and $|V_{0}|\ge T(V_{0})>z$
because $\forall v\in V,\Delta(v),T(v)\le1$. By the ball growing
argument applying in two directions, there is an index $0<i\le h$
such that $c(E(V_{i},V_{i-1})\cup E(V_{i-1},V_{i}))\le\min\{\vol^{c}(V_{\ge i}),\vol^{c}(V_{<i})\}\cdot\frac{10\log C}{h}$.
Otherwise, $\vol^{c}(V_{\ge h/2})\ge(1+\frac{10\log C}{h})^{h/2}>2C$
which is a contradiction. We fix such $i$. Set $S=V_{\ge i}$. As
$0<i\le h$, so we have $|S|,|V\setminus S|>z$. We have, by \Cref{prop:cap
of skip}, that 
\begin{align*}
c(E(S,V\setminus S)) & =c(E(V_{i},V_{i-1}))+c(E(V_{\ge i},V_{<i})\setminus E(V_{i},V_{i-1}))\\
 & \le c(E(V_{i},V_{i-1}))+\Delta(V_{\ge i})+f(V_{i-i},V_{i})-\ab_{f}(V_{\ge i})-\ex_{f}(V_{\ge i})\\
 & <\Delta(S)-T(S)-z+c(E(V_{i},V_{i-1})\cup E(V_{i-1},V_{i})).
\end{align*}
By the choice of $i$, we are done. 
\end{proof}

\subsubsection{Flow for Matching}

\label{sec:flow_for_matching}

The algorithm below is for computing approximate bipartite matching.
That is why the graph $G=(L,R,E)$ is bipartite and only has edges
from $L$ to $R$. The algorithm either send at least $\Delta(V)-z$
flow (i.e.~large fractional matching) or find a cut $S$ such that
the residual capacity is at most $2\frac{\Delta(V)-z}{h}$. So this
gives a $2/h$-approximation algorithm for bipartite matching. 
\begin{lem}
[Global Flow for Matchings]\label{lem:global flow matching}There
is a deterministic algorithm that, given a directed bipartite $m$-edge
graph $G=(V=(L,R),E)$ where $E\subseteq L\times R$, an excess parameter
$z\ge0$, a height parameter $h\ge1$, and a flow problem $\Pi=(\Delta,T,c)$
and the following holds 
\begin{enumerate}
\item $\Delta(R)=0$, $T(L)=0$, $\Delta(L)\le T(R)$, 
\item $\forall v\in V,\Delta(v),T(v)\le1$, 
\end{enumerate}
in $O(mh\log m)$ time, either 
\begin{itemize}
\item returns a preflow $f$ with total excess $\ex_{f}(L\cup R)\le z$,
or 
\item returns a set $S\subset V(G)$ where $|S|,|V\setminus S|>z$ and $c(E(S,V\setminus S))\le\Delta(S)-T(S)-z+2\cdot\frac{\Delta(V)-z}{h}$. 
\end{itemize}
\end{lem}

\begin{proof}
We call \Cref{prop: preflow and label} with parameter $h$ in $O(mh\log m)$.
Suppose that $\ex_{f}(V)>z$ otherwise we are done. By \Cref{prop:excess
give lower bound}, we know $|V_{h}|\ge\Delta(V_{h})>z$ and $|V_{0}|\ge T(V_{0})>z$
because $\forall v\in V,\Delta(v),T(v)\le1$. By \Cref{lem:bipartite
label}, $V_{h},V_{h-2},V_{h-4},\dots$ are subsets of $L$ and $V_{h-1},V_{h-3},V_{h-5},\dots$
are subsets of $R$. So $\sum_{i\ge1}f(V_{h-2i},V_{h-2i+1})\le f(L,R)\le\Delta(V)-z$.
So there is $1\le i\le h/2$ such that $f(V_{h-2i},V_{h-2i+1})\le2(\Delta(V)-z)/h$.
Fix such $i$ and set $S=V_{>h-2i}$. As $S\supseteq V_{h}$ and $S\cap V_{0}=\emptyset$
we have that $|S|,|V\setminus S|>z$. We have, by \Cref{prop:cap
of skip}, that 
\begin{align*}
c(E(S,V\setminus S)) & =c(E(V_{h-2i+1},V_{h-2i}))+c(E(V_{>h-2i},V_{\le h-2i})\setminus E(V_{h-2i+1},V_{h-2i}))\\
 & \le\Delta(V_{>h-2i})+f(V_{h-2i},V_{h-2i+1})-\ab_{f}(V_{>h-2i})-\ex_{f}(V_{>h-2i})\\
 & <\Delta(S)-T(S)-z+\frac{2(\Delta(V)-z)}{h}
\end{align*}
where the first inequality is because $E(V_{h-2i+1},V_{h-2i})=\emptyset$
as $V_{h-2i+1}\subset R$ and $V_{h-2i}\subset L$ and the second
inequality is by the choice of $i$. 
\end{proof}
This immediately implies the subroutine that we need in \Cref{sec:matching}.
We simply plug in the parameters correctly.

\MatchingOrCutLemma* 
\begin{proof}
W.l.o.g. we can assume that $|L|\le|R|$ and then we treat edges in
$G$ are directed edges from $L$ to $R$. Let $h=2/\epsilon$ and
$z=n-\mu(1-\epsilon)$. For each $v\in L$, let $\Delta(v)=1$ and
$T(v)=0$. For each $v\in R$, let $\Delta(v)=0$ and $T(v)=1$. Let
$\Pi=(\Delta,T,\kappa)$. We call the algorithm from \Cref{lem:global
flow matching} with $(G,z,h,\Pi)$ as input. Observe that the input
satisfies all the conditions of \Cref{lem:global flow matching}.

If \Cref{lem:global flow matching} returns a preflow $f$ with
excess at most $z$, then this means that we obtain a flow of size
at least $\Delta(V)-z=\mu(1-\epsilon)$ by \Cref{rem:excess at source
and path decomposition}. Obviously, $f(e)\le\kappa(e)$ for all
$e$. If \Cref{lem:global flow matching} returns a set $S\subset V(G)$
such that $\kappa(E(S,V\setminus S))\le\Delta(S)-T(S)-z+2\cdot\frac{\Delta(V)-z}{h}$.
Let $S_{L}=S\cap L$ and $S_{R}=S\cap R$. Note that $\Delta(S)=|S_{L}|$,
$T(S)=|S_{R}|$, $\kappa(E(S,V\setminus S))=\kappa(S_{L},R\setminus S_{R})$,
and $2\frac{\Delta(V)-z}{h}-z=\epsilon\mu(1-\epsilon)-(n-\mu(1-\epsilon))=\mu(1-\epsilon^{2})-n.$
So we have 
\[
\kappa(S_{L},R\setminus S_{R})\le|S_{L}|-|S_{R}|+\mu-n
\]
as desired. 
\end{proof}

\subsection{The Common Framework for Vertex-capacitated Graphs and Hypergraphs}

\label{sec:flow framework vertex}

In this section, our goal is to build a framework for solving flow
problems on vertex-capacitated graphs and hypergraphs. Throughout
this section, we will work on an incidence graph $G$ (i.e.~the bipartite
representation) of a directed hypergraph $H$. That is, $G=(V\cup\Vinf,E)$
is a bipartite such that vertices in $V$, representing regular vertices
in $H$, have finite capacities and vertices in $\Vinf$, representing
hyperedges in $H$, have infinite capacities. Edges in $G$ are directed
and can go either from $V$ to $\Vinf$ or from $\Vinf$ to $V$.
As we will see in \Cref{sec:flow subroutine vertex}, flow algorithms
on graphs of this form will imply flow algorithms on vertex-capacitated
graphs and hypergraphs.

\paragraph{Vertex-capacitated flow.}

A\emph{ vertex-capacitated flow problem} $\Pi=(\Delta,T,\kappa)$
on a graph $G=(V\cup\Vinf,E)$ is defined by a source function $\Delta:V\rightarrow\mathbb{R}_{\geq0}$,
a sink capacity function $T:V\rightarrow\mathbb{R}_{\geq0}$, and
a vertex capacities $\kappa:V\rightarrow\mathbb{R}_{\ge0}$. Note
that all these functions are defined only on $V$ and not on $\Vinf$.
We say that a routing $f$ is a\emph{ feasible flow} for $\Pi$ if
the total out-going mass from $v$ is at most $f(v,V)\leq\kappa(v)$
for $v\in V$ (i.e.~obey vertex capacities)%

, $\sum_{u}f(v,u)\leq\Delta(v)$ for each $v\in V$, and $f(v)\leq T(v)$
for each $v\in V$. Note that the last two conditions are the same
for edge-capacitated feasible flow. Again, a pre-flow $f$ is a feasible
flow for $\Pi$ except the condition $f(v)\leq T(v)$ may not be satisfied.
For any vertex $v\in V\cup\Vinf$, $\ab_{f}(v)=\min(f(v),T(v))$ and
$\ex_{f}(v)=f(v)-\ab_{f}(v)$ are defined as before.

The goal of this section to prove the following theorem.
\begin{thm}
\label{thm:main vertex flow}There is an algorithm that, given a vertex-capacitated
flow problem $\Pi=(\Delta,T,\kappa)$ on $G=(V\cup\Vinf,E)$ with
$m$ edges, an excess parameter $z\ge0$, and a height parameter $h\ge2$,
with the following conditions
\begin{itemize}
\item $\Delta(v),T(v)\le\kappa(v)/2$ for all $v\in V$ 
\item $\kappa(v)\le\kappa(V)/2$ for all $\in V$
\item $\Delta(V)\le T(V)$
\end{itemize}
then in $O(mh\log m)$ time either returns 
\begin{itemize}
\item a feasible preflow $f$ for $\Pi$ with total excess $\ex_{f}(V)\le z$;
or 
\item a partition $(L,S,R)$ of $V$ such that no path from a vertex in
$L$ to another vertex in $R$ in $G-S$ such that $\min\left\{ \Delta(L\cup S),T(R\cup S)\right\} >z$
and 
\begin{align*}
\kappa(S) & \le\Delta(L\cup S),T(R\cup S)-z+O(\frac{\log\kappa(V)}{h})\cdot\min\left\{ \kappa(L\cup S),\kappa(S\cup R)\right\} \\
 & \le2(\min\left\{ \Delta(L),T(R)\right\} -z)+O(\frac{\log\kappa(V)}{h})\cdot\min\left\{ \kappa(L\cup S),\kappa(S\cup R)\right\} 
\end{align*}
\end{itemize}
If $\Pi$ is $1/d$-integral and the preflow $f$ is returned, then,
for any $\eps_{\len}\ge0$, in additional time $O(d\Delta(V)h)$,
we can compute a feasible flow for $\Pi$ of value at least $(1-\eps_{\len})\cdot(\Delta(V)-ex_{f}(V))$
together with its path decomposition where each path is a simple path
containing at most $O(h/\eps_{\len})$ edges. 
\end{thm}

We will use both bounds on $\kappa(S)$ in \Cref{thm:main vertex flow}
above. Note that if $L,R\neq\emptyset$, then $(L,S,R)$ is an out-vertex-cut
in the hypergraph $H$ corresponding to $G$. The remaining of this
section is for proving \Cref{thm:main vertex flow}. 

\paragraph{Reduction to edge-capacitated flow.}

To prove \Cref{thm:main vertex flow}, we will reduce the vertex-capacitated
flow problem to a edge-capacitated one using the standard vertex-splitting
reduction. Given a vertex-capacitated flow problem $\Pi=(\Delta,T,\kappa)$
on a vertex-capacitated graph $G=(V\cup\Vinf,E,\kappa)$, we will
define a flow problem $\Pi'=(\Delta',T',c)$ on an edge-capacitated
graph $G'=(V',E',c)$ where $G'$ is constructed as follows. We have
$V'=V^{in}\dot{\cup}V^{out}\dot{\cup}\Vinf$. For each vertex $v\in V$,
we create $v_{in}\in V^{in}$ and $v_{out}\in V^{out}$ and add a
directed edge $(v_{in},v_{out})$ with capacity $c(v_{in},v_{out})=\kappa(v)$.
Each $x\in\Vinf$ appears in $G'$ too. For each edge $(v,x)\in E\cap(V\times\Vinf)$,
we add a directed edge $(v_{out},x)$ with capacity $c(v_{out},x)=\infty$.
For each edge $(x,v)\in E\cap(\Vinf\times V)$, we add a directed
edge $(x,v_{in})$ with capacity $c(x,v_{in})=\infty$. Observe that
$E(G')=E_{G'}(V^{in},V^{out})\cup E_{G'}(V^{out},\Vinf)\cup E_{G'}(\Vinf,V^{in})$
and only edges in $E_{G'}(V^{in},V^{out})$ have finite capacities.
Finally, for each $v\in V$, we set $\Delta'(v_{in})=\Delta(v)$ and
$T'(v_{out})=T(v)$. Other vertices $v$ in $G'$ has $\Delta'(v)=T'(v)=0$.
This completes the description of the edge-capacitated flow problem
$\Pi'$. 

We define initial mass $\Delta'$ only on $V^{in}$ and sink capacity
$T'$ only on $V^{out}$ because it implies the following important
property of (feasible) preflow in $G'$:
\begin{prop}
\label{prop:bound mass via vertex cap}Let $f'$ be a feasible preflow
w.r.t. $\Pi'$. Then, for any $v'\in\{v_{in},v_{out}\}$ in $G'$
corresponding to a vertex $v$ in $G$, the total mass $f'(v',V')$
going out from $v'$ and the total mass $f'(V',v')$ coming into $v'$
can be at most $\kappa(v)$.
\end{prop}

\begin{proof}
We first consider $v_{in}$. As $(v_{in},v_{out})$ is the only outgoing
edge from $v_{in}$ in $G'$, we have $f'(v_{in},V')\le c(v_{in},v_{out})=\kappa(v)$.
Next, as total flow in-coming to $v_{in}$ can be at most total flow
out-going from $v_{in}$ plus the sink capacity at $v_{in}$, we have
$f'(V',v_{in})\le f'(v_{in},V')+T'(v_{in})=\kappa(v)+0$. The argument
is symmetric for $v_{out}$. As $(v_{in},v_{out})$ is the only in-coming
to $v_{out}$ in $G'$, we have $f'(V',v_{out})\le c(v_{in},v_{out})=\kappa(v)$.
As total flow out-going from $v_{out}$ can be at most total flow
in-coming to $v_{out}$ plus the initial mass at $v_{out}$, we have
$f'(v_{out},V')\le f'(V',v_{out})+\Delta'(v_{out})=\kappa(v)+0$. 
\end{proof}
Given a flow problem $\Pi=(\Delta,T,\kappa)$ on a vertex-capacitated
graph $G=(V\cup\Vinf,E)$, let $\Pi'=(\Delta',T',c)$ be the corresponding
edge-capacitated flow problem on $G'=(V',E')$ defined above. The
first step is to simply call \Cref{prop: preflow and label} with parameter
$h$ and obtain in $O(mh\log m)$ time a preflow $f'$ in $G'$ and
a vertex labeling $l:V'\rightarrow\{0,\dots,h\}$. In fact, we need
a natural preprocessing step so that the preflow $f'$ satisfies some
technical property.
\begin{prop}
\label{prop:if ex then full}The preflow $f'$ returned by \Cref{prop: preflow and label}
is such that, if $\ex_{f'}(v_{in})>0$, then $\ab_{f'}(v_{out})=T'(v_{out})$. 
\end{prop}

\begin{proof}
When we invoke the bounded-height blocking flow algorithm in \Cref{prop: preflow and label}.
The algorithm will first push mass of $\min\{\Delta'(v_{in}),T'(v_{out})\}=\min\{\Delta(v),T(v)\}$
units through $(v_{in},v_{out})$ for every $v\in V$. So this means
if there is still excess at $v_{in}$ (i.e.~$\ex_{f'}(v_{in})>0$),
then this means that $v_{out}$ is fully absorbed (i.e.~$\ab_{f'}(v_{out})=T'(v_{out})$).
By \Cref{prop: preflow and label}(\Cref{enu:monotone}), the absorbed
mass never decrease. So this remains true when the algorithm return
the final preflow $f'$.
\end{proof}
Define $V'_{i}=\{u\in V'\mid l(u)=i\}$ and we write $V_{i}^{in}=V'_{i}\cap V^{in}$,
$V_{i}^{out}=V'_{i}\cap V^{out}$, and $\Vinf_{i}=V'_{i}\cap\Vinf$.
Also, let $V'_{\ge i}=\{u\in V'\mid l(u)\ge i\}$. $V'_{>i},V'_{<i},V_{>i}^{in}$
and so on are defined similarly. It remains to show how to obtain
the output required by \Cref{thm:main vertex flow} from the preflow
$f'$ in $G'$ and the labeling $l$ of $V'$.%

\paragraph{Defining vertex-capacitated preflow.}

Below, we define a vertex-capacitated preflow $f$ from the feasible
preflow $f'$ for $\Pi'$ in $G'$ and show that $f$ satisfies several
basic properties.
\begin{prop}
\label{prop:preflow vertex}Given a feasible preflow $f'$ for $\Pi'$
in $G'$, we let $f$ be induced by $f'$ after contracting each $v_{in}\in V^{in}$
and $v_{out}\in V^{out}$ into $v\in V$. We have the following:
\begin{enumerate}
\item $f$ is a feasible preflow for $\Pi$ in $G$.
\item For each $v\in V$, $\ab_{f}(v)=\ab_{f'}(v_{out})$ and $\ex_{f}(v)=\ex_{f'}(v_{in})$.%

\end{enumerate}
\end{prop}

\begin{proof}
(1): Fix any vertex $v\in V$. Observe that $f'(v_{out},V')=f(v,V)$
and $f'(V',v_{in})=f(V,v)$. So we have $f(v,V)=f'(v_{out},V')\le\kappa(v)$,
by \Cref{prop:bound mass via vertex cap}, i.e. $f$ obeys vertex capacities.
Next, as $f'$ is feasible, we have $f'(v_{out},V')-f'(V',v_{out})\le\Delta'(v_{out})=0$
and $f'(v_{in},V')-f'(V',v_{in})\le\Delta'(v_{in})=\Delta(v)$. As
$f'(V',v_{out})=f'(v_{in},V')=f'(v_{in},v_{out})$, we have $f'(v_{out},V')-f'(V',v_{in})\le\Delta(v)$.
Therefore, we have $f(v,V)-f(V,v)\leq\Delta(v)$, i.e.~the net amount
of mass routed away from a vertex is at most the amount of its initial
mass as desired. 

(2): First, note that $\ab_{f'}(v_{in})=0$ and $\ex_{f'}(v_{out})=0$
because $T'(v_{in})=\Delta'(v_{out})=0$. By \Cref{prop:if ex then full},
for any vertex $v\in V$ where $\ex_{f'}(v_{in})>0$, we have $\ab_{f'}(v_{out})=T'(v_{out})=T(v)$.
So if $\ex_{f'}(v_{in})>0$, no excess mass at $v_{in}$ can be furthur
absorbed after contracting $v_{in}$ and $v_{out}$ and so $\ab_{f}(v)=\ab_{f'}(v_{out})$
and $\ex_{f}(v)=\ex_{f'}(v_{in})$. If $\ex_{f'}(v_{in})=0$, then
$\ex_{f}(v)=\ex_{f'}(v_{in})=0$ and $\ab_{f}(v)=\ab_{f'}(v_{out})$
because there is no excess to be absorbed after contraction.%

\end{proof}

\paragraph{Defining vertex cuts.}

Fix $i$ where $0<i<h$. Instead of defining a unique vertex cut associated
with level $i$, we will define a partition $(L_{i},S_{i},R_{i},F_{i})$
of $V$ (i.e. the set of vertices with finite capacities in $G$).
The partition will satisfy the following property: for any partition
$(F_{i}^{L},F_{i}^{R})$ of $F_{i}$, there is no path from a vertex
in $L_{i}\cup F_{i}^{L}$ to another vertex in $R_{i}\cup F_{i}^{R}$
in $G-S_{i}$. That is, $(L_{i}\cup F_{i}^{L},S_{i},R_{i}\cup F_{i}^{R})$
is indeed an \emph{out-vertex-cut} of a hypergraph corresponding to
$G$, if $L_{i}\cup F_{i}^{L}$ and $R_{i}\cup F_{i}^{R}$ are non-empty.
We call the set $F_{i}$ a \emph{free} set because it can be freely
distributed into $L_{i}$ and $R_{i}$ and still gives us a vertex-cut.
At the end, we will divide $F_{i}$ into $F_{i}^{L}$ and $F_{i}^{R}$
equally.

Below, we describe $(L_{i},S_{i},R_{i},F_{i})$. First, $S_{i}$ is
a union of three sets: $S_{i}=S_{i}^{\skp}\cup S_{i}^{\finite}\cup S_{i}^{\infinite}$
where 
\begin{align*}
S_{i}^{\skp} & =\{v\in V\mid v_{in}\in V{}_{>i}^{in}\text{ and }v_{out}\in V{}_{<i}^{out}\},\\
S_{i}^{\finite} & =\{v\in V\mid v_{in}\in V_{i}^{in}\text{ or }v_{out}\in V_{i}^{out}\},\\
S_{i}^{\infinite} & =\{v\in V\mid(x,v_{in})\in E_{G'}(\Vinf_{i},V_{i+1}^{in})\text{ or }(v_{out},x)\in E_{G'}(V_{i+1}^{out},\Vinf_{i})\}.
\end{align*}
Note that $S_{i}^{\skp}$, $S_{i}^{\finite}$ and $S_{i}^{\infinite}$
may not be disjoint. Next, we define 
\begin{align*}
\Lbar_{i} & =\{v\in V\mid v_{out}\in V'_{>i}\}\setminus S_{i},\\
\Rbar_{i} & =\{v\in V\mid v_{in}\in V'_{<i}\}\setminus S_{i},\\
F_{i} & =\Lbar_{i}\cap\Rbar_{i}=\{v\in V\mid v_{out}\in V'_{>i}\text{ and }v_{in}\in V'_{<i}\}\setminus S_{i},\\
L_{i} & =\Lbar_{i}\setminus\Rbar_{i}=\{v\in V\mid v_{in},v_{out}\in V'_{>i}\}\setminus S_{i},\\
R_{i} & =\Rbar_{i}\setminus\Lbar_{i}=\{v\in V\mid v_{in},v_{out}\in V'_{<i}\}\setminus S_{i}.
\end{align*}
To see why $L_{i}=\{v\in V\mid v_{in},v_{out}\in V'_{>i}\}\setminus S_{i}$,
note that $\Lbar_{i}\setminus\Rbar_{i}=\{v\in V\mid v_{in}\in V'_{>i}\text{ and }v_{out}\in V'_{\ge i}\}\setminus S_{i}$
from definition, but this set equals $\{v\in V\mid v_{in},v_{out}\in V'_{>i}\}\setminus S_{i}$
because if $v_{out}\in V'_{i}$, then $v\in S_{i}$. The similar argument
holds for $R_{i}$. We start with a simple observation.
\begin{prop}
\label{prop:define get partition }$(L_{i},S_{i},R_{i},F_{i})$ is
indeed a partition of $V$.
\end{prop}

\begin{proof}
$L_{i},R_{i},F_{i}$ are mutually disjoint as they partition $\Lbar_{i}\cup\Rbar_{i}$.
Also, $L_{i},R_{i},F_{i}$ are disjoint from $S_{i}$ by definition.
Therefore, all sets are mutually disjoint. Next, we prove that $L_{i}\cup S_{i}\cup R_{i}\cup F_{i}=V$.
Indeed, $v\notin\Lbar_{i}\cup\Rbar_{i}=L_{i}\cup R_{i}\cup F_{i}$,
then $v_{in}\in V{}_{\ge i}^{in}\text{ and }v_{out}\in V{}_{\le i}^{out}$.
If $v\notin S_{i}^{\finite}$, then we must have $v_{in}\in V{}_{>i}^{in}\text{ and }v_{out}\in V{}_{<i}^{out}$,
but this implies that $v\in S_{i}^{\skp}$. 
\end{proof}
Next, we prove that, for any partition $(F_{i}^{L},F_{i}^{R})$ of
$F_{i}$, there is no path from a vertex in $L_{i}\cup F_{i}^{L}$
to another vertex in $R_{i}\cup F_{i}^{R}$ in $G-S_{i}$. This is
implied by the following:
\begin{lem}
Let $a\in\Lbar_{i}$ and $b\in\Rbar_{i}\setminus\{a\}$. Then, there
is no directed $a$-$b$ path in $G-S_{i}$. 
\end{lem}

\begin{proof}
Suppose that there is an $a$-$b$ path $P$ in $G$. We will prove
that there must exist a vertex $v\in S_{i}\cap P$. Note that the
path $P$ corresponds to the path $P'$ from $a_{out}$ to $b_{in}$
in $G'$ where each vertices $v\in V\cup P$ is split into $v_{in}$
and $v_{out}$. By definition of $\Lbar_{i}$ and $\Rbar_{i}$, we
have that $a_{out}\in V'_{>i}$ and $b_{in}\in V'_{<i}$. Since one
endpoint of $P'$ is in $V'_{>i}$ and another is in $V'_{<i}$, there
must exist an edge $(u',w')\in P'$ where $u'\in V'_{>i}$ and $w'\in V'_{\le i}$.
There are two cases.

First, suppose that $w'\in V'_{<i}$. We claim that $(u',w')\in E_{G'}(V^{in},V^{out})$
which implies that $(u',w')=(v_{in},v_{out})$ for some $v\in V$
and so $v\in S_{i}^{\skp}$. To see the claim, suppose otherwise
that $(u',w')\notin E_{G'}(V^{in},V^{out})$, we have $(u',w')\in E_{G'}(\Vinf,V^{in})\cup E_{G'}(V^{out},\Vinf)$
and thus $c(u',w')=\infty$. As $(u',w')$ is an edge that skips level
$i$,  by \Cref{prop: preflow and label}(\ref{enu:cross for}), $f'(u',w')=c(u',w')$.
But this is impossible because $c(u',w')=\infty$ while $f'(u',w')$
is finite.

Second, suppose that $w'\in V'_{i}$. If $w'=v_{in}$ or $w'=v_{out}$
for some $v\in V$, then $v\in S_{i}^{\finite}$ and we are done.
So let us assume that $w'\in\Vinf$ and so $(u',w')\in E_{G}(V'_{>i},\Vinf_{i})$.
We claim that $u'=v_{out}\in V_{i+1}^{out}$ for some $v\in V$ and
so $v\in S_{i}^{\infinite}$.\footnote{Note that we did not exploit the fact that $S_{i}^{\infinite}$ contains
$v$ where $(x,v_{in})\in E_{G'}(\Vinf_{i},V_{i+1}^{in})$. We will
use this property of $S_{i}^{\infinite}$ later in \Cref{prop:cap of skip vertex}.} To see the claim, observe that $u'\in V^{out}$ because $w'\in\Vinf$
only have incoming edges from $V^{out}$. Thus, the edge $(u',w')\in E_{G'}(V^{out},\Vinf)$
and so $c(u',w')=\infty$. Again, by \Cref{prop: preflow and label}(\ref{enu:cross for}),
$(u',w')$ cannot skip levels because it cannot be saturated by the
flow and so $u'\in V'_{i+1}$ as $w'\in V'_{i}$. Thus, $u'\in V^{out}\cap V'_{i+1}=V_{i+1}^{out}$
as claimed. In all cases considered, there exists $v\in S_{i}$ in
$P$ as desired.
\end{proof}
To bound the size of vertex cut, we bound $\kappa(S_{i})\le\kappa(S_{i}^{\skp})+\kappa(S_{i}^{\infinite})+\kappa(S_{i}^{\finite})$.
We show that there must exist a level $i$ where $\kappa(S_{i}^{\infinite})+\kappa(S_{i}^{\finite})$
is small.
\begin{lem}
\label{prop:cap of nonskip vertex}There exists $i$ where $h>i>0$
where
\[
\kappa(S_{i}^{\infinite})+\kappa(S_{i}^{\finite})\le\frac{200\log\kappa(V)}{h}\cdot\min\left\{ \kappa(V\setminus R_{i}),\kappa(V\setminus L_{i})\right\} .
\]
\end{lem}

\begin{proof}
There are several steps. First, we prove that 
\begin{equation}
\bigcup_{j\ge i}\left(S_{j}^{\infinite}\cup S_{j}^{\finite}\right)\subseteq V\setminus R_{i}\text{ and }\bigcup_{j\le i}\left(S_{j}^{\infinite}\cup S_{j}^{\finite}\right)\subseteq V\setminus L_{i}.\label{eq:S nonskip subset}
\end{equation}
The first holds because, for any $v\in\bigcup_{j\ge i}S_{j}^{\infinite}\cup S_{j}^{\finite}$,
then either $v_{in}$ or $v_{out}$ is in $V'_{\ge i}$. But, if $v\in R_{i}$,
then $v_{in},v_{out}\in V'_{<i}$ which is a contradiction. For second
inclusion, $L_{i}$ is disjoint from $S_{i}$ by definition, and for
any $v\in\bigcup_{j\le i-1}S_{j}^{\infinite}\cup S_{j}^{\finite}$,
either $v_{in}$ or $v_{out}$ is in $V'_{\le i}$ but if $v\in L_{i}$,
then $v_{in},v_{out}\in V'_{>i}$ which is a contradiction.

Next, observe that, for each vertex $v\in V$ in $G$, $|\{i\mid v\in S_{i}^{\finite}\}|\le2$
and $|\{i\mid v\in S_{i}^{\infinite}\}|\le2$ because for each $i$,
both $S_{i}^{\infinite}$ and $S_{i}^{\finite}$ correspond to a single
layer in $G'$ and each vertex $v$ is associated $v_{in}$ and $v_{out}$
each of which appears in a single layer of $G'$. Combining with \Cref{eq:S nonskip subset},
we have 
\[
\sum_{j\ge i}\left(\kappa(S_{j}^{\infinite})+\kappa(S_{j}^{\finite})\right)\le4\kappa(V\setminus R_{i})\text{ and }\sum_{j\le i}\left(\kappa(S_{j}^{\infinite})+\kappa(S_{j}^{\finite})\right)\le4\kappa(V\setminus L_{i}).
\]
Now, suppose for contradiction that the lemma is not true. Assume
 $\sum_{j\ge h/2}\left(\kappa(S_{j}^{\infinite})+\kappa(S_{j}^{\finite})\right)\le\sum_{j\le h/2}\left(\kappa(S_{j}^{\infinite})+\kappa(S_{j}^{\finite})\right)$ by symmetry.
Then, for all $h>i\ge h/2$, we would have $\kappa(S_{i}^{\infinite})+\kappa(S_{i}^{\finite})>\frac{200\log\kappa(V)}{h}\cdot\kappa(V\setminus R_{i})\ge\frac{50\log\kappa(V)}{h}\cdot\sum_{h\ge j\ge i}(\kappa(S_{j}^{\infinite})+\kappa(S_{j}^{\finite}))$.
This means that $\kappa(S_{h/2+1}^{\infinite})+\kappa(S_{h/2+1}^{\finite})\ge(1+\frac{50\log\kappa(V)}{h})^{h/2}>4\kappa(V)$
which is a contradiction. The argument is symmetric if $\sum_{j\ge h/2}\left(\kappa(S_{j}^{\infinite})+\kappa(S_{j}^{\finite})\right)>\sum_{j\le h/2}\left(\kappa(S_{j}^{\infinite})+\kappa(S_{j}^{\finite})\right)$.
\end{proof}
Next, $\kappa(S_{i}^{\skp})$ can be bounded by $\kappa(S_{i}^{\infinite})+\kappa(S_{i}^{\finite})$
plus some other terms. The proof will be similar to the one in \Cref{prop:cap of skip}
but more complicated.
\begin{lem}
\label{prop:cap of skip vertex}$\kappa(S_{i}^{\skp})\le\Delta(L_{i}\cup S_{i})-T(L_{i}\cup F_{i})-\ex_{f}(V)+\kappa(S_{i}^{\infinite})+2\kappa(S_{i}^{\finite})$. 
\end{lem}

\begin{proof}
From the definition of $S_{i}^{\skp}$, we have $\kappa(S_{i}^{\skp})\le c(E(V_{>i}^{in},V_{<i}^{out}))$.
By \Cref{prop: preflow and label}(\ref{enu:cross for}), we have,
for each $(u',w')\in E(V_{>i}^{in},V_{<i}^{out})$, $c(u',w')=f(u',w')$.
So $c(E(V_{>i}^{in},V_{<i}^{out}))$ is at most the total amount of
mass going out of $V'_{>i}$. Observe that the total mass in-coming
into $V'_{>i}$ is $f(V'_{i},V'_{i+1})$ because any other edge $(u,v)\in E(V'_{\le i},V'_{>i})\setminus E(V'_{i},V'_{i+1})$
is ``skipping levels'' and so $f(u,v)=0$ by 
 \Cref{prop: preflow and label}(\ref{enu:cross back}). Therefore, the total mass going
out of $V'_{>i}$ is at most 
\begin{align*}
\underset{\textrm{initial mass}}{\underbrace{\Delta'(V'_{>i})}}+\underset{\textrm{incoming mass}}{\underbrace{f'(V'_{i},V'_{i+1})}}-\underset{\textrm{absorbed mass}}{\underbrace{\ab_{f'}(V'_{>i})}}-\underset{\textrm{excess}}{\underbrace{\ex_{f'}(V'_{>i})}}
\end{align*}
We upper bound each term one by one. First, we prove that $\Delta'(V'_{>i})\le\Delta(L_{i}\cup S_{i})$.
To see this, as $\Delta'(V^{out}\cup\Vinf)=0$ by definition, we have
$\Delta'(V'_{>i})=\Delta'(V{}_{>i}^{in})$. For each $v_{in}\in V_{>i}^{in}$,
we have $v\in L_{i}\cup S_{i}$, otherwise $v\in\Rbar_{i}$ and so
$v_{in}\in V_{<i}^{in}$ which is a contradiction. So $\Delta'(V{}_{>i}^{in})\le\Delta(L_{i}\cup S_{i})$.%

{} Next, we prove that $T(L_{i})=\ab_{f}(L_{i})\le\ab_{f'}(V'_{>i})$.
To see this, for any $v\in L_{i}\cup F_{i}$, by \Cref{prop:preflow vertex}
and \Cref{prop: preflow and label}(\ref{enu:full absorb}) $\ab_{f}(v)=\ab_{f'}(v_{out})=T'(v_{out})=T(v)$
where $v_{out}\in V'_{>i}$. So $T(L_{i}\cup F_{i})=\ab_{f}(L_{i}\cup F_{i})\le\ab_{f'}(V'_{>i})$
as desired. Also, we have $\ex_{f'}(V'_{>i})=\ex_{f'}(V'_{h})=\ex_{f'}(V')=\ex_{f}(V)$
by \Cref{prop:preflow vertex}. Lastly, we claim that $f(V'_{i},V'_{i+1})\le2\kappa(S_{i}^{\finite})+\kappa(S_{i}^{\infinite})$.
To see this, observe that $f(V'_{i},V'_{i+1})=f(V_{i}^{in}\cup V_{i}^{out},V'_{i+1})+f(\Vinf_{i},V'_{i+1})$.
Now, we will show that $f(V_{i}^{in}\cup V_{i}^{out},V'_{i+1})\le2\kappa(S_{i}^{\finite})$
and $f(\Vinf_{i},V'_{i+1})\le\kappa(S_{i}^{\infinite})$. 

To see the first bound, we trivially have $f(V_{i}^{in}\cup V_{i}^{out},V'_{i+1})\le f(V_{i}^{in}\cup V_{i}^{out},V')$
which is the total amount of out-going mass from $V_{i}^{in}\cup V_{i}^{out}$.
For any vertex $v'\in V_{i}^{in}\cup V_{i}^{out}$, by \Cref{prop:bound mass via vertex cap},
the total mass that may flow out of $v'$ is at most $\kappa(v)$
where $v\in V$ is the vertex corresponding to $v'$. So $f(V_{i}^{in}\cup V_{i}^{out},V')\le\sum_{v'\in V_{i}^{in}\cup V_{i}^{out}}\kappa(v)\le2\kappa(S_{i}^{\finite})$.
Therefore, $f(V_{i}^{in}\cup V_{i}^{out},V'_{i+1})\le2\kappa(S_{i}^{\finite})$
as desired.

To see the second bound, we have $f(\Vinf_{i},V'_{i+1})=f(\Vinf_{i},V_{i+1}^{in})$
because edges from $\Vinf$ only go to $V^{in}$. Note that the mass
that goes directly from $\Vinf_{i}$ to $V_{i+1}^{in}$ must go through
some edge $(x,v_{in})\in E_{G'}(\Vinf_{i},V_{i+1}^{in})$ and, by
\Cref{prop:bound mass via vertex cap}, this vertex $v_{in}\in V_{i+1}^{in}$
can receive total mass at most $\kappa(v)$ where $v\in V$ is the
vertex corresponding to $v_{in}$. Crucially, observe that these vertices
$v$ must be in $S_{i}^{\infinite}$. Thus, $f(\Vinf_{i},V'_{i+1})=f(\Vinf_{i},V_{i+1}^{in})\le\kappa(S_{i}^{\infinite})$
as desired.
\end{proof}

\paragraph{Proof of \Cref{thm:main vertex flow}.}

After we have described how to define the preflow and vertex cuts,
we are ready now to prove \Cref{thm:main vertex flow}. Given a flow
problem $\Pi=(\Delta,T,\kappa)$ on a vertex-capacitated graph $G=(V\cup\Vinf,E)$,
we call \Cref{prop: preflow and label} with parameter $h$ and obtain
in $O(mh\log m)$ time a feasible preflow $f'$ for $\Pi'$ and a
vertex labeling $l:V'\rightarrow\{0,\dots,h\}$. Then, by \Cref{prop:preflow vertex},
we can construct a feasible preflow $f$ for $\Pi$ in $O(m)$ time. 

If $\ex_{f}(V)\le z$, we return the preflow $f$. Suppose that $\Pi$
is $1/d$-integral. Similar to \Cref{rem:excess at source and path decomposition},
given the preflow $f$ on $G$, we can obtain a feasible flow $f''$
on $G$ with total value at least $\Delta(V)-\ex_{f}(V)$ together
with its path decomposition $\cP_{f''}$ in $O(d\Delta(V)h)$ time
using \Cref{lem:path decomposition}. We have the the total value of
paths whose length greater than $h/\eps_{\len}$ is at most $\eps_{\len}(\Delta(V)-\ex_{f}(V))$
otherwise $\sum_{P\in\pset_{f''}}\val(P)|P|>(\Delta(V)-\ex_{f}(V))h$
which contradicts \Cref{lem:path decomposition}. Therefore, the total
flow value of paths whose length at most $h/\eps_{\len}$ is at least
$(1-\eps_{\len})\cdot(\Delta(V)-\ex_{f}(V))$. We will return the
flow corresponding to these short paths as an output. 

Now, suppose that $\ex_{f}(V)>z$. By \Cref{prop:cap of nonskip vertex},
there is an index $i$ where $h>i>0$ such that $\kappa(S_{i}^{\infinite})+\kappa(S_{i}^{\finite})\le\frac{200\log\kappa(V)}{h}\cdot\min\left\{ \kappa(V\setminus R_{i}),\kappa(V\setminus L_{i})\right\} $.
Fix such $i$. Consider the partition $(L_{i},S_{i},R_{i},F_{i})$
of $V$ defined above \Cref{prop:define get partition }. 
\begin{claim}
There is a partition $(F_{i}^{L},F_{i}^{R})$ of $F_{i}$ such that
$$\min\left\{ \kappa(V\setminus R_{i}),\kappa(V\setminus L_{i})\right\} =O(\min\left\{ \kappa(L\cup S),\kappa(S\cup R)\right\} )$$
when we define $(L,S,R)=(L_{i}\cup F_{i}^{L},S_{i},R_{i}\cup F_{i}^{R})$.
\end{claim}

\begin{proof}
Recall that $\min\left\{ \kappa(V\setminus R_{i}),\kappa(V\setminus L_{i})\right\} =\min\left\{ \kappa(L_{i}\cup S_{i}\cup F_{i}),\kappa(R_{i}\cup S_{i}\cup F_{i})\right\} $.
There are three cases. First, suppose that $\kappa(L_{i}\cup S_{i})\ge\kappa(F_{i})/10$.
Then, we define $L=L_{i}$ and $R=R_{i}\cup F_{i}$. Thus, $\kappa(L_{i}\cup S_{i}\cup F_{i})\le11\kappa(L_{i}\cup S_{i})=11\kappa(L\cup S)$
and $\kappa(R_{i}\cup S_{i}\cup F_{i})=\kappa(R\cup S)$. So the claim
holds. Second, suppose that $\kappa(R_{i}\cup S_{i})\ge\kappa(F_{i})/10$.
Then, we define $L=L_{i}\cup F_{i}$ and $R=R_{i}$. Thus, $\kappa(L_{i}\cup S_{i}\cup F_{i})=\kappa(L\cup S)$
and $\kappa(R_{i}\cup S_{i}\cup F_{i})\le11\kappa(R_{i}\cup S_{i})=11\kappa(R\cup S)$.
So the claim also holds.

Lastly, suppose that $\kappa(L_{i}\cup S_{i}),\kappa(R_{i}\cup S_{i})<\kappa(F_{i})/10$.
So $\kappa(F_{i})\ge3\kappa(V)/4$, otherwise we have $\kappa(V)\le\kappa(L_{i}\cup S_{i})+\kappa(R_{i}\cup S_{i})+\kappa(F_{i})<\kappa(F_{i})(1+\frac{2}{10})<\kappa(V)\cdot\frac{3}{4}\cdot(1+\frac{2}{10})$
which is a contradiction. Since $\kappa(v)\le\kappa(V)/2$ for all
$v\in V$, there is a partition $(F_{i}^{L},F_{i}^{R})$ of $F_{i}$
such that both $\kappa(F_{i}^{L}),\kappa(F_{i}^{R})\ge\kappa(V)/3\ge\kappa(F_{i})/4$.
By setting $L=L_{i}\cup F_{i}^{L}$ and $R=R_{i}\cup F_{i}^{R}$,
we conclude 
\begin{align*}
	\min\left\{ \kappa(L_{i}\cup S_{i}\cup F_{i}),\kappa(R_{i}\cup S_{i}\cup F_{i})\right\}  & \le\frac{11}{10}\cdot\kappa(F_{i})\\
	& \le\frac{44}{10}\cdot\min\left\{ \kappa(F_{i}^{L}),\kappa(F_{i}^{R})\right\} \\
	& =O(\min\left\{ \kappa(L\cup S),\kappa(S\cup R)\right\} ).
\end{align*}
\end{proof}
Let $(L,S,R)=(L_{i}\cup F_{i}^{L},S_{i},R_{i}\cup F_{i}^{R})$ as
in the above claim, we have that 

\[
\kappa(S_{i}^{\infinite})+\kappa(S_{i}^{\finite})=O(\frac{\log\kappa(V)}{h})\cdot\min\left\{ \kappa(L\cup S),\kappa(S\cup R)\right\} .
\]
By \Cref{prop:cap of skip vertex}, we also have $\kappa(S_{i}^{\skp})\le\Delta(L_{i}\cup S_{i})-T(L_{i}\cup F_{i})-\ex_{f}(V)+\kappa(S_{i}^{\infinite})+2\kappa(S_{i}^{\finite})$.
We observe that 

\begin{align*}
\Delta(L_{i}\cup S_{i})-T(L_{i}\cup F_{i}) & \le\min\{\Delta(L_{i}\cup S_{i}),T(R_{i}\cup S_{i})\}\\
 & \le\min\{\Delta(L_{i}),T(R_{i})\}+\kappa(S_{i})/2
\end{align*}
The first inequality is because $\Delta(L_{i}\cup S_{i})-T(L_{i}\cup F_{i})\le T(R_{i}\cup S_{i})-\Delta(R_{i}\cup F_{i})\le T(R_{i}\cup S_{i})$
as $\Delta(V)\le T(V)$. The second inequality is because $\Delta(v),T(v)\le\kappa(v)/2$
for all $v\in V$. As $\kappa(S)\le\kappa(S_{i}^{\skp})+\kappa(S_{i}^{\infinite})+\kappa(S_{i}^{\finite})$
and $L_{i}\subseteq L$ and $R_{i}\subseteq R$, we have
\begin{align*}
\kappa(S) & \le\min\{\Delta(L_{i}\cup S_{i}),T(R_{i}\cup S_{i})\}-z+O(\frac{\log\kappa(V)}{h})\cdot\min\left\{ \kappa(L\cup S),\kappa(S\cup R)\right\} \\
 & \le\min\left\{ \Delta(L),T(R)\right\} +\kappa(S)/2-z+O(\frac{\log\kappa(V)}{h})\cdot\min\left\{ \kappa(L\cup S),\kappa(S\cup R)\right\} 
\end{align*}
which implies that 
\[
\kappa(S)\le2(\min\left\{ \Delta(L),T(R)\right\} -z)+O(\frac{\log\kappa(V)}{h})\cdot\min\left\{ \kappa(L\cup S),\kappa(S\cup R)\right\} 
\]
as desired.

Next, observe that if $\ex_{f}(v)>0$, then by \Cref{prop:preflow vertex}
$\ex_{f'}(v_{in})>0$ which means that $v_{in}\in V'_{h}$ by \Cref{prop: preflow and label}.
So the vertex $v\notin\Rbar_{i}$ and so $v\in L\cup S$. As $\ex_{f}(v)\le\Delta(v)$
for any $v$. We conclude $\ex_{f}(V)\le\Delta(L\cup S)$.%

{} Finally, if $T(v)>\ab_{f}(v)$, then by \Cref{prop:preflow vertex}
$T'(v_{out})=T(v)>\ab_{f}(v)=\ab_{f'}(v_{out})$ which means that
$v_{out}\in V'_{0}$ by \Cref{prop: preflow and label}. So the vertex
$v\notin\Lbar_{i}$ and so $v\in R\cup S$. This means that $T(V)-\ab_{f}(V)\le T(R\cup S)-\ab_{f}(R\cup S)$.
As $\ex_{f}(V)=\Delta(V)-\ab_{f}(V)\le T(V)-\ab_{f}(V)$ by the assumption
that $T(V)\ge\Delta(V)$ and by the definition of $\ex_{f}(V)$, we
conclude that $\ex_{f}(V)\le T(R\cup S)$.%

{} Therefore, $\min\left\{ \Delta(L\cup S),T(R\cup S)\right\} >z$ and
so the partition $(L,S,R)$ satisfies all the requirement from \Cref{thm:main vertex flow}.

\subsection{Flow Subroutines on Vertex-capacitated Graphs and Hypergraphs}

\label{sec:flow subroutine vertex}

In this section, we show how \Cref{thm:main vertex flow} can be applied
to give several useful flow algorithms for vertex-capacitated graphs
and hypergraphs. 

The lemma below (which is a restatement of \Cref{lem:vertex-matching}) either returns a sparse and balanced vertex cut
or an embedding that embeds a near perfect matching with small vertex
congestion. This subroutine can be used for checking if a graph contain
a sparse balanced vertex cut, via the cut-matching game.
\begin{lem}
[Unit-capacity Graphs]\label{lem:vertex matching unit cap}There
is an algorithm $\vertexmatching(G,A,B,\phi,\eps)$ that, given a
directed graph $G=(V,E)$ with $n$ vertices and $m$ edges, two disjoint
sets $A,B\subset V$ where $n/4\le|A|\le|B|$, $\phi\in(1/n,o(1))$,
and $\eps\in(0,1)$, in $\Otil(m/\phi)$ time, either returns 
\begin{itemize}
\item \textbf{(Sparse Cut):} an out-vertex-cut $(L,S,R)$ such that $\min\{|L|,|R|\}\ge\eps n/10$
and $|S|\le6\phi\cdot\min\{|L|,|R|\}$; or
\item \textbf{(Matching):} an embedding $\cP$ that embeds an integral directed
matching $M$ from $A$ to $B$ of size at least $(1-\eps)|A|$ into
$G$ with vertex congestion at most $\left\lfloor 1/\phi\right\rfloor $
where the length of $\cP$ is at most $\len(\cP)\le O(\log(n)/(\eps\phi))$.
Moreover, each path in $\cP$ is a simple path. 
\end{itemize}
\end{lem}

\begin{proof}
Let $G_{bip}=(V\cup\Vinf,E_{bip})$ be a bipartite representation
of $G$. Let $\Pi=(\Delta,T,\kappa)$ be a vertex-capacitated flow
problem on $G_{bip}$ defined as follows. We set $\Delta(v)=1$ for
all $v\in A$ and otherwise $\Delta(v)=0$. We set $T(v)=1$ for all
$v\in B$ and otherwise $T(v)=0$. Lastly, we set $\kappa(v)=\left\lfloor 1/\phi\right\rfloor $
for all $v\in V$. Note that $\Pi$ is integral. Also $\Delta(v),T(v)\le\kappa(v)/2$
for all $v\in V$ , $\kappa(v)\le\kappa(V)/2$ for all $\in V$, and
$\Delta(V)\le T(V)$. So we can call \Cref{thm:main vertex flow} with
parameter $h=O(\log(\kappa(V))/\phi)$ and $z=\frac{\eps}{2}|A|$
using $\Otil(m/\phi)$ time. 

If a preflow $f$ is returned where $\ex_{f}(V)\le z$, by setting
$\eps_{\len}=\frac{\eps}{2}$, we can spend additional $O(\Delta(V)h)=\Otil(n/\phi)$
time to obtain a feasible flow $f_{bip}$ in $G_{bip}$ of value at
least $(1-\eps/2)\cdot(\Delta(V)-\ex_{f}(V))\ge(1-\eps)|A|$ and its
path decomposition where each path is a simple path of length at most
$O(h/\eps_{\len})=O(\log(n)/(\eps\phi))$. By the correspondence between
$G_{bip}$ and $G$, $f_{bip}$ and its path decomposition corresponds
exactly to desired embedding $\cP$. 

Next, suppose that a partition $(L,S,R)$ of $V$ is returned. By
\Cref{thm:main vertex flow}, we have 
\begin{align*}
|S|\cdot\left\lfloor 1/\phi\right\rfloor =\kappa(S) & \le2(\min\{\Delta(L),T(R)\}-z)+O(\frac{\log\kappa(V)}{h})\cdot\min\left\{ \kappa(L\cup S),\kappa(S\cup R)\right\} \\
 & \le2\min\{|L\cap A|,|R\cap B\}+\min\left\{ |L\cup S|,|S\cup R|\right\} .
\end{align*}
This implies that $|S|\le6\phi\cdot\min\{|L|,|R|\}$. Since $\Delta(v),T(v)\le1$
for all $v$, we have 
\[
\min\{|L\cup S|,|R\cup S|\}\ge\min\{\Delta(L\cup S),T(R\cup S)\}\ge z\ge\eps n/8.
\]
As $|S|\le6\phi\cdot\min\{|L|,|R|\}=o(\min\{|L|,|R|\})$, we have
$|L|,|R|\ge\eps n/10$. Since there is no path from a vertex in $L$
to another vertex in $R$ in $G_{bip}-S$, there is no such path in
$G-S$ too and so $(L,S,R)$ is indeed a vertex cut in $G$. 
\end{proof}
The algorithm below is similar to the algorithm from \Cref{lem:vertex matching unit cap}
above. However, now the vertex capacity function $\kappa$ and a terminal
set $A\cup B$ are given, and the returned vertex cut $(L,S,R)$ should
be sparse in the following sense: $\kappa(S)=O(\min\{|L\cap A|,|R\cap B|\})$.
That is, the total capacity of separator $S$ is small compared to
the \emph{number of terminals }on each side of the cuts. This subroutine
will not be used in this paper, but we will use it in our subsequent
work \cite{BernsteinGS21detsssp} as it turns out that the above notion
of sparsity is crucial for that work.
\begin{lem}
[Capacitated Hypergraphs with respect to Terminals]\label{lem:vertex matching cap terminal}There
is an algorithm\\ $\EmbedMatching (H,A,B,\kappa,\eps)$ that is given
a hypergraph graph $H=(V,E)$, two disjoint sets of terminals $A,B\subseteq V$
where $|A|\le|B|$, a vertex capacity function $\kappa:V\rightarrow\frac{1}{d}\mathbb{Z}_{\ge0}$
such that $\kappa(v)\ge2$ for all terminals $v\in A\cup B$ and $\kappa(v)\le\kappa(V)/2$
for all vertices $v\in V$, and a balanced parameter $\eps>0$, then
in $\Otil(|H|\frac{\kappa(V)}{\eps|A|}+d\kappa(V)/\eps)$ time where
$|H|=\sum_{e\in E}|e|$ either returns 
\begin{itemize}
\item \textbf{(Sparse Cut):} a vertex cut $(L,S,R)$ in $H$ such that $\min\{|L\cap A|,|R\cap B|\}\ge\eps|A|$
and $\kappa(S)\le2\min\{|L\cap A|,|R\cap B|\}$; or 
\item \textbf{(Matching):} an embedding $\cP$ that embeds a $\frac{1}{d}$-integral
directed matching $M$ from $A$ to $B$ of total value at least $(1-3\eps)|A|$
into $H$ where the congestion of $\cP$ w.r.t. $\kappa$ is at most
$1$ and the length of $\cP$ is at most $\len(\cP)\le O(\kappa(V)\log(\kappa(V))/(|A|\eps^{2}))$.
Moreover, each path in $\cP$ is a simple path. 
\end{itemize}
\end{lem}

\begin{proof}
Let $G_{bip}=(V\cup\Vinf,E_{bip})$ be a bipartite representation
of $H$ where $m\triangleq|E_{bip}|=\Theta(|H|)$ Let $\Pi=(\Delta,T,\kappa)$
be a vertex-capacitated flow problem on $G_{bip}$ defined as follows.
We set $\Delta(v)=1$ for all $v\in A$ and otherwise $\Delta(v)=0$.
We set $T(v)=1$ for all $v\in B$ and otherwise $T(v)=0$. The vertex
capacities $\kappa$ is given to us and is $1/d$-integral, and so
$\Pi$ is $1/d$-integral. Also, we can check that $\Delta(v),T(v)\le\kappa(v)/2$
for all $v\in V$ , $\kappa(v)\le\kappa(V)/2$ for all $\in V$, and
$\Delta(V)\le T(V)$. So we can call \Cref{thm:main vertex flow} with
parameters $h=O(\frac{\kappa(V)}{\eps|A|}\log\kappa(V))$ and $z=2\eps|A|$
using $\Otil(|H|\frac{\kappa(V)}{\eps|A|})$ time. 

If a preflow $f$ is returned, by setting $\eps_{\len}=\eps$, we
can spend additional $O(d\Delta(V)h)=\Otil(d|A|\cdot\frac{\kappa(V)}{\eps|A|})=\Otil(d\kappa(V)/\eps)$
time to obtain a feasible flow $f_{bip}$ in $G_{bip}$ of value at
least $(1-\eps)\cdot(\Delta(V)-\ex_{f}(V))\ge(1-3\eps)|A|$ and its
path decomposition where each path is a simple path of length at most
$O(h/\eps_{\len})=O(\kappa(V)\log(\kappa(V))/(|A|\eps^{2}))$. By
the correspondence between $G_{bip}$ and $H$, $f_{bip}$ and its
path decomposition corresponds exactly to desired embedding $\cP$. 

Next, suppose that a partition $(L,S,R)$ of $V$ is returned. By
\Cref{thm:main vertex flow}, we have 
\begin{align*}
\kappa(S) & \le2(\min\{\Delta(L),T(R)\}-z)+O(\frac{\log\kappa(V)}{h})\cdot\min\left\{ \kappa(L\cup S),\kappa(S\cup R)\right\} \\
 & \le2\min\{|L\cap A|,|R\cap B|\}-2\eps|A|+\eps|A|\\
 & \le2\min\{|L\cap A|,|R\cap B|\}
\end{align*}
where the last inequality is by choosing the constant in the definition
of $h$ to be large enough. Also, \Cref{thm:main vertex flow} guarantees
that $z\le\min\{\Delta(L),T(R)\}+\Delta(S)\le\min\{\Delta(L),T(R)\}+\kappa(S)/2.$
Since $\Delta(v),T(v)\le1$ for all $v$, we have 
\[
2\eps|A|=z\le\min\{\Delta(L),T(R)\}+\kappa(S)/2\le2\min\{|L\cap A|,|R\cap B|\}.
\]
So $\min\{|L\cap A|,|R\cap B|\}\ge\eps|A|$ as desired. Since there
is no path from a vertex in $L$ to another vertex in $R$ in $G_{bip}-S$,
there is no such path in $G-S$ too and so $(L,S,R)$ is indeed a
vertex cut in $G$. 
\end{proof}
The algorithm below is also similar to the algorithms from \Cref{lem:vertex matching unit cap}
and \Cref{lem:vertex matching cap terminal} above. However, now the
vertex capacity function $\kappa$ and an expansion parameter $\phi$are
given, and the returned vertex cut $(L,S,R)$ should be sparse in
the following sense: $\kappa(S)=O(\phi\min\{\kappa(L),\kappa(R)\})$.
This notion of sparsity is natural and therefore we believe that this
subroutine will be useful for future applications. 
\begin{lem}
[Capacitated Hypergraphs]There is an algorithm $\EmbedMatching(H,A,B,\kappa,\phi,\eps)$
that is given a hypergraph graph $H=(V,E)$, two disjoint set $A,B\subseteq V$
where $\kappa(A)\le\kappa(B)$, a vertex capacity function $\kappa:V\rightarrow\frac{1}{d}\mathbb{Z}_{\ge0}$,
an expansion parameter $\phi\in(1/n,o(1))$, and a balanced parameter
$\eps>0$, then in $\Otil(|H|/\phi+d\kappa(A)/\phi)$ time where $|H|=\sum_{e\in E}|e|$
either returns 
\begin{itemize}
\item \textbf{(Sparse Cut):} a vertex cut $(L,S,R)$ in $H$ such that $\min\{\kappa(L),\kappa(R)\}\ge\eps\kappa(A)$
and $\kappa(S)\le6\phi\cdot\min\{\kappa(L),\kappa(R)\}$; or 
\item \textbf{(Matching):} an embedding $\cP$ that embeds a $\frac{1}{d}$-integral
directed matching $M$ from $A$ to $B$ of total weight at least
$(1-3\eps)\kappa(A)$ into $H$ with vertex congestion at most $\left\lfloor 1/\phi\right\rfloor $
where the length of $\cP$ is at most $\len(\cP)\le O(\log(\kappa(V))/(\eps\phi))$.
Moreover, each path in $\cP$ is a simple path. 
\end{itemize}
\end{lem}

\begin{proof}
Let $G_{bip}=(V\cup\Vinf,E_{bip})$ be a bipartite representation
of $H$ where $m\triangleq|E_{bip}|=\Theta(|H|)$ Let $\Pi=(\Delta,T,\kappa')$
be a vertex-capacitated flow problem on $G_{bip}$ defined as follows.
We set $\Delta(v)=\kappa(v)$ for all $v\in A$ and otherwise $\Delta(v)=0$.
We set $T(v)=\kappa(v)$ for all $v\in B$ and otherwise $T(v)=0$.
The vertex capacities $\kappa'(v)=\left\lfloor 1/\phi\right\rfloor \cdot\kappa(v)$
for each $v\in V$. As $\kappa$ is is $1/d$-integral, and so $\Pi$
is $1/d$-integral. Also, we can check that $\Delta(v),T(v)\le\kappa(v)/2$
for all $v\in V$ , $\kappa(v)\le\kappa(V)/2$ for all $\in V$, and
$\Delta(V)\le T(V)$. So we can call \Cref{thm:main vertex flow} with
parameters $h=O(\log\kappa'(V)/\phi)$ and $z=2\eps\kappa(A)$ using
$\Otil(|H|/\phi)$ time. 

If a preflow $f$ is returned, by setting $\eps_{\len}=\eps$, we
can spend additional $O(d\Delta(V)h)=\Otil(d\kappa(A)/\phi)$ time
to obtain a feasible flow $f_{bip}$ in $G_{bip}$ of value at least
$(1-\eps)\cdot(\Delta(V)-\ex_{f}(V))\ge(1-3\eps)\kappa(A)$ and its
path decomposition where each path is a simple path of length at most
$O(h/\eps_{\len})=O(\log(\kappa(V))/(\eps\phi))$. By the correspondence
between $G_{bip}$ and $H$, $f_{bip}$ and its path decomposition
corresponds exactly to desired embedding $\cP$. 

Next, suppose that a partition $(L,S,R)$ of $V$ is returned. By
\Cref{thm:main vertex flow}, we have 
\begin{align*}
\kappa(S)\cdot\left\lfloor 1/\phi\right\rfloor =\kappa(S') & \le2(\min\{\Delta(L),T(R)\}-z)+O(\frac{\log\kappa(V)}{h})\cdot\min\left\{ \kappa(L\cup S),\kappa(S\cup R)\right\} \\
 & \le2\min\{\kappa(L),\kappa(R)\}+\min\left\{ \kappa(L\cup S),\kappa(S\cup R)\right\} 
\end{align*}
which implies that $\kappa(S)\le6\phi\cdot\min\{\kappa(L),\kappa(R)\}$.
Also, as $\Delta(v),T(v)\le\kappa(v)$ for all $v\in V$, we have
\[
2\eps\kappa(A)=z\le\min\{\Delta(L\cup S),T(R\cup S)\}\le\min\{\kappa(L),\kappa(R)\}+\kappa(S).
\]
Since $\kappa(S)\le o(\min\{\kappa(L),\kappa(R)\})$, we have $\min\{\kappa(L),\kappa(R)\}\ge\eps\kappa(A)$.
Since there is no path from a vertex in $L$ to another vertex in
$R$ in $G_{bip}-S$, there is no such path in $G-S$ too and so $(L,S,R)$
is indeed a vertex cut in $G$. 
\end{proof}
Let $H=(V,E)$ be a hypergraph with $n$ vertices and $m$ hyperedges
where $|H|=\sum_{e\in E}|e|$. Below, we show an algorithm for computing
a maximum number of $s$-$t$ vertex-disjoint paths in hypergraphs
with $\tilde{O}(|H|\sqrt{n})$ running time, which generalizes the
algorithm by Even and Tarjan \cite{EvenT75} since 1975 for ordinary
graphs with running time $O(m\sqrt{n})$. 

To the best of our knowledge, the fastest known algorithms for this
problem take $\Otil(|H|\sqrt{m+n})$ \cite{LeeS14}, $O(|H|^{4/3+o(1)})$
\cite{LiuS20}, or $O(|H|n)$ \cite{AhujaOST94}. Therefore, our algorithm
below is fastest when $m$ are $|H|$ are large. We note that the
logarithmic factor in our running time could be removed but we do
not try to optimize it.
\begin{lem}
[Vertex-disjoint Paths in Hypergraphs]There is an algorithm that
is given a directed hypergraph $n$-vertex graph $H=(V,E)$ with unit
vertex capacity, two vertices $s,t\in V$, and then in $\Otil(|H|\sqrt{n})$
time where $|H|=\sum_{e\in E}|e|$ returns a maximum number of $s$-$t$
vertex-disjoint paths.
\end{lem}

\begin{proof}
Let $G_{bip}=(V\cup\Vinf,E_{bip})$ be a bipartite representation
of $H$ where $m\triangleq|E_{bip}|=\Theta(|H|)$ Let $\Pi=(\Delta,T,\kappa)$
be a vertex-capacitated flow problem on $G_{bip}$ defined as follows.
We set $\Delta(s)=n$ and $\Delta(v)=0$ for $v\neq s$. We set $T(t)=n$
and otherwise $T(v)=0$ for $v\neq t$. The vertex capacities $\kappa(v)=1$
for each $v\in V\setminus\{s,t\}$ and $\kappa(v)=2n$ for $v\in\{s,t\}$.
Note that $\Pi$ is integral. Also, we can check that $\Delta(v),T(v)\le\kappa(v)/2$
for all $v\in V$ , $\kappa(v)\le\kappa(V)/2$ for all $\in V$, and
$\Delta(V)\le T(V)$. So we can call \Cref{thm:main vertex flow} with
parameters $h=O(\sqrt{n}\log n)$ using $\Otil(|H|\sqrt{n})$ time.
Let $z$ be such that the algorithm from \Cref{thm:main vertex flow}
with excess parameter $z$ returns a preflow $f$ with $\ex_{f}(V)\le z$,
but with excess parameter $z-1$, it returns a partition $(L,S,R)$
of $V$. 

If a preflow $f$ is returned, we can spend additional $O(\Delta(V)h)=\Otil(n\sqrt{n})=\Otil(|H|\sqrt{n})$
time to obtain a feasible flow $f_{bip}$ in $G_{bip}$ of value at
least $\Delta(V)-\ex_{f}(V)\ge n-z$ and its path decomposition where
each path is a simple path. By the correspondence between $G_{bip}$
and $G$, $f_{bip}$ and its path decomposition corresponds to a collection
of at least $n-z$ many $s$-$t$ vertex-disjoint paths in $H$. By
\Cref{thm:main vertex flow}, we have 
\begin{align*}
\kappa(S) & \le\min\{\Delta(L\cup S),T(R\cup S)\}-z+O(\frac{\log\kappa(V)}{h})\cdot\min\left\{ \kappa(L\cup S),\kappa(S\cup R)\right\} \\
 & \le n-z+\sqrt{n}
\end{align*}
As $\kappa(s)=\kappa(t)=2n>\kappa(S)$, we have $s,t\notin S$. Also,
by \Cref{thm:main vertex flow}, $\min\left\{ \Delta(L\cup S),T(R\cup S)\right\} >z-1\ge0$,
so $s\in L$ and $t\in R$. (Note that $z>0$, otherwise there are
$n$ many $s$-$t$ vertex-disjoint paths which is impossible.) Therefore,
$(L,S,R)$ is an $s$-$t$ out-vertex-cut of size at most $n-z+\sqrt{n}$.
Since we have already found at least $n-z$ many $s$-$t$ vertex-disjoint
paths, we can invoke Ford-Fulkerson algorithm for finding augmenting
paths for at most $\sqrt{n}$ paths (there cannot be more paths because
of $(L,S,R)$) so obtain a maximum collection of $s$-$t$ vertex-disjoint
paths. This takes $O(|H|\sqrt{n})$ additional time. The total time
is thus $\Otil(|H|\sqrt{n})$.
\end{proof}

\section{Proof of Proposition \ref{thm:lacki}}
\label{sec:lacki}

We prove \Cref{thm:lacki} in this section. For convenience, we restate the proposition below. 

\LackiProp*
\begin{proof}
In order to prove the proposition, let us first define the following notion.

\begin{defn}
For any graph $H$, where the SCCs of $H$ are the sets $C_1, C_2, \dots, C_k$, we say that the \emph{condensation} $\textsc{Cond}(H)$ of the graph $H$ is the graph of $H$ after contracting vertices in each SCC $C_i$ into a supervertex, i.e. the graph $H / \{C_1, C_2, \dots, C_k\}$.
\end{defn}

We then use the following claim that extends a condensation of the subgraph $G \setminus X$ to a condensation of $G \setminus (X \setminus \{x\})$ where $x \in X$. This is the key ingredient in our data structure. We defer the proof to the end of the section.

\begin{restatable}{claim}{clmSSRforLacki}\label{clm:lackiSSR}
There exists a data structure $\mathcal{C}$ that given a decremental graph $G$ and an increasing set $X \subseteq V$, a (dynamic) condensation of the graph $\textsc{Cond}(G \setminus X)$ and a vertex $x \in X$, can maintain the condensation $\textsc{Cond}(G \setminus (X \setminus \{x\})$ in total update time $O(m \log n)$. The data structure can return a path in the condensation $\textsc{Cond}(G \setminus X)$ from or to $x$ for every vertex $y$ in the same SCC in time linear in the number of edges. The path is strictly contained in the SCC of $x$.
\end{restatable}

Now, throughout the algorithm, we maintain the data structure $\mathcal{A}$ on $G$ to monitor the SCCs in the graph $G \setminus S$ which allows us to maintain the condensation $\textsc{Cond}(G \setminus S)$. 

We then, arbitrarily order the vertices $s_1, s_2, \dots, s_k$ in $S$, and for $i = 1, 2 \dots k$, we take the vertex $s_i$ and build a data structure as described in \Cref{clm:lackiSSR} to run on the condensation $\textsc{Cond}(G \setminus (S \setminus \{s_1, s_2, \dots, s_{i-1}\}))$ for vertex $s_i$ to maintain the condensation $\textsc{Cond}(G \setminus (S \setminus \{s_1, s_2, \dots, s_{i}\}))$. Thus, the condensation that is maintained by the data structure at the final vertex $s_k$ is the condensation of $G$ that has a supernode for every SCC with the same underlying vertex set.

To maintain this data structure, we pass edge deletions to $G$, to the data structures at the vertices in $S$ in their order which allows updates to percolate up and to enforce that the final data structure again maintains the condensation of $G$.

Whenever a vertex $y$ is added to $S$ by $\mathcal{A}$, we prepend $y$ to the vertices $s_1, s_2, \dots, s_t$, build a new data structure as described in \Cref{clm:lackiSSR} from $y$ on the condensation $G \setminus S$ and is now responsible to maintain the condensation of $G \setminus (S \setminus \{y\})$ and to communicate changes to $s_1$. It is not hard to see that $s_1$ thus runs on the condensation of the same underlying graph as before. 

The total update time is dominated by the time to maintain the condensation at each vertex $s \in S$. Since each such data structure runs in total update time $O(m \log n)$ and since we only run a single instance of $\mathcal{A}$, we derive total update time $O(T(m,n) + |S|m\log n)$, as desired.

To compute a path between any two vertices $x,y$ in the same SCC in $G$, we can locate straight-forwardly the condensation where they are first contained in the same supernode (for example by using a least-common ancestor data structure). If this condensation was derived by data structure $\mathcal{A}$, we directly query $\mathcal{A}$. Otherwise, there is some vertex $s_i \in S$ associated with the condensation and we can query its data structure. Whilst this only returns a path in $\textsc{Cond}(G \setminus (S \setminus \{s_1, s_2, \dots, s_{i-1}\}))$ by \Cref{clm:lackiSSR}, we can then check each the returned path and if two endpoints at the same supernode differ, we can recursively find a path between these endpoints. Since we find the paths strictly in the induced SCCs on lower levels, we have that no endpoint on the final path is visited more than once, thus we can return a simple path between the vertices $x,y$ in time almost-linear in the number of edges.
\end{proof}

Finally, we prove \Cref{clm:lackiSSR}.

\clmSSRforLacki*
\begin{proof}
Given $\textsc{Cond}(G \setminus X)$ of a decremental graph $G \setminus X$ for a set $X \subseteq V$, and a vertex $x \in X$. Then, for every vertex $v \in V \setminus X$ we monitor the in-degree of $v$ in the graph $H'$ initialized to $\textsc{Cond}(G \setminus X) \cup E(x, V \setminus X)$ and if the in-degree of one such vertex drops to $0$, we remove $v$ and its out-going edges from $H'$. This might cause additional vertices to have their in-degree drop to $0$. Similarly, we monitor for every vertex $v$ the out-degree in the graph $H''$ initialized to $\textsc{Cond}(G \setminus X) \cup E(V \setminus X,x)$ and remove $v$ and its out-going edges from $H''$ once a vertex has no longer any in-coming edges. If a vertex $y$ is added to $X$ throughout the algorithm, then we simply remove $y$ with all incident edges from $H', H''$ and $H'''$.

The condensation $\textsc{Cond}(G \setminus (X \setminus \{x\}))$ is then derived by contracting all vertices in the condensation that have non-zero in- and out-degree in $H'$ \emph{and} $H''$ together with vertex $x$ into a new SCC supervertex. 

To see that this correctly maintains $\textsc{Cond}(G \setminus (X \setminus \{x\}))$, observe that the graph $\textsc{Cond}(G \setminus X)$ is a DAG and therefore every SCC in the graph $H''' = \textsc{Cond}(G \setminus X) \cup E(x,V \setminus X) \cup E(V \setminus X, x)$ has to contain $x$, since every cycle has to go through $x$. Further, it is not hard to establish by induction that a vertex $v$ is only removed from $H'$ if and only if there is no path from $x$ to $v$ in $H'''$ and similarly $v$ is removed only from $H''$ iff there is no path from $v$ to $x$ in $H'''$. Thus, $v$ remains in the graphs $H'$ and $H''$ if and only if it is strongly-connected to $x$. This establishes correctness. To obtain an upper bound on the running time of $O(m \log n)$ observe that $H'''$ is a multigraph where vertices slowly decompose since $G$ is decremental and therefore the underlying condensation $\textsc{Cond}(G \setminus X)$ has an increasing supervertex set. However, every time a supervertex is split into multiple vertices, the operation can be done in time linear in the number of edges incident to the new supernodes that only contain at most half the number of vertices than the previous supernode that they were part of. Copying these edges can thus by done in $O(m \log n)$ time since an edge is copied to a new supervertex when the vertex halves in size which happens at most $O(\log n)$ times. Further after every edge deletion to $G$, we have to check the in and out-degree of the supernodes in which the endpoints are contained in and every edge might be deleted at some point. But this can be implemented straight-forwardly in at most $O(m)$ total update time which is subsumed in the total update time of $O(m \log n)$.

To return a path from $x$ to any other vertex $y$ in the same SCC as $x$ in $\textsc{Cond}(G \setminus X)$, we can maintain a dynamic tree where we add an in-edge from every vertex rooted at $x$ that is still in $H'$. It is not hard to see that this dynamic tree is indeed a spanning tree since the graph $H'$ is a DAG. Thus, the root to $y$ path is a path in  $\textsc{Cond}(G \setminus X)$ that can be extracted in time linear in the number of edges. Since every edge might be added once to the dynamic tree at some stage until it is deleted in $G$, the total number of insertions and deletions to the dynamic tree is at most $O(m)$. Since a dynamic tree can be implemented with $O(\log n)$ operations for insertions and deletions, the total running time is again subsumed by $O(m\log n)$.
\end{proof}

\section{Proof of Theorem \ref{thm:path-to-witness}}
\label{sec:forest}

For the sake of convenience, we restate the theorem proved in this section.

\PathToWitnessTheorem*
\begin{proof}
To implement the data structure $\rwitness(G,\phi)$, we use the following data structure internally.

\begin{theorem}[ES-tree, see \cite{EvenS, HenzingerK99}]
Given a directed decremental graph $G=(V,E)$, a fixed vertex $s \in V$, and a depth threshold $\delta \geq 1$. There exists a deterministic data structure that maintains explicitly the shortest path tree from $s$ in $G$ truncated at distance $\delta$ (that is the shortest path tree in the graph induced by vertices at distance at most $\delta$ from $s$), in total time $O(m\delta)$.
\end{theorem}

Instead of running it on $G$ directly, we introduce a new graph $G_s$ that is initialized to $G$ and an additional node $s$ along with an edge to and from $s$ to every vertex $w$ in $W$ (i.e. there are the anti-parallel edges $(s,w)$ and $(w,s)$ in $G_s$). Throughout the algorithm, we update $G_s$ with edge and vertex deletions (i.e. $G_s$ is a decremental graph), such that $G_s[V]$ remains at all stages a subgraph of $G$.

Throughout, we run an ES-tree $\mathcal{E}$ from $s$ on $G_s$ to depth $\lceil 1/\phi \rceil + 1$ and an ES-tree $\rev{\mathcal{E}}$ from $s$ on $\rev{G_s}$ to depth $\lceil 1/\phi \rceil + 1$. We let the corresponding truncated shortest-path trees be denoted by $\mathcal{T}$ and $\rev{\mathcal{T}}$.

Now, to update $G_s$, we pass edge deletions to $G$ directly to $G_s$ and whenever a vertex $w$ is deleted from the set $W$, we remove the edges $(s,w)$ and $(w,s)$ from $G_s$. Additionally, whenever a vertex $r \in V$, is no longer present in the tree $\mathcal{T}$ or $\rev{\mathcal{T}}$, we run a separator procedure, that prunes out a part of the graph containing $r$ using a vertex-sparse separator. A static procedure to compute such a separator is stated below.

\begin{lemma}[Balanced Separator, see Lemma 6.1 in \cite{BernsteinPW19}]
\label{lma:sepBPW}
Given a graph $G=(V,E)$, a vertex $r \in V$ and $d$ a positive integer such that the ball $B_{G_s}(r, d) = \{ v \in V \;|\; \mathbf{dist}_{G_s}(r,v) \leq d\}$ contains at most $n/2$ vertices. Then, there exists a deterministic algorithm that outputs two disjoint vertex sets $S_{Sep}, V_{Sep} \subseteq V$ with $r \in V_{Sep}$ such that
\begin{enumerate}
    \item \label{step:separator-distance} $\forall v \in V_{Sep} \cup S_{sep}$, we have $\mathbf{dist}_G(r,v) \leq d$,
    \item the cut $(V_{Sep}, S_{Sep}, V \setminus (V_{Sep} \cup S_{Sep})$ is a $\Ohat(1/d)$-vertex-sparse cut.
\end{enumerate}
The running time of the procedure is bounded by $O(|E(V_{Sep})|)$.
\end{lemma}

Given this separator procedure, whenever a vertex $r$ is removed from a tree $\mathcal{T}$ by data structure $\mathcal{E}$, we have that its distance from $s$ exceeds $\lceil 1/\phi \rceil + 1$ and since every vertex in $W$ is at distance $1$ from $s$, we have that the distance from $r$ to any vertex $w$ in $W$ is at least $\lceil 1/\phi \rceil + 1$. We then invoke the separator procedure from \Cref{lma:sepBPW} on $G_s$ from $r$ with depth parameter $d = \lceil 1/\phi \rceil$. Since this ensures that no vertex in $W$ is in $B_{G_s}(r, d) = \{ v \in V \;|\; \mathbf{dist}_{G_s}(r,v) \leq d\}$, the ball contains at most $|V(G_s) \setminus W| \leq n/2$ vertices and therefore our parameters are sound. We thus get vertex sets $S_{Sep}, V_{Sep}$ such that $(V_{Sep}, S_{Sep}, V \setminus (V_{Sep} \cup S_{Sep})$ is a $\Ohat(\phi)$-vertex-sparse cut which we output (to efficiently output, we only write $V_{Sep}$ and $S_{Sep}$) and then remove the vertices $S_{Sep} \cup V_{Sep}$ with all incident edges from $G_s$ which leaves the graph $G[V \setminus (V_{Sep} \cup S_{Sep}]$ as specified by the theorem. Since this only removes vertices not in $W$, this satisfies the requirement of the theorem regarding the subgraph that is worked upon.

Analogously, whenever a vertex $r$ is removed from a tree $\rev{\mathcal{T}}$ by data structure $\rev{\mathcal{E}}$, we find a separator using the procedure from \Cref{lma:sepBPW} on graph $\rev{G_s}$ from $r$ to depth $d = \lceil 1/\phi \rceil$. The same line of reasoning applies regarding the soundness of parameters.

Finally, let us describe how to maintain the forest $\mathcal{F}_{out}$ (the maintenance of $\mathcal{F}_{in}$ is analogous). Therefore, we observe that since the shortest path tree $\mathcal{T}$ is maintained explicitly by the ES-tree algorithms, we have at most $\Ohat(m/\phi)$ edge changes to the trees. We can thus maintain $\mathcal{F}_{out}$ to consist of the edges of the shortest-path tree $\mathcal{T}$ without the vertex $s$ and incident edges to $s$ in $\Ohat(m/\phi)$ time. Clearly, each such tree $T \in \mathcal{F}_{out}$ is rooted at a vertex $w \in W$ since $s$ only has edges to vertices $W$. Further, it is clear that $\mathcal{F}_{out}$ spans exactly the vertices in $V(G_s) \setminus \{s\} = V(G)$. Using a dynamic cut-link tree data structure to implement the trees, we can further straight-forwardly answer queries for every vertex $u \in V(G)$, on which vertex $w$ in $W$ is the root of its tree in time $O(\log n)$. This completes the proof.

\end{proof}

\section{Short-path Oracles on Expanders}

\label{sec:oracle}

In this section, we prove the following theorem. 

\OracleTheorem*

The idea from this section is completely identical to the analogous subroutine for undirected graphs
by Chuzhoy and Saranurak \cite{ChuzhoyS20_apsp}. In particular, \Cref{subsec:shortkpath} is copied from that paper with small changes to make it work in directed graphs.
As we only translate their ideas to our setting, and plug in our primitives for directed
expanders instead of using the primitives for undirected expanders, we do not claim any contribution in this part.

\subsection{Embedding A Small Witness }

First, we define a variant of the \emph{witness} from \Cref{def:witness}
using edge congestion instead of vertex congestion. As we will never
benefit from allowing the witness $W$ to be a weighted graphs as
we need in \Cref{sec:witness} and allow some vertex to have high (unweighted) degree,
we will restrict our witness in this section to be an unweighted graph
with small maximum degree. Moreover, we requite the embedding of the
witness to be short (as this is the point of this section).
\begin{defn}
	[Witness with Edge Congestion]\label{def:witness edge}We say that
	$W$ is a \emph{$\phi$-edge-witness} of $G$ if $V(W)\subseteq V(G)$,
	$W$ is a unweighted $\Omegahat(1)$-(edge)-expander with maximum
	degree $O(\log|V(W)|)$, and there is an embedding that embeds $W$
	into $G$ with edge-congestion $1/\phi$ and length $\Otil(1/\phi)$.
\end{defn}

We will show an algorithm for finding a \emph{$\Omega(\phi)$}-edge-witness
$W$ on $\phi$-expander $G$. In our application, $W$ will be ``small''
in the sense that $|V(W)|\ll|V(G)|$. To do find a witness, we again
employ a cut-matching game from \Cref{thm:CMG}. The lemma below is needed as
an algorithm for the matching player:
\begin{lem}
	[Matching Embedder on Expanders]\label{lem:terminal matching}There
	is an algorithm $\termMatching(G,A,B,\phi)$ with following inputs:
	a parameter $\phi\in(0,1)$, a directed unweighted $\phi$-expander
	$G=(V,E)$ with $n$ vertices and $m$ edges, and terminal sets $A,B\subset V$
	where $|A|=|B|$. In $\Otil(m/\phi)$ time, the algorithm returns
	a perfect (integral) matching $M$ from $A$ to $B$ and an embedding $\pset$
	that embeds $M$ into $G$ with edge-congestion $O(\log(n)/\phi)$
	and length $O(\log(n)/\phi)$.
\end{lem}

\begin{proof}
	We first define a flow problem $\Pi=(\Delta,T,c)$ on $G$ as follows.
	For all $v\in A,$ $\Delta(v)=1$, otherwise $\Delta(v)=0$. For all
	$v\in B,$ $T(v)=1$, otherwise $T(v)=0$. Let $c(e)=2/\phi$ for
	all $e\in E$. Let $C=\sum_{e\in E}c(e)$. Let $z=0$ and $h=\frac{40\log C}{\phi}$.
	Now, we call \Cref{lem:global flow} with $(G,z,h,\Pi)$ as
	input in time $O(mh\log m)=\tilde{O}(m/\phi)$. 
	
	We claim that the algorithm cannot return a cut $S$. Otherwise, there
	is $S$ where
	\begin{align*}
	c(E(S,V\setminus S)) & =\Delta(S)-T(S)-z+\min\{\vol^{c}(S),\vol^{c}(V\setminus S)\}\cdot\frac{10\log C}{h}.
	\end{align*}
	Note that $\Delta(S)-T(S)=T(V\setminus S)-\Delta(V\setminus S)$.
	As $\Delta(S)\le|S|\le\vol(S)$ and $T(V\setminus S)\le|V\setminus S|\le\vol(V\setminus S)$,
	we have $\Delta(S)-T(S)\le\min\{\vol(S),\vol(V\setminus S)\}$. Also,
	note that 
	\[
	\min\{\vol^{c}(S),\vol^{c}(V\setminus S)\}\cdot\frac{10\log C}{h}=\min\{\vol(S),\vol(V\setminus S)\}\cdot\frac{2}{\phi}\cdot\frac{10\log C}{h}=\min\{\vol(S),\vol(V\setminus S)\}/2
	\]
	by the choice of $h$. So, we have
	
	\begin{align*}
	|E(S,V\setminus S)| & =\frac{\phi}{2}\cdot c(E(S,V\setminus S))\\
	& \le\frac{\phi}{2}\cdot(\min\{\vol(S),\vol(V\setminus S)\}+\min\{\vol(S),\vol(V\setminus S)\}/2)\\
	& <\phi\min\{\vol(S),\vol(V\setminus S)\}
	\end{align*}
	which contradicts the fact that $G$ is a $\phi$-expander. 
	
	By \Cref{rem:excess at source and path decomposition}, so we obtain a feasible flow $f$ and its path decomposition
	$\pset_{f}$ in time $O(\Delta(V)h)=O(m/\phi)$. Let $\pset_{f}^{s}$
	contains all path in $\pset_{f}$ whose length is at most $2h$. As
	$\sum_{P\in\pset_{f}}\val(P)|P|\le\Delta(V)h$ from \Cref{lem:path decomposition}, $|\pset_{f}^{s}|\ge|\pset_{f}|/2=|A|/2$.
	By reading the endpoints of paths in $\pset_{f}^{s}$, we obtain an
	integral matching $\hat{M}$ from $\hat{A}\subseteq A$ to $\hat{B}\subseteq B$
	of size at least $|A|/2$ that can be embedded into $G$ with congestion
	$2/\phi$.
	
	As we want a perfect matching, we set $A\gets A\setminus\hat{A}$,
	$B\gets B\setminus\hat{B}$, $M\gets M\cup\hat{M}$. Then repeat the
	process $\log m$ time. At the end, we obtain an integral perfect
	matching $M$ from $A$ to $B$ that can be embedded into $G$ with
	$2\log(m)/\phi$ edge congestion and $2h$ length. We also obtain
	its corresponding embedding.
\end{proof}
Now, we are ready to apply the cut-matching game for finding a small
$\tilde{\Omega}(\phi)$-edge-witness in a $\phi$-expander.
\begin{lem}
	[Witness Embedder on Expanders]\label{lem:terminal witness}There
	is an algorithm $\termWitness(G,T,\phi)$ with the following parameters:
	a parameter $\phi\in(0,1)$, a directed unweighted $\phi$-expander
	$G=(V,E)$ with $n$ vertices and  $m$ edges, and terminal sets $T\subset V$.
	In $\Ohat(m/\phi)$ time, the algorithm finds a $\Omega(\phi/\log^{2}(n))$-edge-witness
	$W$ in $G$ where $V(W)=T$ and its corresponding embedding $\pset$.
	Let $\alphawit=1/n^{o(1)}$ such that $W$ is a $\alphawit$-expander
	and the running time is at most $O(m/(\alphawit\phi))$ (we will use
	this parameter in other lemmas).
\end{lem}

\begin{proof}
	We perform a cut matching game from \Cref{thm:CMG} for building an expander
	$W$ on $T$. 
	
	Starting from round $i=1$ of the game, \Cref{thm:CMG} gives us $A_{i},B_{i}\subset T$
	where $|A_{i}|=|B_{i}|\ge|T|/4$. Then, we call $\termMatching(G,A_{i},B_{i},\phi)$
	and $\termMatching(G,B_{i},A_{i},\phi)$ to obtain integral directed
	matchings $\Mto_{i}$ and $\Mback_{i}$ that matches $A_{i}$ to $B_{i}$
	and back. We set $W\gets W\cup\Mto_{i}\cup\Mback_{i}$ and proceed
	with round $i+1$. 
	
	After $O(\log|T|)$ rounds, $W$ is an unweighted $\Omegahat(1)$-expander
	with maximum degree $O(\log|T|)$. As $\Mto_{i}$ and $\Mback_{i}$
	can be embedded into $G$ with $O(\log(n)/\phi)$ edge congestion
	and length, $W$ can be embed into $G$ with $O(\log^{2}(n)/\phi)$
	edge congestion, and $O(\log(n)/\phi)$ length. Therefore, $W$ is
	a $\Omega(\phi/\log^{2}(n))$-edge-witness where $V(W)=T$. Note that,
	we also explicitly have the embedding of $W$. The total running time
	is $\Otil(m/\phi)+\Ohat(|T|)=\Ohat(m/\phi)$ by \Cref{lem:terminal matching}
	and \Cref{thm:CMG}.
\end{proof}

\subsection{A Recursive Scheme}

\label{subsec:shortkpath}

Now, we are ready to prove \Cref{thm:short-path-oracle}. Let $\alphawit=1/n^{o(1)}$
be the conductance bound from \Cref{lem:terminal witness}. For any
$L\ge1$, let $\gamma_{L}(\phi)=\phi^{3^{O(L)}}$ be the conductance
bound from \Cref{thm:pruning}. Below, we say that a vertex set $S$
is \emph{incremental} if vertices in $S$ can never leave $S$ as
time progresses. 
\begin{thm}
	\label{thm:shortkquery}For any number $q\ge1$ and $L\ge1$ where
	$L=q^{2}$, there is a deterministic algorithm that, given a $m$-edge
	$n$-vertex $\alphawit$-expander $G$ undergoing a sequence of edge
	deletions of length $\gamma_{L}(\alphawit)\vol(G)/n^{1/L}$, maintains
	an incremental \emph{vertex set} $P$ using $O(m^{1+2/q}/\gamma_{L}^{O(q)}(\alphawit))$
	total update time such that 
	\begin{itemize}
		\item $\vol_{G^{(0)}}(P)=O(tn^{1/L}/\gamma_{L}(\alphawit))$ after the $t$-th
		deletion where $G^{(0)}$ denotes $G$ before any deletion, and 
		\item given $u,v\in V(G)-P$, returns a $u$-$v$ simple path $Q$ in $G[V(G^{(0)})-P]$
		of length $1/\gamma_{L}^{O(q)}(\alphawit)$ in time $1/\gamma_{L}^{O(q)}(\alphawit)$. 
	\end{itemize}
\end{thm}

\paragraph{Proof of \Cref{thm:short-path-oracle} from \Cref{thm:shortkquery}.}

Let $q=\sqrt{\frac{1}{100c}\log_{3}\log_{1/\alphawit}(n)}=\omega(1)$.
Observe that $1/\gamma_{L}^{O(q)}(\alphawit)=n^{o(1)}$. This is because
$1/\gamma_{L}^{O(q)}(\alphawit)=1/(\alphawit^{3^{O(L)}})^{O(q)}=(1/\alphawit)^{3^{cq^{2}}}$
for some constant $c$. So we have $3^{cq^{2}}=3^{\frac{1}{100}\log_{3}\log_{1/\alphawit}(n)}=(\log_{1/\alphawit}(n))^{1/100}$
and so 
\[
(1/\alphawit)^{3^{cq^{2}}}=(1/\alphawit)^{(\log_{1/\alphawit}(n))^{1/100}}=n^{\log_{n}(1/\alphawit)\cdot(\log_{1/\alphawit}(n))^{1/100}}=n^{1/(\log_{1/\alphawit}(n))^{99/100}}=n^{o(1)}.
\]
So the total update time of \Cref{thm:shortkquery} is $\Ohat(m)$
and the query time is $n^{o(1)}$. This implies \Cref{thm:short-path-oracle}.

\paragraph{Proof of \Cref{thm:shortkquery}.}

The algorithm has $q$ levels. For each $1\le i\le q$, we describe
the implementation of $\mathtt{GrowTree}(i)$, $\mathtt{Delete}(i,e)$,
and $\mathtt{Query}(i,u,v)$ in \Cref{alg:core init}, \Cref{alg:core delete},
and \Cref{alg:core query} respectively. The algorithm is recursive.
Recall that we are given an input $\alphawit$-expander $G$ with
$m$ initial edges and $n$ vertices. We let $m$, $n$, and $\alphawit$
be global variables that do not change when we recurse. %

Now, we describe how we call each subroutine given an input $G$ and
an update sequence. We initialize $G_{q}^{(0)}=G$ and call $\mathtt{GrowTree}(q)$.
We initialize the expander pruning algorithm from \Cref{thm:pruning}
and maintain the set $P_{q}\subseteq V(G_{q}^{(0)})$. Whenever an
edge $e$ is deleted from $G_{q}^{(0)}$, we call $\mathtt{Delete}(q,e)$
and update the set $P_{q}$ using \Cref{thm:pruning}. Recall that
$P_{q}$ only grows. Let $G_{q}^{(d)}$ denote $G_{q}^{(0)}$ after
$d$ edge deletions. As $d$ increases, we maintain $G_{q}=G_{q}^{(d)}[V(G_{q}^{(0)})-P_{q}]$.
That is, $G_{q}$ is obtained from $G_{q}^{(0)}$ after deleting all
edges deleted by the adversary and deleting all vertices in $P_{q}$.
By \Cref{thm:pruning}, $G_{q}$ is always a $\gamma_{L}(\alphawit)$-expander
and $P_{q}$ has volume at most $O(\frac{dn^{1/L}}{\gamma_{L}(\alphawit)})$
after $d$ deletions. We let $P=P_{q}$ be the output set of the algorithm
for \Cref{thm:shortkquery}. This satisfies the first guarantee of
the output of \Cref{thm:shortkquery}.

Given a query $u,v\in V(G)-P$, we can return a $u$-$v$ simple path
in $G[V(G^{(0)})-P]$ of length $1/\gamma_{L}^{O(q)}(\alphawit)$
in time $1/\gamma_{L}^{O(q)}(\alphawit)$ by doing the following.
First, we call $\mathtt{Query}(q,u,v)$ and return a $u$-$v$ path
$Q'$ of length $1/\gamma_{L}^{O(q)}(\alphawit)$ in $O(|Q'|)$ time
(will be proved in \Cref{lem: oracle nonsimple path}). However, $Q'$
might not be simple. So we extract a simple $u$-$v$ path $Q$ from
$Q'$ in time $|Q'|\le1/\gamma_{L}^{O(q)}(\alphawit)$. This satisfies
the second guarantee of the output of \Cref{thm:shortkquery}.

It remains to bound the total update time in \Cref{lem: oracle update time}
and prove the guarantee about $\mathtt{Query}(q,u,v)$ in \Cref{lem: oracle nonsimple path}.

\begin{algorithm}
	\textbf{Assert:} $G_{i}$ is a $\gamma_{L}(\alphawit)$-expander. 
	\begin{enumerate}
		\item If $i=1$, compute a shortest path tree $T_{1}$ rooted at an arbitrary
		vertex. Then, return. 
		\item Build a \emph{subdivided} graph $G'_{i}$ obtained from $G_{i}$ by
		subdividing each edge $e=(u,v)\in E(G)$ is into $(u,x_{e})$ and
		$(x_{e},v)$. 
		\item Set $F_{i-1}$ to be an arbitrary set of edges in $G_{i}$ of size
		$m^{(i-1)/q}$. Let $X_{F_{i-1}}=\{x_{e}\in V(G'_{i})\mid e\in F_{i-1}\}$. 
		\item \label{enu:core embed}Using \Cref{lem:terminal witness}, compute
		a $\Omega(\gamma_{L}(\alphawit)/\log^{2}n)$-witness $G_{i-1}^{(0)}$
		in $G'_{i}$ where $V(G_{i-1}^{(0)})=X_{F_{i-1}}$ and $G_{i-1}^{(0)}$
		is a $\alphawit$-expander. Let $\pset_{i-1}$ be an embedding of
		$G_{i-1}^{(0)}$.
		\item Initialize the expander pruning algorithm from \Cref{thm:pruning}
		on $G_{i-1}^{(0)}$ and maintain $P_{i-1}\subseteq V(G_{i-1}^{(0)})$. 
		\item Let $G_{i-1}^{(d)}$ denote $G_{i-1}^{(0)}$ after $d$ edge deletions.
		As $d$ increases, maintain $G_{i-1}=G_{i-1}^{(d)}[V(G_{i-1}^{(0)})-P_{i-1}]$.
		By \Cref{thm:pruning}, $G_{i-1}$ is always a $\gamma_{L}(\alphawit)$-expander. 
		\item \label{enu:APSP root ES tree}Initialize two ES-trees $T_{i}^{in}$
		and $T_{i}^{out}$ in $G'_{i}$ rooted at $V(G_{i-1})$ of depth $O(\log(n)/\gamma_{L}(\alphawit))$.
		Edges of $T_{i}^{in}$ and $T_{i}^{out}$ are directed inwards and
		outwards $V(G_{i-1})$ respectively.
		\item Call $\mathtt{GrowTree}(i-1)$. 
	\end{enumerate}
	\caption{$\mathtt{GrowTree}(i)$\label{alg:core init}}
\end{algorithm}

\begin{algorithm}
	\begin{enumerate}
		\item If $i=1$, delete $e$ in $G_{1}$. Recompute a shortest path tree
		$T_{1}$ in $G_{1}$. Then, return. 
		\item Delete $e$ from $G_{i}$. Update the vertex set $P_{i}$ using \Cref{thm:pruning}. 
		\item Let $D_{i}^{new}$ denote the set of edges that are just removed from
		$G_{i}$. That is, $D_{i}^{new}$ contains $e$ and all edges incident
		to vertices that are newly added into $P_{i}$. 
		\item For each $e\in D_{i}^{new}$, 
		\begin{enumerate}
			\item Let $\pset_{i-1}^{(e)}$ be a set of paths from the embedding $\pset_{i-1}$
			of $G_{i-1}$ that contains $e$. Let $D_{i-1}^{(e)}$ be the set
			of edges in $E(G_{i-1})$ corresponds to $\pset_{i-1}^{(e)}$.
			\item $\mathtt{Delete}(i-1,e')$ for each $e'\in D_{i-1}^{(e)}$. 
		\end{enumerate}
		\item Whenever there are more than $d_{i-1}=\gamma_{L}(\alphawit)\vol(G_{i-1}^{(0)})/n^{1/L}$
		deletions to $G_{i-1}^{(0)}$, call $\mathtt{GrowTree}(i)$. 
	\end{enumerate}
	\caption{$\mathtt{Delete}(i,e)$ where $e\in E(G_{i})$\label{alg:core delete}}
\end{algorithm}

\begin{algorithm}
	\begin{enumerate}
		\item If $i=1$, return a $u$-$v$ path by traversing $T_{1}^{in}$ and
		$T_{1}^{out}$. 
		\item Let $Q_{uu'}$ be the path in $T_{i}^{in}$ from $u$ to $u'\in V(G_{i-1})$.
		Let $Q_{v'v}$ be the path in $T_{i}^{out}$ from $v'\in V(G_{i-1})$
		to $v$. 
		\item Let $R_{u'v'}=\mathtt{Query}(i-1,u',v')$ be the returned $u'$-$v'$
		path in $G_{i-1}$. 
		\item Let $Q_{u'v'}$ be obtained by concatenating, over all $e'\in R_{u'v'}$,
		the corresponding paths from the embedding $\pset_{i-1}$ of $G_{i-1}$.
		\item Return the concatenation $Q_{uv}=Q_{uu'}\circ Q_{u'v'}\circ Q_{v'v}$
		as a path in $G_{i}$. (Note that $Q_{uu'},Q_{u'v'},Q_{v'v}$ are,
		strictly speaking, paths in $G'_{i}$.)
	\end{enumerate}
	\caption{$\mathtt{Query}(i,u,v)$ where $u,v\in V(G_{i})$\label{alg:core query}}
\end{algorithm}
\begin{lem}
	\label{lem: oracle update time} The total update time is $O(m^{1+2/q}/\gamma_{L}^{O(q)}(\alphawit))$. 
\end{lem}

\begin{proof}
	Let $\mathrm{Time}(i)$ be the total update time that the data structure
	at level $i$ takes for handling $d_{i}=\gamma_{L}(\alphawit)\vol(G_{i}^{(0)})/n^{1/L}$
	edge deletions in $G_{i}^{(0)}$. So $\mathrm{Time}(q)$ is the total
	update time of our algorithm. For each level $i$, throughout $d_{i}$
	edge deletions in $G_{i}^{(0)}$, the total volume of edges pruned
	out by \Cref{thm:pruning} is $O(d_{i}n^{1/L}/\gamma_{L}(\alphawit))\le\vol(G_{i}^{(0)})/2$
	(by scaling $d_{i}$ by some constant). As the embedding of $G_{i-1}$
	in $G_{i}$ has congestion at most $\congest=\Otil(1/\gamma_{L}(\alphawit))$,
	this corresponds to at most $\congest\cdot\vol(G_{i}^{(0)})$ edge
	deletions to $G_{i-1}$. As we call $\mathtt{GrowTree}(i)$ only when
	there are more than $d_{i-1}=\gamma_{L}(\alphawit)\vol(G_{i-1}^{(0)})/n^{1/L}$
	deletions to $G_{i-1}^{(0)}$, the number of calls to $\mathtt{GrowTree}(i)$
	throughout $d_{i}$ deletions is at most 
	\[
	\frac{\congest\cdot\vol(G_{i}^{(0)})}{\gamma_{L}(\alphawit)\vol(G_{i-1}^{(0)})/n^{1/L}}=\tilde{O}(m^{1/q}n^{1/L}/\gamma_{L}^{2}(\alphawit)).
	\]
	where we use the fact that $|V(G_{i})|=m^{i/q}$ and $\vol(G_{i})=\tilde{O}(|V(G_{i})|)$
	by \Cref{lem:terminal witness}.
	
	Consider the total work for executing $\mathtt{GrowTree}(i)$ and
	maintaining the data structure until right before the next call of
	$\mathtt{GrowTree}(i)$. We divide the work into two parts. First,
	the work for executing $\mathtt{GrowTree}(i)$ itself (which embeds
	$G_{i-1}^{(0)}$ into $G_{i}$). Second, the work for maintaining
	the data structure at level $i-1$ throughout $d_{i-1}$ deletions
	to $G_{i-1}^{(0)}$. The second part takes at most $\mathrm{Time}(i-1)$
	by definition.
	
	Now, we analyze the first part, the work for executing $\mathtt{GrowTree}(i)$.
	Consider \Cref{alg:core init}. Embedding $G_{i-1}$ into $G_{i}$
	takes time $O(\vol(G_{i})/(\alphawit\cdot\gamma_{L}(\alphawit))$
	by \Cref{lem:terminal witness}. \Cref{thm:pruning} takes $\tilde{O}(\frac{\vol(G_{i})n^{1/L}}{\gamma_{L}(\alphawit)})$.
	The total time for maintaining the ES-tree $T_{i}$ is also $\tilde{O}(\vol(G_{i})/\gamma_{L}(\alphawit))$.
	So each call to $\mathtt{GrowTree}(i)$ takes at most $\tilde{O}(m^{i/q}n^{1/L}/\gamma_{L}^{2}(\alphawit))$
	time. Therefore, we have 
	\[
	\mathrm{Time}(i)=\left(\tilde{O}(m^{i/q}n^{1/L}/\gamma_{L}^{2}(\alphawit))+\mathrm{Time}(i-1)\right)\times\tilde{O}(m^{1/q}n^{1/L}/\gamma_{L}^{2}(\alphawit)).
	\]
	Solving this recursion, we have $\mathrm{Time}(i)=O(m^{(i+1)/q}n^{i/L}\log^{O(i)}(m)/\gamma_{L}^{2i}(\alphawit))$.
	So 
	\begin{align*}
	\mathrm{Time}(q) & =O(m^{1+1/q}n^{q/L}\log^{O(q)}(m)/\gamma_{L}^{O(q)}(\alphawit))\\
	& =O(m^{1+2/q}/\gamma_{L}^{O(q)}(\alphawit))
	\end{align*}
	because $L=q^{2}$ and $\gamma_{L}(\alphawit)\ll1/\log^{O(1)}m$
	as desired. 
\end{proof}
\begin{lem}
	\label{lem: oracle nonsimple path} Given any pair of vertices $u,v\in V(G)-P$,
	$\mathtt{Query}(q,u,v)$ returns a $u$-$v$ (possibly non-simple)
	path $Q$ of length $\log^{O(q)}n$ in $O(|Q|)$ time. 
\end{lem}

\begin{proof}
	Let $\mathrm{Len}(i)$ be the maximum length of the path in $G_{i}$
	returned by $\mathtt{Query}(i,u,v)$. As $G_{i}$ is always a $\gamma_{L}(\alphawit)$-expander
	by \Cref{thm:pruning}, we have that the diameter of $G_{i}$ is $O(\log(n)/\gamma_{L}(\alphawit))$.
	So $T_{i}^{in}$ and $T_{i}^{out}$ span $G'_{i}$. Consider \Cref{alg:core query}.
	Let $Q_{uu'}$ be the path in $G'_{i}$ from $u$ to $u'\in V(G_{i-1})$
	and let $Q_{v'v}$ the path in $G'_{i}$ from $v'\in V(G_{i-1})$
	to $v$. As $T_{i}^{in}$ and $T_{i}^{out}$ span $G'_{i}$, $Q_{uu'}$
	and $Q_{v'v}$ do exist. Let $R_{u'v'}=\mathtt{Query}(i-1,u',v')$
	where $|R_{u'v'}|\le\mathrm{Len}(i-1)$. Let $Q_{u'v'}$ be obtained
	by concatenating, over all $e'\in R_{u'v'}$, the corresponding paths
	from the embedding $\pset_{i-1}$ of $G_{i-1}$. We have $|Q_{u'v'}|\le\ell\cdot|R_{u'v'}|$.
	It is clear that the concatenation $Q_{uu'}\circ Q_{u'v'}\circ Q_{v'v}$
	is indeed a $u$-$v$ path in $G'_{i}$ and hence in $G_{i}$. The
	length of this path is at most 
	\[
	\mathrm{Len}(i)=O(\log(n)/\gamma_{L}(\alphawit))+\tilde{O}(1/\gamma_{L}(\alphawit))\cdot\mathrm{Len}(i-1).
	\]
	Solving the recursion gives us $\mathrm{Len}(i)=1/\gamma_{L}^{O(i)}(\alphawit)$.
	So $\mathtt{Query}(q,u,v)$ returns a (possibly non-simple) $u$-$v$
	path of length $1/\gamma_{L}^{O(q)}(\alphawit)$. Observe that the
	query time is proportional to the returned path. 
\end{proof}

\end{document}